\documentclass[11pt]{article}
\usepackage{preamble}
\usepackage[margin=1in]{geometry}
\title{Self-replication and Computational Universality}
\mathtoolsset{showonlyrefs}

\usepackage{defs}

\author{
Jordan Cotler\thanks{Email: \texttt{jcotler@fas.harvard.edu}.} \\
Harvard University
\and
Cl\'{e}ment Hongler\thanks{Email: \texttt{clement.hongler@epfl.ch}.} \\
EPFL
\and
Barbora Hudcov\'{a}\thanks{Email: \texttt{barbora.hudcova@epfl.ch}.} \\
EPFL
}

\usepackage[titles]{tocloft}

\begin{document}
\pagestyle{empty}
{
  \renewcommand{\thispagestyle}[1]{}
  \maketitle

\begin{abstract}
Self-replication is central to all life, and yet how it dynamically emerges in physical, non-equilibrium systems remains poorly understood.  Von Neumann's pioneering work in the 1940s and subsequent developments suggest a natural hypothesis: that any physical system capable of Turing-universal computation can support self-replicating objects. In this work, we challenge this hypothesis by clarifying what computational universality means for physical systems and constructing a cellular automaton that is Turing-universal but cannot sustain non-trivial self-replication. By analogy with biology, such dynamics manifest transcription and translation but cannot instantiate replication.  More broadly, our work emphasizes that the computational complexity of translating between physical dynamics and symbolic computation is inseparable from any claim of universality (exemplified by our analysis of Rule 110) and builds mathematical foundations for identifying self-replicating behavior. Our approach enables the formulation of necessary dynamical and computational conditions for a physical system to constitute a living organism.
\end{abstract}

}

\clearpage
\pagestyle{plain}
\pagenumbering{arabic}

\tableofcontents

\section{Introduction}
\label{sec:intro}

The emergence of self-replicating structures in Earth's oceans and their evolution into the diversity of life we witness today remains one of science's greatest mysteries. This gap in understanding prevents us from recreating the process through simulation or experiment, raising fundamental philosophical and mathematical questions about life's nature. In the 1940s, John von Neumann \cite{neumann_vnsr} illuminated self-replication's key principles using an abstract framework known as cellular automata (CAs)~\cite{kari2022cellular}.  As illustrated in Fig.~\ref{fig:main1}(a), physical or chemical systems can be abstracted as CAs that, under the right conditions, dynamically produce self-replicating organisms. Our ultimate goal is to identify which dynamics and configurations permit the emergence of self-replication.  Von Neumann's pioneering work provides a foundation: he designed a CA with a non-trivial self-replicator leveraging the dual use of an organism's description -- once to be interpreted for construction and once to be transcribed to the new offspring. Remarkably, von Neumann's model distinguished between replication and transcription before DNA's structure was discovered \cite{brenner2012life}.

Von Neumann's work highlights a key distinction between genuine self-replication and superficial copying. Consider red dye diffusing in water: though color spreads throughout, we would not call dye a `self-replicating organism.' To exclude such examples, von Neumann designed a universal self-replicator: a finite CA pattern that performs arbitrary computations on its input tape and constructs arbitrary patterns from tape descriptions. When reading its own description, the machine builds an exact copy of itself. Thus, the underlying CA rules must be `Turing-complete', capable of arbitrary computation. Crucially, self-replication emerges from the machine's computational universality.

There have been many simplifications of von Neumann's construction \cite{langton1984self, byl1989self, sayama1999toward, sipper1998fifty}, including most famously Conway's Game of Life \cite{gardner1970games} which features deceptively simple rules generating viscerally lifelike structures \cite{johnston2022conway}. Game of Life was believed to contain self-replicators before it was formally proven, due to its Turing completeness \cite{berlekamp2004winning, turing_universality_of_gol}. Wolfram \cite{wolfram1984universality} later emphasized that Turing-complete CAs exhibit distinctive qualitative features, prompting his conjecture that Rule 110 -- a one-dimensional automaton with an extremely compact description -- was Turing-universal based on these characteristics, subsequently proven through elaborate analysis of the rule's dynamics \cite{cook_og, p_completeness_of_rule110, cook2009concrete}.

However, a fundamental ambiguity persists in these core notions \cite{universalities_in_cas, mirage_of_universality_in_cas}. Turing constructed a universal Turing machine without formally specifying the notion of Turing completeness \cite{turing1936computable}. Similarly, von Neumann constructed a non-trivial self-replicator without formally characterizing non-trivial self-replication. A natural hypothesis is that Turing universality implies self-replication \cite{langton1984self} -- after all, such systems support quines, programs outputting their own description \cite{jones1997computability, hofstadter1999godel}. However, as we show, nuances about how physical systems perform computation can matter profoundly.

We formalize two kinds of Turing universality for CAs: local universality -- an automaton's ability to simulate any Turing machine (implicitly assumed by von Neumann and Wolfram), and global universality~\cite{simple_universal_ca} -- the ability to simulate any same-dimension CA. Local universality focuses on computational power within the CA \cite{lindgren1990universal}, while global universality concerns emulating the entirety of other physical systems \cite{universalities_in_cas}.

We examine the relationships between local and global simulation and establish their precise connections with self-replication. Though most natural examples like Game of Life exhibit both local and global universality \cite{the_game_of_life_universality_revisited}, the distinction between these two forms of universality becomes fundamental when it comes to self-replication. Our analysis addresses several key questions: whether local and global universality are equivalent (they are not, global is stronger), whether local universality implies self-replication (it does not), whether global universality implies self-replication (it does), whether Rule 110 supports self-replication (unknown), and whether dimensionality affects these relationships, summarized in Fig.~\ref{fig:main1}(b).

\begin{figure}[t!]
    \centering
    \includegraphics[width=0.95\linewidth]{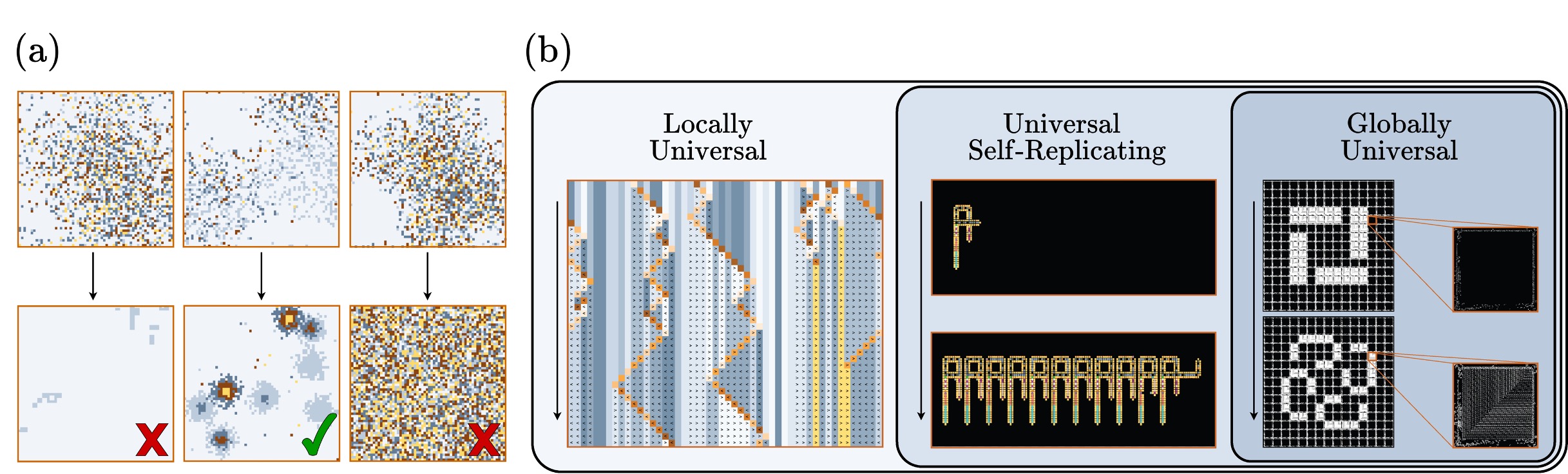}
    \vspace{.2cm}
    \caption{(a) Physical or chemical systems, abstracted as cellular automata (CAs), may in appropriate circumstances evolve from random initial conditions into self-replicating organisms. We develop a theory constraining which dynamics and configurations can give rise to self-replicating organisms. (b) We rigorously define what it means for a CA to be Turing-universal and establish its relationship to self-replication. We identify and study three classes of CAs with increasing capabilities: (i) locally universal CAs can host a universal Turing machine within their dynamics, (ii) universal self-replicating CAs support self-replicating structures that also perform Turing-universal computation, and (iii) globally universal CAs can simulate any other CA of the same dimension via local, translation-invariant encodings. We prove these form a strict hierarchy: Locally Universal $\supset$ Universal Self-Replicating $\supset$ Globally Universal. Consequently, some CAs (such as the one shown in the `locally universal' block) are Turing-universal and yet provably cannot support self-replicating organisms, analogous to biological systems capable of transcription but not replication.}
    \label{fig:main1}
\end{figure}

Our work establishes foundations for a mathematical theory of self-replication and its connections to computability and complexity theory. By disentangling self-replication's interface with computational universality, we identify mechanisms essential for replication vis-à-vis transcription. Our definitions and results enable the specification of precise, necessary conditions for physical systems to constitute self-replicating organisms.

\section{Results}
\label{sec:results}

Section~\ref{sec:univ} defines local and global universality, demonstrating the latter is strictly stronger via reversible CAs \cite{kari2005reversible}. We verify that our definitions align with previous CA constructions, including Rule 110's elaborate Turing completeness proof. Section~\ref{sec:necess} presents necessary conditions for non-trivial self-replication and modifies von Neumann's construction to demonstrate 1D universal self-replicators. Section~\ref{sec:selfrep} constructs a 1D locally universal CA that cannot sustain non-trivial self-replication, as calibrated by our necessary conditions.

\subsection{Computational universality in physical dynamics}
\label{sec:univ}

We use CAs to model spatially local, translation-invariant physical dynamics. A CA consists of a regular grid where each cell can be in one of finitely many states, evolving synchronously according to local rules that depend only on each cell's current state and its neighbors. Rule 110 \cite{wolfram1984universality} is a famous 1D example with black and white cells where each cell's next state depends on its current state and nearest neighbors, as shown in Fig.~\ref{fig:main2}(a).

We examine the relationship between dynamics and computation: in what sense does physical dynamics perform computation? This question, originating with Turing \cite{turing1936computable}, has a rich history \cite{piccinini2015physical}. The Church-Turing Thesis stipulates that a universal Turing machine can simulate any physical system's dynamics, including another Turing machine, but is vague about what constitutes simulation. When you run Python code, symbolic manipulation is encoded into electron configurations in semiconductors, evolving according to physical laws, and then decoded back to symbolic output. This illustrates the necessity of `encoding' and `decoding' operations mediating between symbolic computation and physical dynamics (Fig.~\ref{fig:main2}(b)). Physical dynamics instantiate specific symbolic computation if an encoder-decoder pair exists relating them, subject to sensible constraints. Crucially, the encoder and decoder must be computationally simpler than the computation attributed to the physical system. Otherwise, we risk computational misattribution where interesting computation is performed by the encoding/decoding apparatus rather than the physical dynamics.

\begin{figure}
    \centering
    \includegraphics[width=0.95\linewidth]{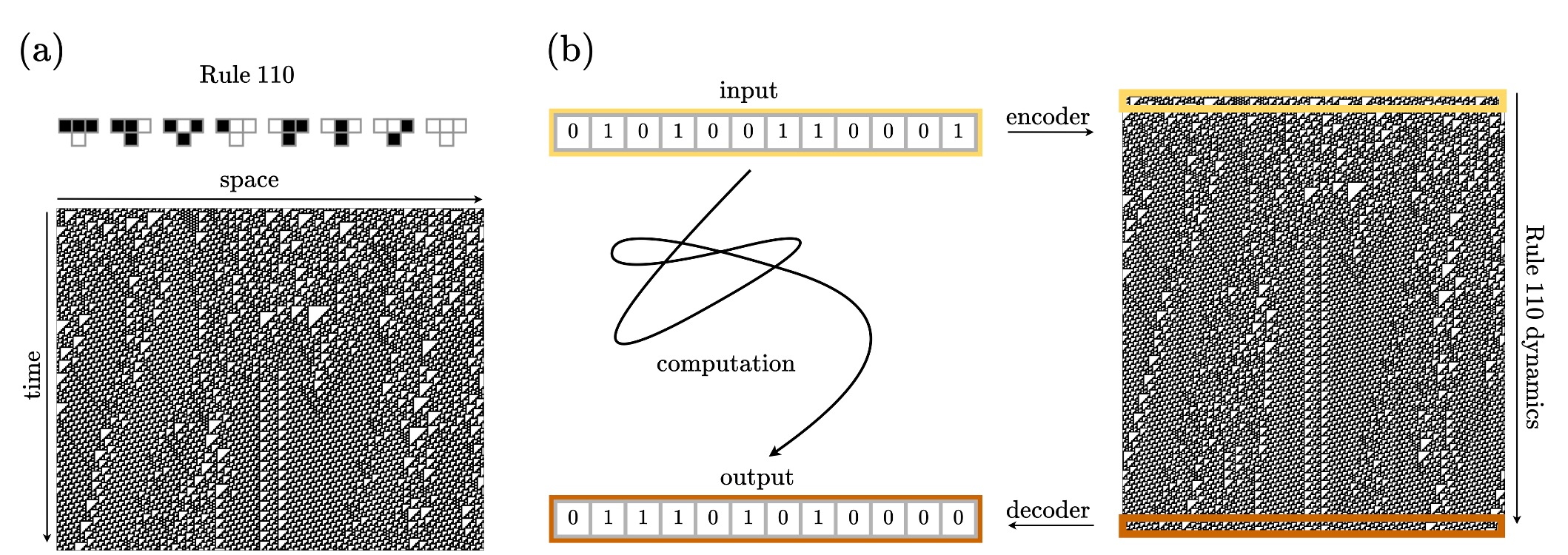}
    \vspace{0.2cm}
    \caption{(a) Space-time diagram of the elementary CA, Rule 110. Each cell assumes a binary state (black or white) and evolves deterministically based on its current state and those of its two nearest neighbors. Time progresses downward from a random initial configuration (top row), revealing the emergence of complex spatiotemporal structures. (b) Schematic of a computational process, as implemented by a Turing machine, where an input bit string undergoes computation to produce an output bit string. To simulate this computation using Rule 110, three requirements must be satisfied: (i) the input bit string must be encoded into the CA's initial configuration, (ii) the CA's evolution must faithfully track the Turing machine's computational steps, and (iii) a decoding scheme must extract the final output from the CA dynamics such that it matches the Turing machine's output.}
    \label{fig:main2}
\end{figure}

We define two forms of CA Turing-universality, starting with:
\begin{definition}[Local universality of a CA, informal]
A CA is \text{\rm locally universal} if it can simulate a universal Turing machine in a bounded region whose size scales with the machine's tape usage. Specifically, there exists an encoder-decoder pair that maps between universal Turing machine dynamics and CA dynamics, where both the encoder and decoder have time complexity bounded by a function of the size of the non-blank portion of the Turing machine tape.
\end{definition}
\noindent The key idea is that a physical system is Turing-universal if it can host a localized universal computer that expands as memory requirements grow (Fig.~\ref{fig:main3}(a)). Additional technical details on time slowdowns, symmetry transformations, and background configurations appear in the Appendix.

\begin{definition}[Global universality of a CA, informal]
A CA is \text{\rm globally universal} if it can simulate any other CA of the same or lower dimension. The encoder and decoder must be block encodings and thus have bounded time complexity for encoding/decoding between local regions.
\end{definition}
\noindent This captures a fundamentally different notion: a globally universal physical system can simulate any other physical system while preserving spatial locality, i.e.~local interactions in the simulated system correspond to local interactions in the simulating system (see Fig.~\ref{fig:main3}(b)). This is stringent because CAs are infinite in extent; we are characterizing how one infinite system can simulate the entirety of another. Global universality is a form of `intrinsic universality' in the CA literature, formulated in several variants~\cite{simple_universal_ca, inducing_order_on_ca_by_grouping}. By contrast, despite local universality's foundational importance to von Neumann's and Wolfram's work, the field has lacked a rigorous definition. Our framework fills this gap.

\begin{figure}
    \centering
    \includegraphics[width=0.95\linewidth]{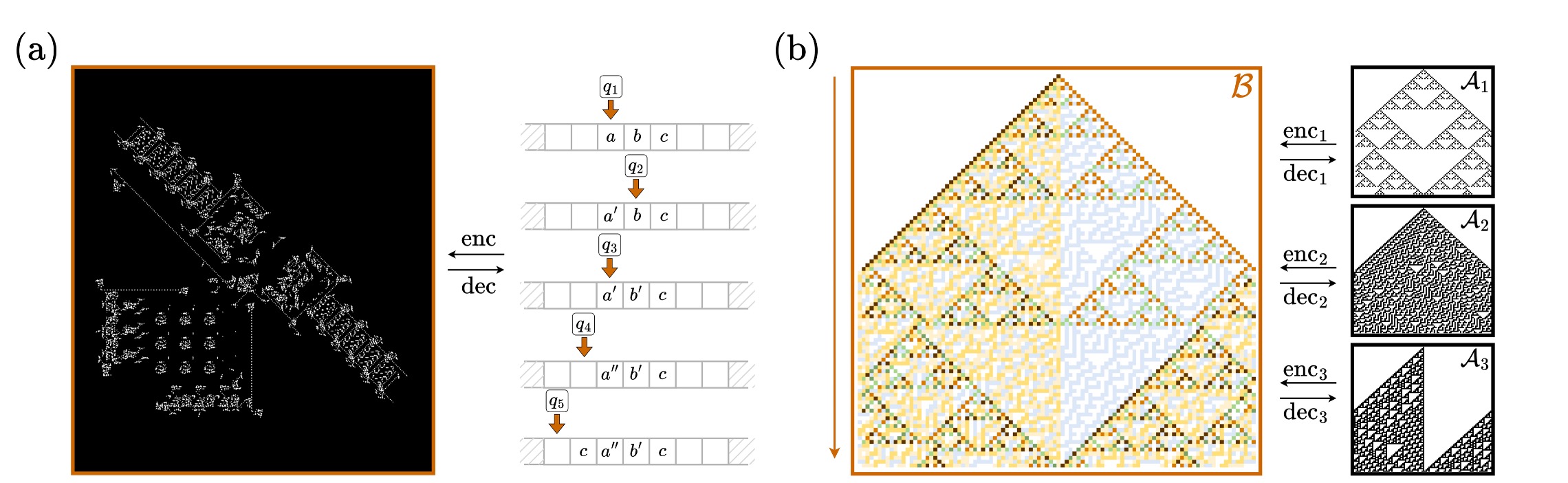}
    \vspace{0.2cm}
    \caption{(a) Local universality in a 2D CA. A spatially finite configuration (shown at a fixed time) can perform Turing-universal computation: any Turing machine can be encoded into such a configuration, with the CA's evolution tracking the machine's computation and its states recoverable through decoding. The CA is termed \textit{locally universal} because Turing machines with finite input can be simulated within finite (though potentially growing) spatial regions, rather than requiring infinite configurations of the CA at the outset.  (b) Global universality in a 1D CA. The CA $\mathcal{B}$ (shown evolving top to bottom) can simulate infinite configurations of any other 1D CA through local, translation-invariant encoding and decoding maps.  By contrast, local universality only requires a CA to be able to simulate a Turing machine within a finite region. Here we exemplify how $\mathcal{B}$ globally simulates dynamics in three different CAs $\mathcal{A}_1, \mathcal{A}_2, \mathcal{A}_3$ simultaneously.}
    \label{fig:main3}
\end{figure}

What, then, is the relationship between local and global universality? Let $\textsf{LocallyUniversal}$ denote the set of locally universal CAs and $\textsf{GloballyUniversal}$ denote the set of globally universal CAs. We have:
\begin{theorem}[Computational universality hierarchy for CAs] $\textnormal{\textsf{GloballyUniversal}} \subsetneq \textnormal{\textsf{LocallyUniversal}}$.
\end{theorem}
\noindent This reveals a hierarchy of computational power. Every globally universal CA must be locally universal: if a CA can simulate any CA, it can simulate one implementing a Turing machine locally. The proof verifies that global simulation of a locally universal CA yields a valid local encoder-decoder pair for the simulating CA (see Appendix). The containment is strict: some CAs are locally universal but not globally universal. Consider reversible CAs, which run equally well forwards or backwards. Some can locally implement reversible universal Turing machines but cannot globally simulate irreversible CAs \cite{hertling1998embedding}, proving $\textsf{GloballyUniversal} \subsetneq \textsf{LocallyUniversal}$.

The literature contains numerous CAs exhibiting either local or global universality \cite{simple_universal_ca, morita1995, intrinsic_universality_of_reversible_1D_ca, turing_universality_of_gol, cook_og}. Much effort has focused on constructing CAs with simple rules that locally implement universal Turing machines \cite{cook_og}. However, these constructions often leave encoder-decoder pairs implicit and neglect their computational complexity.  This oversight can be critical: without proper constraints, complex encoders and decoders can spuriously perform the computation we want to attribute to the CA itself \cite{the_game_of_life_universality_revisited, computational_dynamical_systems}. The elementary Rule 110 provides a compelling case study. This 1D CA with nearest-neighbor interactions simulates a universal Turing machine through a remarkable construction, yet the encoder-decoder computational complexity was never analyzed. Building on previous work~\cite{cook_og, p_completeness_of_rule110, cook2009concrete}, we establish:
\begin{theorem}[Local universality of Rule 110]
Rule 110 is locally universal with respect to an encoder-decoder pair with polynomial time complexity and linear space complexity. Rule 110 incurs a polynomial time slowdown relative to the universal Turing machine it simulates.
\end{theorem}
\noindent Our analysis traces a chain of simulations: universal Turing machine $\leftrightarrow$ clockwise Turing machine $\leftrightarrow$ cyclic tag system $\leftrightarrow$ Rule 110 (Fig.~\ref{fig:main4}(a)). By analyzing the computational complexity at each step and exploiting transformation compositionality, we bound the overall encoding/decoding complexity. This rigorous analysis necessitates making all encodings and decodings fully explicit (Fig.~\ref{fig:main4}(b)), providing a complete characterization of Rule 110's computational universality.

\begin{figure}
    \centering
    \includegraphics[width=0.95\linewidth]{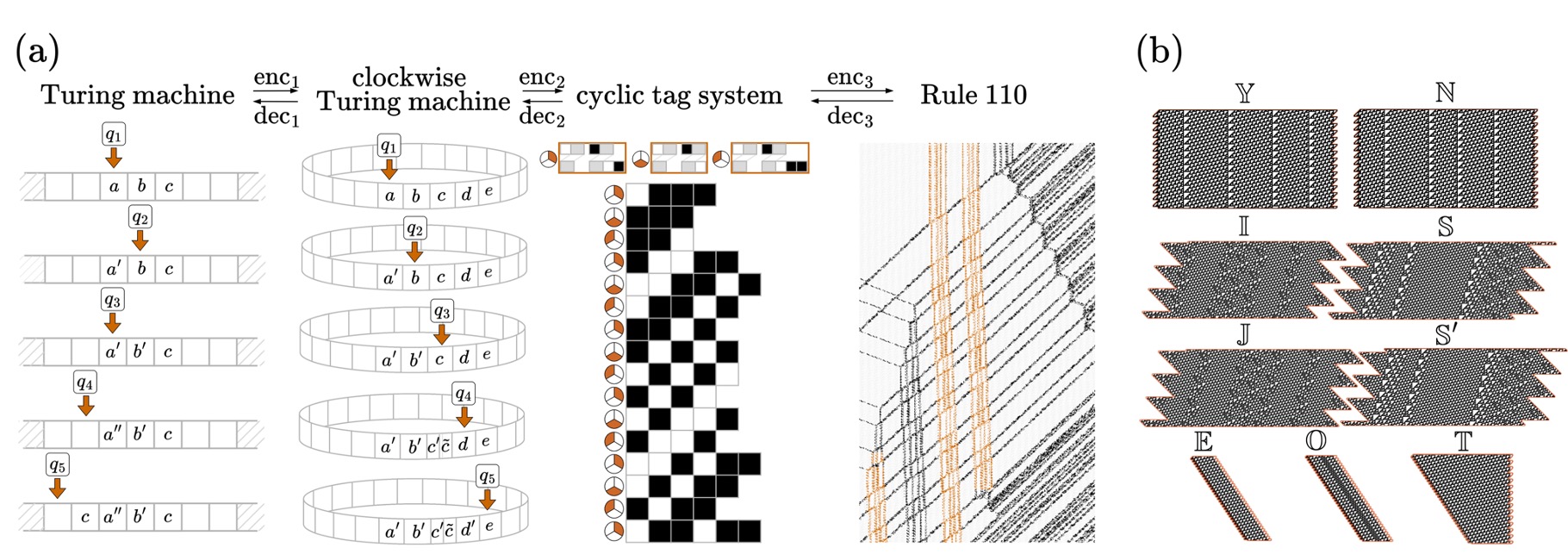}
    \vspace{0.2cm}
    \caption{(a) Rule 110 can simulate any Turing machine via a sequence of intermediate encodings and decodings.  In particular, Rule 110 (where certain ``gliders'' are colored orange for emphasis) simulates a cyclic tag system which itself simulates a clockwise Turing machine, which itself simulates an ordinary Turing machine.  We carefully keep track of the computational complexities of the encodings and decodings as well as their compositionality properties, and establish that the composed encoding and decoding between Rule 110 and a Turing machine requires only polynomial time complexity and linear space complexity. (b) The encoding from a Turing machine to Rule 110 and the subsequent decoding are very intricate. Shown here are encoding components of Turing machine states mapping to modular regions of the Rule 110 CA that fit together like puzzle pieces.
}
    \label{fig:main4}
\end{figure}

\subsection{Necessary conditions for self-replication}
\label{sec:necess}

How does self-replication interface with computational universality? What logical organization is necessary or sufficient for a finite CA pattern to self-replicate? Von Neumann addressed sufficiency by constructing a universal self-replicator in a 29-state, 2D CA. His work and its refinements \cite{neumann_vnsr, thatcher_construction, codd1968cellular} demonstrated self-replication in CAs.  Here, we focus on necessary conditions for CAs to support self-replicating configurations. A natural question from von Neumann's construction is whether two dimensions are necessary for universal self-replication. We establish:
\begin{theorem}
\label{thm:1dvN}
Universal self-replicators exist in one dimension.
\end{theorem}
\noindent We modify von Neumann's 2D design so that its `constructing arm' builds replicas with a horizontal rather than diagonal offset, allowing operation within a bounded-height strip. Encoding this bounded-height strip into a 1D CA with an expanded state space (more than 29 states) yields a 1D CA implementing universal self-replication.

Von Neumann's Turing-universal constructor ensures self-replicating subsystems perform dynamics beyond mere copying.  However, von Neumann did not articulate necessary conditions for CAs to support universal self-replicators; he provided an example but did not define the general set of CAs \textsf{UniversalSelfReplicating}. We define a necessary condition for universal self-replication from CA dynamics, phrased informally below.
\begin{definition}[Local self-replicating condition, informal]
\label{def:locallyperiodicself1}
A CA satisfies the \text{\rm local self-replicating} \text{\rm condition} if it contains a bounded pattern that evolves for at least time $T \geq 2$ and proliferates arbitrarily many copies, where each copy exhibits the same bounded dynamical patterns.
\end{definition}
\noindent Requiring $T \geq 2$ ensures dynamic rather than static patterns. This captures non-trivial self-replication as unbounded proliferation of dynamical patterns.

The local self-replicating condition represents a minimal requirement for non-trivial self-replication. Therefore:
\begin{conditions}[Necessary conditions for universal self-replicators, informal] \label{def:necessaryCA1}
Any CA in \textnormal{\textsf{UniversalSelfReplicating}} is locally universal and satisfies the local self-replicating condition.
\end{conditions}
\noindent These requirements reflect that universal self-replication must be non-trivial (containing dynamical components) and computationally universal. While additional criteria are needed to fully characterize \textsf{UniversalSelfReplicating}, these conditions suffice for our analysis. Von Neumann's universal constructor exemplifies both conditions.

There exist simpler examples of self-replicating CA subsystems that are not themselves Turing-universal, and yet still exhibit non-trivial dynamics. Simpler examples like Langton's loops \cite{langton1984self} through the Turing-universal Perrier-Sipper-Zahnd loop \cite{perrier1996toward} all satisfy our local self-replicating condition, suggesting a hierarchy ordered by dynamical complexity. Such observations motivate the development of more stringent conditions that capture increasingly sophisticated dynamics beyond mere periodicity, potentially enabling a systematic classification scheme for self-replicating systems based on their computational and dynamical properties.

\subsection{Self-replication versus computational universality}
\label{sec:selfrep}

Does local universality, namely the ability to build a local Turing-universal computer within a CA, guarantee non-trivial self-replicating configurations? Surprisingly, we prove that the answer is no. We construct a locally universal CA that cannot support self-replication, demonstrating a separation between computational power and reproductive capability.  Our CA achieves universality through isolated computational elements whose restricted communication precludes self-replication.

We first construct a 1D, Turing-universal `talking heads' CA, with mobile Turing machine heads that communicate through proximity or shared tape. We then modify this to a `non-talking heads' CA, enforcing that proximal heads halt immediately. We further prevent tape-mediated communication by appending directional markers to the state space: when a head moves rightward from a cell it marks that cell with a $\rightarrow$, and when moving leftward it marks with a $\leftarrow$. These markers establish exclusive territories since any head encountering a cell marked by another head halts immediately. This blocking of information transfer prevents self-replication:
\begin{theorem}[Non-talking heads CA does not have self-replication]\label{thm:nontalking}
A non-talking heads CA does not satisfy the local self-replicating condition and therefore cannot belong to $\textnormal{\textsf{UniversalSelfReplicating}}$.
\end{theorem}
\noindent This shows $\textsf{UniversalSelfReplicating}$ and $\textsf{LocallyUniversal}$ to be distinct, as the non-talking heads CA belongs to the latter but not the former. More broadly:
\begin{theorem}[Self-replication lies strictly between global and local universality]\label{thm:selfhierarchy}
The following strict inclusions hold: $\textnormal{\textsf{GloballyUniversal}} \subsetneq \textnormal{\textsf{UniversalSelfReplicating}} \subsetneq \textnormal{\textsf{LocallyUniversal}}$.
\end{theorem}
\noindent This hierarchy (Fig.~\ref{fig:main1}(b)) reveals fundamental relationships between computational power and self-replication. For the proof, $\textnormal{\textsf{GloballyUniversal}} \subsetneq \textnormal{\textsf{UniversalSelfReplicating}}$ holds since globally universal CAs can simulate any CA preserving spatial locality. The inclusion is strict because reversible CAs support universal self-replication yet cannot globally simulate irreversible dynamics (see Appendix). The inclusion $\textnormal{\textsf{UniversalSelfReplicating}} \subsetneq \textnormal{\textsf{LocallyUniversal}}$ follows from Condition~\ref{def:necessaryCA1}, with strict inclusion from Theorem~\ref{thm:nontalking}.

Our results clarify that local computational universality is necessary but insufficient for non-trivial self-replication. The non-talking heads construction abstracts biological concepts of transcription and replication, showing systems can possess transcriptional capabilities without achieving true replication. As alluded to earlier, our result may seem counterintuitive given the existence of quines, which are programs in any Turing-complete language outputting their own source code \cite{hofstadter1999godel, sipser1996introduction}. However, a quine produces a description of the copying mechanism without replicating the mechanism itself. In our framework, this is the difference between a head copying tape regions versus creating new heads or recruiting existing heads. Quines represent information copying; genuine self-replication requires reproducing the physical copier itself, explaining why computational universality cannot guarantee self-replication.

\section{Discussion}
\label{sec:outlook}

Our work clarifies the relationship between computational universality and self-replication in physical systems. A key contribution is our rigorous treatment of encoder-decoder complexity in defining local universality, generalizing the work of~\cite{computational_dynamical_systems}. This is exemplified by establishing polynomial-time encoding/decoding bounds for Rule 110. The complexity constraint prevents computational misattribution, ensuring computation is genuinely performed by CA dynamics rather than sophisticated encoding/decoding schemes. We demonstrated that local Turing-universal computation does not necessarily enable self-replication, a counterintuitive finding since universal computation might seem to enable all reproductive capabilities. Our non-talking heads CA abstracts the biological principle that systems capable of transcription may lack coordinated functionality for replication. At a technical level, identifying `head' and `tape' structures shows that inter-head communication is crucial for replication.

Many challenges remain. Determining whether specific CAs support universal self-replication is open, and Rule 110 presents a compelling test case:
\begin{conjecture}
Rule 110 is in $\textnormal{\textsf{UniversalSelfReplicating}}$.
\end{conjecture}
\noindent Beyond specific examples, sharpening our definition of self-replicating systems is essential, as our local self-replicating condition is only a minimal requirement. Developing increasingly stringent conditions would clarify which dynamics enable sophisticated reproduction. These may depend on CA dimensionality, though our 1D universal self-replicator demonstrates that dimension does not limit self-replication capacity. Further incorporating environmental noise is essential for understanding robustness of self-replication under realistic conditions.

Our work establishes a theoretical foundation for understanding self-replication beyond isolated examples, presenting the first general results on the relationship between computational universality and self-replication. These advances represent initial steps toward a comprehensive theory that will require integrating insights from non-equilibrium physics, chemistry, and computer science.  More broadly, our work lays the ground for formal explorations of abstract mechanisms that lead to the emergence of life.

\subsection*{Acknowledgments}

JC would like to thank Bert Chan, Semon Rezchikov, and Jensen Suther for valuable discussions.  JC is supported by a fellowship from the Alfred P.~Sloan Foundation. BH and CH would like to thank Vassilis Papadopoulos, João Penedones, Jakub Krásenský, and Jiří Tůma. The $\mathrm{UniversalSelfReplicating}$ and $\mathrm{GloballyUniversal}$ images in Fig.~\ref{fig:main1}(b) and the Game of Life Turing machine construction in Fig.~\ref{fig:main3}(a) were generated using an open-source application Golly \cite{golly}.

\appendix

\section{Related work}\label{appsec_related_work}

\subsection{Self-replicators in cellular automata}
The foundational theory of self-replication in CAs was developed by John von Neumann, who designed a universal self-replicator in a 2D automaton. His construction, published posthumously by Burks \cite{neumann_vnsr} and later completed in detail by Thatcher \cite{thatcher_construction}, is estimated to require between 50,000 and 200,000 cells. Subsequent work, termed the `second generation of self-replicating automata research', aimed at reducing the complexity of von Neumann’s model while preserving its computational universality. Codd reduced the number of CA states from 29 to 8 \cite{codd1968cellular}, while Arbib simplified aspects of the construction logic at the expense of introducing a more elaborate transition function \cite{arbib1966simple}. 

The third generation began with Langton, who argued that Turing universality is a sufficient condition for non-trivial replication, but not a necessary one. He constructed 2D loops \cite{langton1984self} that are not computationally universal but implement self-replication via a mechanism analogous to von Neumann's construction, through the transcription of the replicator's description. As such, Langton's loops are deemed non-trivial. Unlike von Neumann’s description of a large scale self-replicator, Langton’s loops are small (86 cells within an 8-state CA), which makes them the first explicit implementation of non-trivial self-replication. Langton's work sparked extensive subsequent research on compact replicating structures.

One branch of this research pursued further simplification: Byl \cite{byl1989self} designed a 12-cell loop in a CA with 6 states, Reggia et al. \cite{reggia1993simple} produced a 5-cell replicator in a 6-state CA, Ibáñez et al. \cite{ibanez1995self} developed loops where, interestingly, replication is based on self-inspection rather than having the replicator's description explicitly available as was the case in previous constructions. Morita and Imai \cite{morita1996self} designed self-replicating loops in a reversible cellular automaton. 

A second branch focused on enriching the computational power of Langton’s loops: Tempesti \cite{tempesti1995new} and Perrier et al.~\cite{perrier1996toward} demonstrated variants capable of executing attached programs, while Chou and Reggia \cite{chou1998problem} showed that populations of self-replicating loops could be adapted to (inefficiently) solve hard problems such as NP-complete satisfiability (SAT); in their design each replicant receives a partial solution to the problem that is further modified during replication and passed on to the next generation. In \cite{morita1996self}, Morita and Imai construct a 2D reversible cellular automaton capable of self-replication.

For more details, nice overviews include \cite{sipper1998fifty}, \cite{reggia1998self} or \cite{sayama2024self}.

\subsubsection{Self-replicators in other systems}

\paragraph{Computer programs}
Studies on quines—programs that output their own code—have been performed in \cite{bratley1972self, burger1980self}. In \cite{ray1991tierra}, Ray constructed a self-replicating program using the assembly language of a virtual machine system called Tierra. He initialized a colony of programs containing this self-replicator and allowed them to evolve through interactions while competing for CPU time and memory. The programs, subject to mutations, diversified in intriguing ways, exhibiting behaviors such as symbiosis and parasitism.

In contrast to the manually designed self-replicators described above, some works explored the fascinating question of how self-replication can naturally emerge from system's dynamics, relating to the research on the origin of life. Koza \cite{koza1994spontaneous} studied populations of Lisp programs and noticed that self-replicating programs can be found via random search, and further evolved using genetic programming. In \cite{arcas2024computationallifewellformedselfreplicating}, the authors considered a population of short programs written in an esoteric programming language that evolve via interacting and modifying one another. Interestingly, the authors observed an emergence of self-replicating programs. 

In the program-based systems above, the \emph{constructor} is largely externalized to the substrate's operational semantics; that is, to the ambient dynamical laws of the interpreter or virtual machine (memory allocation, instruction decoding, scheduling, divide or clone primitives). Consequently, these replicators chiefly realize transcription (copying a description), while construction or translation is effected by the host system. Said differently, the type of self-replication considered in the program-based systems entails copying the description of an organism and not the copying of the means to produce the organism, since the production mechanics is fully abstracted into the ``laws of physics'' of the dynamics.  This stands in contrast to von Neumann– or Langton–style cellular automata, where the local transition rule is comparatively simple and domain-general, and the constructor is encoded within the spatial organization of the replicator itself.

\paragraph{Physical systems} Interesting examples of self-replicators in physical substrates include Penrose's chains of mechanical wooden blocks \cite{penrose1958mechanics, penrose1957self} and Jacobson's circulating trains on model railway tracks \cite{jacobson1958models}; for detailed reviews see \cite{freitas, moses2020robotic}. 

\subsubsection{Criteria for self-replication}

The first attempt at giving a formal criterion for self-replication was presented in the seminal work by Moore \cite{moore1962machine}. 

\paragraph{Moore's criterion} Let $\A = (S^{\Z^d}, F)$ be a cellular automaton with local rule $f$ and a quiescent state (a state $q$ satisfying $f(q, q, q, \ldots, q) = q$). Consider a configuration $c$ which contains a pattern $\mP$ in a finite region surrounded by the quiescent state. Then $c$ is self-replicating if for every $k \in \N$ there exists a time-step $t$ such that $F^t(c)$ contains the pattern $\mP$ in at least $k$ disjoint regions.\\

One can interpret the pattern $\mP$ as an organism surrounded by empty space given by the quiescent state. As time progresses, the organism has to appear in a growing number of disjoint regions of the CA configuration. The definition is very natural; however, it does not capture the notion of self-replication exactly. On the one hand, it may feel too constraining: one may observe systems where the copies of $\mP$ do not appear in the exact same form, but may be perturbed, for instance, by simple geometrical transformations. Moreover, assuming the background has to consist purely of the quiescent state may be too strict; in a similar way it proves to be too constricting when considering local universality, as exemplified by the proof of Rule 110's Turing completeness. Indeed, there may be systems that only allow for self-replication with richer CA backgrounds. On the other hand, since Moore's criterion does not assume anything about the pattern $\mP$, there are trivial cases that satisfy the criterion: such as a single black cell against a white background indefinitely growing in all directions.

Subsequently, more elaborate criteria have been developed, building upon Moore's initial variant. Smith \cite{smith1968simple, smith1991simple} requires the pattern $\mP$ to be capable of universal computation (without specifying such a notion exactly). Lohn and Reggia \cite{lohn2002automatic} developed a variant of Moore's criterion which excludes some trivial cases by requiring the pattern $\mP$ to change its shape as opposed to just being static; moreover, they also allow the pattern's copies to be rotated. In \cite{nehaniv1998self}, Nehaniv and Dautenhahn extend Moore's criterion to allow for modified offsprings. In \cite{adams2003universal} Adams and Lipson discuss a graded measure of a system's capability of self-replication rather than proposing a strict binary criterion, though their application to cellular automata is illustrated only on a finite cyclic grid.

\subsubsection{Self-replicators in one-dimensional cellular automata} In \cite{smith1968simple, smith1991simple} Smith proves the existence of what he calls a 1D universal self-replicator.  This result, compared to our work in Appendix \ref{appsec_1d_sr}, differs in a fundamental way: in the system he describes, the copying mechanism of CA structures is hardwired into the logic of the CA rule. This is in stark contrast with ``genuine'' self-replicators where the copying mechanism is a much higher-level process emerging from the expressive laws of the CA rules.

\subsection{Local universality}
The study of cellular automata as computational systems originates with von Neumann’s pioneering work on self-replication. Indeed, demonstrating a CA's Turing completeness remains the most rigorous way to establish a CA’s non-triviality. In order to prove a CA's local universality, one needs to show the CA can simulate another computationally universal system, most typically a universal Turing machine. There is a rich line of work constructing locally universal CAs with various properties, which we briefly review below.

\paragraph{Compact computation} 
Ever since the work of Shannon \cite{shannon1956universal}, there has been an ongoing quest to construct small computationally universal systems. The classical line of work considers Turing machines and aims to jointly minimize their state space and alphabet size \cite{minsky1960, minsky1966size, watanabe19615, watanabe19724, robinson1991minsky}. Parallel to this line of work is a series of constructions of locally universal cellular automata, progressively more compact with respect to their dimensionality, neighborhood, and number of states. These two lines of work influenced each other; in fact, some of the small universal Turing machines were used to construct compact locally universal cellular automata. In his seminal work, Banks \cite{banks1970universality} used an approach combining ideas of local and global universality: he constructed a variety of CAs that can simulate von Neumann's universal self-replicator, and are thus also computationally universal. His argument relied on implementing basic logic gates in 2D CAs. Banks was also careful enough to distinguish two cases: a 3-state, 5-neighbor CA which used finite patterns against a quiescent background, and a 2-state, 5-neighbor CA which requires an ``infinite'' initial configuration. A locally universal 4-state, 5-neighbor, 2D CA was constructed in \cite{noural2006universal}. Further simplifications in one and two dimensions were presented in \cite{smith1971simple, lindgren1990universal}. 
The Game of Life introduced by Conway is certainly among the most famous cellular automata and the first rule to be proven locally universal by analysis of a given rule \cite{berlekamp2004winning, turing_universality_of_gol} as opposed to previous examples, which were carefully designed to satisfy universality. This line of work culminated in the proof of local universality of Rule 110: a 2-state 3-neighbor 1D CA. The elaborate proof required a careful analysis of the CA's gliders and their interactions \cite{cook_og, cook2009concrete}. 

An important question pertains to the space efficiency of the simulation; i.e.~if system $\B$ (such as a CA) simulates a system $\A$ (such as a Turing machine), how much more time and space does $\B$ need to simulate $t$ steps of $\A$'s computation? Some early constructions of small universal Turing machines relied on a 2-tag system (a computational system defined in \ref{def:tag_system}) simulating any Turing machine: a simulation requiring an exponential blow up in space and time. As a result, it was believed that for small universal systems, such a redundancy might be necessary. In \cite{neary2006small}, Neary and Woods showed this simulation can be realized with just a polynomial time blow-up. The authors also showed that Rule 110 can simulate a universal Turing machine in polynomial time \cite{p_completeness_of_rule110}, a result surprising to Cook \cite{cook2009concrete}.

\paragraph{Reversible computation}
Reversibility and physics-inspired constructions also played a central role. In \cite{toffoli1977computation} Toffoli showed that any $d$-dimensional CA can be simulated by a $d+1$-reversible CA, thus proving local universality for automata of dimensions 2 and higher. 
Subsequently, Morita and Harao \cite{morita1989computation} showed that a 1D reversible cellular automaton can simulate any reversible Turing machine, thus completing the picture for dimension one. In \cite{dubacq1995simulate}, Dubacq later showed this can be done in real time. In \cite{morita1995}, Morita shows that any 1D CA with finite configurations (finite patterns surrounded by a quiescent state) can be simulated by a 1D reversible CA.

\paragraph{Computation and dynamics}

A closely related line of work examines the relation between computation and differentiable dynamical systems. Given a differentiable dynamical system $f : M \to M$ on a compact manifold $M$, we can ask what kind of computation this dynamics can instantiate, and whether it is Turing-universal. Various works~\cite{moore1990unpredictability, moore1991generalized, branicky1995universal, moore1998finite, moore1998dynamical, koiran1999closed, asarin2001perturbed, gracca2005robust, gracca2008computability, rosen2011anticipatory, rosen1991life, tao2017universality, cardona2021constructing, cardona2022turing, cardona2023computability, cardona2024hydrodynamic, computational_dynamical_systems} have studied this question in various frameworks. We most closely follow and adapt the ``computational dynamical systems'' approach of~\cite{computational_dynamical_systems}.

The work~\cite{computational_dynamical_systems} emphasizes that to associate the dynamics of $f : M \to M$ with a Turing machine $\mathcal{T}$ with dynamics $T : C \to C$, we require encoders $\mathcal{E} : C \to M$ and decoders $\mathcal{D} : M \to C$ such that
\begin{align}
\mathcal{D} \circ f^n \circ \mathcal{E} = T^n \label{eq:encoding}
\end{align}
for all $n \geq 0$ (this generalizes to settings where there is a computational slowdown of $f$ relative to $T$), where the encoder and decoder have lower computational complexity than $T$. While this is formalized in~\cite{computational_dynamical_systems}, we provide some intuition here. Suppose $T$ is Turing-universal and we have found $\mathcal{E}, \mathcal{D}$ satisfying~\eqref{eq:encoding}. We might be tempted to conclude that $f$ is Turing-universal as well. However, if $\mathcal{E}, \mathcal{D}$ are themselves Turing-universal, they may `do the work' of the Turing-universal computation in the relation~\eqref{eq:encoding} even if the dynamics of $f$ is extremely simple, thus obscuring whether $f$ exhibits computationally interesting dynamics. The basic principle, stated heuristically, is that $\mathcal{E}$ and $\mathcal{D}$ must be less computationally complex than the computation we ascribe to $f$ (here, the computation performed by $T$).

In the context of differentiable dynamical systems on a compact manifold $M$, constraining the complexity of the encoder and decoder amounts to constraining the complexity of regions in $M$ that encode Turing machine configurations. Regions of higher complexity are those carved out by a larger number of semi-algebraic inequalities in a controlled manner. This constraint effectively precludes the encoder and decoder from hiding computation in the extremely fine-grained, short-distance structure of the regions in $M$ coding for Turing machine configurations.

In defining `local universality' for CAs, we develop and adapt the work of~\cite{computational_dynamical_systems}. Whereas in the differentiable dynamical systems setting on a compact manifold $M$ the concern is that we might encode Turing machine configurations in arbitrarily complex, small distance-scale structures of subsets of $M$, in the CA setting the concern is that we might encode Turing machine configurations in arbitrarily complex, large distance-scale structures of a CA configuration. In the differentiable dynamical systems setting, we grapple with arbitrarily fine structure in the continuum at finite volume, whereas in the CA setting we grapple with structure at infinite volume (where no structure exists below the distance scale of a single cell). Accordingly, our definitions in the CA setting constrain the possible complexity of infinite volume encodings of Turing machine configurations.

\subsection{Global universality}
A more recent line of work focuses on studying the relationship between CAs, as opposed to relating a CA's computational dynamics to a Turing machine. Indeed, in the case of local simulation, one has to grapple with the profound differences between a CA and a Turing machine: CAs are a massively parallel model of computation, operating on an infinite grid, with no predefined halting state; this is in stark contrast with the serial design of Turing machines. Informally, we say that a CA $\B$ globally simulates $\A$ if $\B$ can effectively reproduce any space-time diagram of $\A$. There are multiple ways of making precise the notion of effective reproduction of space-time diagrams which led to various notions of global simulation studied in the past. For a technical overview of the different definitions, see Section \ref{appD:related_word}. Here, we briefly summarize the main results.

The first line of work focuses on positive results: typically constructing examples of globally universal CAs with various properties. Often, the notion of global (also often called intrinsic) simulation is assumed implicitly from the construction. The seminal work of Albert and Culik \cite{simple_universal_ca} shows a construction of a 1D globally universal CA. This ignited the quest for the most compact, globally universal CAs \cite{imai2000computation, ollinger2002quest}. The already mentioned work of Toffoli \cite{toffoli1977computation} considers global simulation to show that any $d$-dimensional CA can be simulated by a $d+1$-reversible CA. This work is complemented by Hertling \cite{hertling1998embedding} who precisely formalizes a modern notion of global simulation to show that any $d$-dimensional reversible CA cannot simulate $d$-dimensional irreversible CAs, thus, can never be globally universal. On the other hand, Durand-Lose shows the existence of reversible automata that can simulate any other reversible automata in the same dimension \cite{intrinsic_universality_of_reversible_1D_ca}. Durand and Róka \cite{the_game_of_life_universality_revisited} revisit the notion of computational universality of Game of Life and show it is globally universal. 

A complementary line of research focuses on the global simulation notion itself, and studying its properties, leading to a series of interesting results. Mazoyer and Rapaport \cite{inducing_order_on_ca_by_grouping} show the existence of infinitely many CA classes that are incomparable with respect to the relation of global simulation, and further show that no CA can be globally universal in real-time. Their results were further generalized in \cite{bulking1, bulking2, hudcova2024simulation}. In \cite{intrinsic_universality_problem_of_1d_cas}, Ollinger shows that it is in general undecidable whether a CA is globally universal. For more details on these results, see Section \ref{appD:related_word}.

\section{Preliminaries}
\label{appsec_preliminaries}

\subsection{Cellular automata}
Letting $d \in \N$, we call $\Z^d$ the \emph{$d$-dimensional grid} and we call its elements the cells. For a $v \in \Z^d$ where $v = (v_1, \ldots, v_d)$, we denote
$$\lVert v \rVert = \max \{|v_1|, \ldots, |v_d|\}.$$ 
For a finite set $S$ we define an \emph{$S$-configuration} (sometimes just configuration) of the grid to be a mapping $c: \Z^d \rightarrow S$. We denote the space of all $S$-configurations by $S^{\Z^d}$, and will often write $c \in S^{\Z^d}$. Sometimes, for $v \in \Z^d$ we write $c_v$ instead of $c(v)$. 

A \emph{$d$-dimensional neighborhood} is a tuple $N = (n_1, n_2, \ldots, n_k)$ where each $n_i \in \Z^d$. This yields \emph{a relative neighborhood of each cell} $v \in \Z^d$ to be $(v + n_1, \ldots, v + n_k)$. Now we can proceed with the definition of a CA.

\begin{definition}[Cellular automaton]
A $d$-dimensional cellular automaton $\A$ is given by a neighborhood $N = (n_1, n_2, \ldots, n_k)$, a finite set of states $S$ and a local update rule $f: S^k \rightarrow S$. This yields the CA's global update rule $F: S^{\Z^d} \rightarrow S^{\Z^d}$ defined for each configuration $c \in S^{\Z^d}$ as:
$$F(c)(v) = f(c(v+n_1), \ldots, c(v + n_k)) \quad \text{for each cell } v \in \Z^d. $$
The global rule updates the state of a cell $v$ by looking at the states of cells in its relative neighborhood, and using the local rule to determine the update. This is done synchronously for all the cells in parallel. We often write $\A = (S^{\Z^d}, F)$.
\end{definition}

We note that a CA $\A = (S^{\Z^d}, F)$ can be defined by multiple neighborhoods and local rules. For instance, if $\A$ is given by neighborhood $N$ and local rule $f$, we can increase the neighborhood size by adding extra cells, and adjusting the local rule so that its updates do not depend on the newly added cells. Using this approach, we can extend any neighborhood to contain exactly all cells $v \in \Z^d$ such that $\lVert v \rVert \leq r$ for some fixed $r \in \N$. If a CA's neighborhood is of such a form, we say the CA has \emph{radius $r$}. In certain scenarios, this allows us to assume, without loss of generality, that each CA has radius $r$ for some $r \in \N$.

Given a CA $\A = (S^{\Z^d}, F)$ and in \emph{initial configuration} $c \in S^{\Z^d}$ we can iterate the CA on $c$ to obtain a \emph{trajectory}
$c \mapsto F(c) \mapsto F^2(c) \mapsto F^3(c) \mapsto \ldots$ A nice introduction to cellular automata is for instance \cite{kari2022cellular}.

\paragraph{Visualizing CA dynamics}
For practical purposes, when visualizing the CA trajectories, we consider a \emph{finite cyclic grid} of the form $\Z_{m_1} \times \Z_{m_2} \times \cdots \times \Z_{m_d}$ for some $m = (m_1, \ldots, m_d) \in \Z^d$ and simply compute all the cell indices modulo $m$. Thus, we essentially obtain the dynamics on a discrete $d$-dimensional torus. 

For a 1D CA operating on a cyclic grid of size $n$ with global rule $F$ we define the \emph{space-time diagram of the CA with initial configuration $u \in S^n$ and $t$ time-steps} to be a matrix whose rows are exactly $u, F(u), F^2(u), \ldots, F^t(u)$. An example of a particular space-time diagram is shown in Fig.~\ref{fig:rule90}.

\begin{figure}[htbp!]
    \centering
    \includegraphics[width=0.9\linewidth]{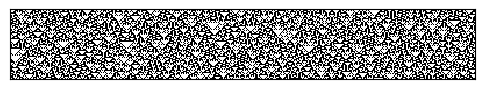}
    \caption{Space time diagram of a 1D CA with states 0 (white) and 1 (black), radius 1, and local rule $f(x, y, z) = x + z \bmod 2$ operating on a cyclic configuration of size 250. Time is progressing downwards.}
    \label{fig:rule90}
\end{figure}

\paragraph{Periodic configurations and shifts} Let $v \in \Z^d$. We define the shift operator $\sigma_v: S^{\Z^d} \rightarrow S^{\Z^d}$ by $\left(\sigma_v(c)\right)_w = c_{v+w}$ for each $w \in \Z^d$. Let $n \in \N$. We say that a configuration $c \in S^{\Z^d}$ is periodic if there exist $d$ linearly independent vectors $v_1, \ldots, v_d$ such that $\sigma_{v_i}(c)=c$ for all $i \in \{1, \ldots, d\}$. It is clear that every CA preserves the periodicity of a configuration.

\paragraph{Topological dynamics}
We can endow the space $S^{\Z^d}$ with a metric as follows. We define the distance between two configurations $c_1, c_2 \in S^{\Z^d}$ as:
\begin{equation*}
    d(c_1, c_2) =
    \begin{cases*}
      0 & if $c_1 = c_2$ \\
      2^{-\min \{ \lVert v \rVert \mid c_1(v) \neq c_2(v)\} }      & if $c_1 \neq c_2$
    \end{cases*}\,.
  \end{equation*}
Thus, two configurations are close if they agree on a large area around the origin $0$. One can easily verify that this is indeed a metric that induces a topology on $S^{\Z^d}$, making $S^{\Z^d}$ a compact space with a countable base. This allows us to talk, for instance, about the continuity of maps $F: S^{\Z^d} \rightarrow S^{\Z^d}$. We note that different choices of vector norms $\lVert \,\cdot\, \rVert$ can induce the same topology on $S^{\Z^d}$.  For now, we highlight a famous result characterizing cellular automata in terms of their dynamical properties. It is quite easy to see that the global map of every CA is a continuous mapping of the topological space, and that it commutes with every shift operator. The following theorem shows also the converse, using the compactness of the topological space.

\begin{theorem}[Curtis–Hedlund–Lyndon theorem]\label{thm:og_curtis_hedlund_lyndon}
Let $S$ be a finite set and $d \in \N$. A function $F: S^{\Z^d} \rightarrow S^{\Z^d}$ is a global rule of a CA if and only if it is continuous and commutes with every shift operator.
\end{theorem}

For more details on topological dynamics and CAs, see for instance \cite{kurka2003topological}.

\paragraph{Subshifts of finite type} Let $F \subseteq \Z^d$ be a finite set. By an \emph{$F$-pattern} over $S$ we understand a mapping $u: F \rightarrow S$. Each finite set of patterns $L$ defines a \emph{subshift of finite type (SFT)} by
$$S_L = \{c \in S^{\Z^d} \mid \text{no element of } L \text{ appears in c} \}.$$
It is easy to see that each SFT is a topologically closed, shift-invariant set.

\paragraph{Elementary cellular automata}
1D cellular automata with states $\2=\{0, 1 \}$ and a local rule with radius 1 are called elementary cellular automata (ECAs). Each local rule $f: \2^3 \rightarrow \2$ is uniquely described by its \textit{Wolfram number} given by:
$$2^0 f(0, 0, 0) + 2^1 f(0, 0, 1)  + \cdots + 2^7 f(1, 1, 1).$$
We will refer to each ECA by the corresponding Wolfram number of its local rule. The class of ECAs is frequently used for studying different CA properties due to its relatively small size; there are only 256 of them. Even such a small class contains CAs with complex behavior, e.g.~some of them have been shown to be Turing-complete in \cite{cook_og}; as we will discuss in much more detail later.

\begin{example}
    Fig.~\ref{fig:rule90} shows the space-time dynamics of Rule 90. 
\end{example}

\section{Global simulation}

\subsection{Defining global simulation}
Intuitively, a CA $\B$ globally simulates $\A$ if we can efficiently recover any space-time diagram of $\A$ from a suitable space-time diagram of $\B$. The correspondence between space-time diagrams of $\A$ and $\B$ are (analogously to local simulation) realized by an encoder and decoder, which will be some suitable ``computationally reasonable'' mappings. The configurations of the two systems are \emph{infinite} sequences of states, which as we will see is different from the local simulation case. These considerations will force us to only consider a fairly restricted class of encoders and decoders. The technical difficulty comes from specifying how exactly can the space-time diagrams of a CA be transformed. We say that a CA is globally universal if it can globally simulate any CA in the same dimension. Before we proceed with the formal definition, we discuss a few examples of global simulation.

\subsubsection{Examples}
\begin{example}[Reflections of Space-time Diagrams]\label{ex:reflections}
Let us consider two elementary CAs, Rule 110 and Rule 124. By observing their diagrams in Fig.~\ref{fig:spacetime_reflection} we can notice that their local update rules are identical upon exchanging the role of left and right neighbor. This induces a relationship between the two CAs' space-time diagrams. Let $\Theta: \{0,1 \}^\Z \rightarrow \{0,1 \}^\Z$ be the reflection of configurations over the origin. If we compare the space-time diagram of Rule 110 iterated on $c \in \{0,1 \}^\Z$ and the space-time diagram of Rule 124 iterated from $\Theta(c)$, it is apparent that one is just a reflection of the other, as shown in Fig.~\ref{fig:spacetime_reflection}. Since reflection is a computationally simple transformation, it would be natural to say that Rule 110 and Rule 124 globally simulate one another.

\begin{figure}[!htb]
    \centering
    \begin{minipage}{.45\textwidth}
        \centering
        Rule 110\\
        \includegraphics[width=0.8\linewidth]{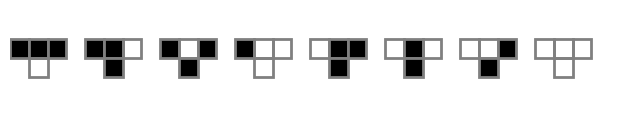}
        \includegraphics[width=0.9\linewidth]{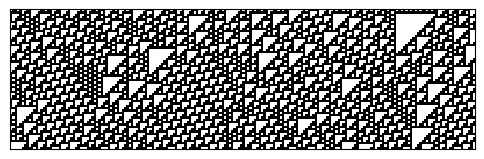}
    \end{minipage}%
    \begin{minipage}{0.45\textwidth}
        \centering
        Rule 124\\
        \includegraphics[width=0.8\linewidth]{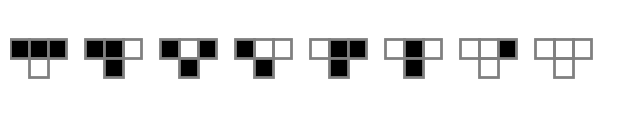}
        \includegraphics[width=0.9\linewidth]{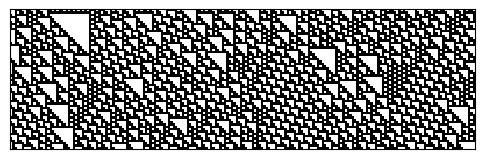}
    \end{minipage}
\caption{Rule tables and space-time diagrams of two elementary CAs.}\label{fig:spacetime_reflection}
\end{figure}
\end{example}

In \cite{simple_universal_ca}, the authors show what is probably the first construction of a globally universal CA. Specifically, they construct a 1D CA that can globally simulate any other 1D CA, however, without explicitly specifying what ``globally simulate'' means. Their construction follows a series of reductions, one of which is showcased in the next example.

A 1D CA is \emph{one-way} if its local function depends only on the central cell and its right neighbor. A CA is totalistic if its state space can
be embedded in $\N$ in such a way that its local function depends only on the sum of the states in its input.

\begin{example}[Shifts and Time Delays]\label{ex:global_sim_shifts_delays} For each 1D cellular automaton $\A$ with radius $r=1$ there exists a one-way CA $\B$ that can simulate it \cite{simple_universal_ca}.

\begin{proof}
    Let $\A = (S^\Z, F)$ be a CA with radius 1. We construct a one-way CA $\B = (T^\Z, G)$ as follows: $\B$ will simulate one iteration of $\A$ in two steps, first to aggregate information, second to process it. Suppose that $\A$ is given by the local rule $f: S^3 \rightarrow S$. We define the states of $\B$ as $S \cup S \times S$ and we define its local rule $g$ as follows: 
    \begin{alignat*}{2}
        g(s_1, s_2) &= s_1s_2 &&\text{ if } s_1, s_2 \in S\\
        & = f(a, b, c) &&\text{ if } s_1 = ab, \, s_2 = bc\\
        &\text{and is otherwise defined arbitrarily.}
    \end{alignat*}
The simulation is illustrated in Fig.~\ref{fig:one_way_simulation}.

\begin{figure}[htpb!]
    \centering
    \includegraphics[width=0.7\linewidth]{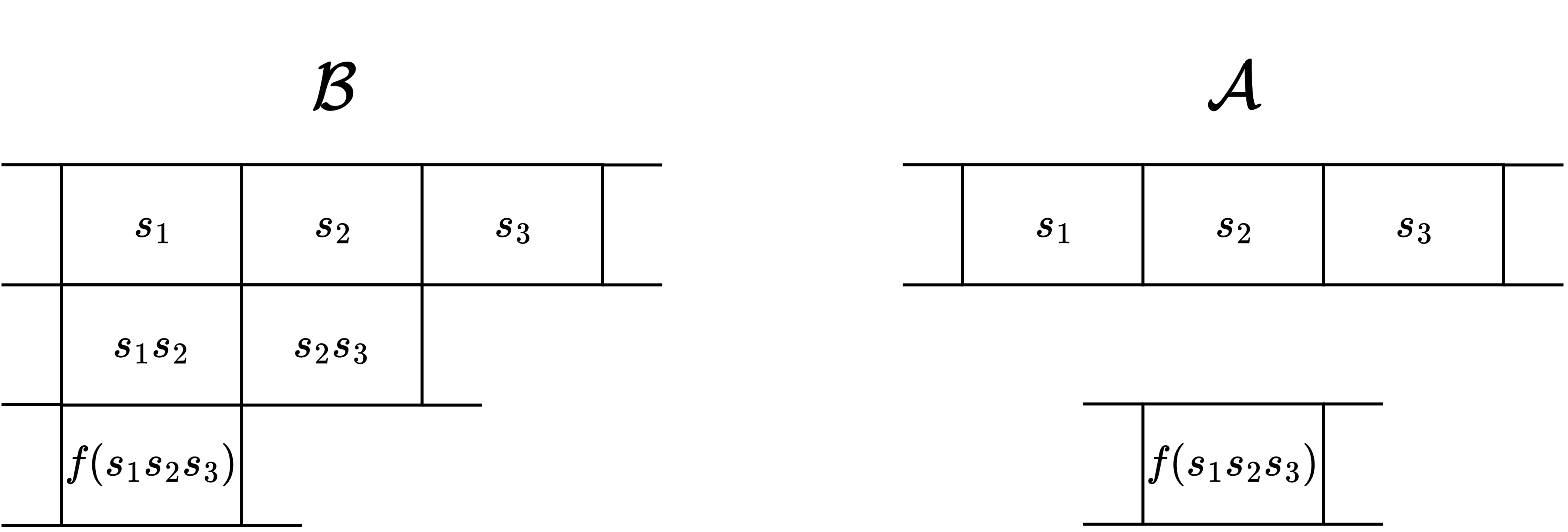}
    \caption{One-way CA simulating a CA with radius $r=1$.}
    \label{fig:one_way_simulation}
\end{figure}

Notice that in general, the dynamics of $\B$ can be richer, but if the initial configuration $c$ of $\B$ contains only states from $S$; i.e.~$c \in S^\Z \subseteq (S \cup S \times S)^\Z$ then the dynamics of $\B$ mimics the dynamics of $\A$. Specifically, in order to transform any space-time diagram of $\B$ starting from some $c \in S^\Z$ we have to skip every second time-step since $\B$ takes longer to aggregate information due to its smaller neighborhood. Lastly, in order for the exact cell positions to match, we have to shift the $k$-th row that we did not skip by $k$ cells to the right. We can summarize the relationship of the two CAs by relating their global rules:
$$\sigma_{-1} \circ G^2 (c) = F(c) \quad \text{for all } c \in S^\Z.$$
\end{proof}
\end{example}

In the previous two examples, whenever $\B = (T^{\Z^d}, G)$ globally simulated $\A = (S^{\Z^d}, F)$ we had that $T \supseteq S$. Clearly, if it was a necessary condition that $\B$ has to have a larger state space than $\A$ in order to simulate it, there could not exist any globally universal CA. Thus, given a CA $\B = (T^{\Z^d}, G)$ we have to be able to interpret its state space as being arbitrarily rich. This can be readily done by yet another geometrical transformation: packing. Informally, packing partitions the grid $\Z^d$ into larger blocks of cells; for instance for the 1D grid, the blocks could be 2 consecutive cells. This way, we can interpret every CA $\B = (T^\Z, G)$ as operating on a larger state space; for the case of the two consecutive cells, this would be $\B = (T^\Z, G) \rightsquigarrow \B^{\langle 2 \rangle} = \left((T^2)^\Z, G^{\langle 2 \rangle}\right)$. Before we define packing formally, we give one last example where this is utilized.

\begin{figure}[h!]
    \centering
    \begin{minipage}{.45\textwidth}
        \centering
        Rule 76\\
        \includegraphics[width=0.8\linewidth]{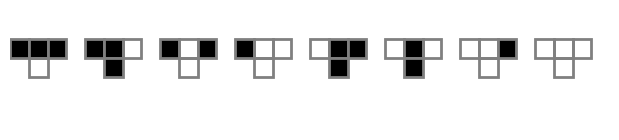}
        \includegraphics[width=0.9\linewidth]{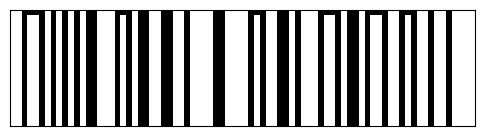}
    \end{minipage}%
    \begin{minipage}{0.45\textwidth}
        \centering
        Rule 4\\
        \includegraphics[width=0.8\linewidth]{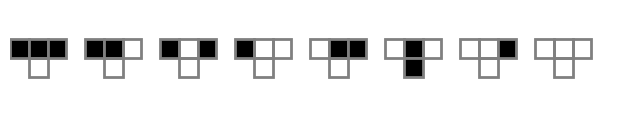}
        \includegraphics[width=0.9\linewidth]{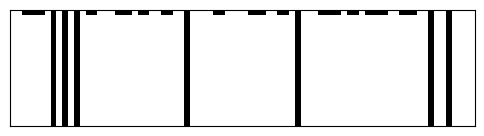}
    \end{minipage}
\caption{Rule tables and space-time diagrams of elementary CAs 76 and 4.}\label{fig:spacetime_packing}
\end{figure}

\begin{example}[Packing Operator]\label{ex:global_sim_packing}
    Let us consider the elementary CAs Rule 4 and Rule 76, with global functions $F$ and $G$ respectively (see Fig.~\ref{fig:spacetime_packing}). 
    
The dynamics of both CAs is relatively simple, since they are both idempotent; $F^2 = F$ and $G^2 = G$. Whereas Rule 4 only preserves blocks of the form 010, Rule 76 acts as identity for all neighborhoods except for 111. We will discuss that Rule 76 can simulate Rule 4.
We define two local mappings:\\
\begin{minipage}{.45\textwidth}
        \centering
       \begin{align*}
           e: \2 &\rightarrow \2^2\\
           0 &\mapsto 00\\
           1 &\mapsto 11.
       \end{align*}
    \end{minipage}%
    \begin{minipage}{0.45\textwidth}
        \centering
       \begin{align*}
           d: \2^2 &\rightarrow \2\\
           00, 01, 10 &\mapsto 0\\
           11 &\mapsto 1.
       \end{align*}
    \end{minipage}\\
We define the encoding $\E: \2^\Z \rightarrow \2^\Z$ that maps each configuration, cell by cell, using the local map $e$ and concatenating the cell tuples. Similarly, we define $\D: \2^\Z \rightarrow \2^\Z$ that first ``partitions'' each configuration into blocks of two consecutive cells, and uses the local map $d$ to decode them. It is a simple exercise to show that for each initial configuration $c \in \2^\Z$ and for any number of iterations of $n$, it holds that $\D \circ G^{2n} \circ \E(c) = F^n(c)$; the simulation is illustrated in Fig.~\ref{fig:76_sim_4}.

 \begin{figure}[!htb]
    \centering
    \includegraphics[width=0.6\linewidth]{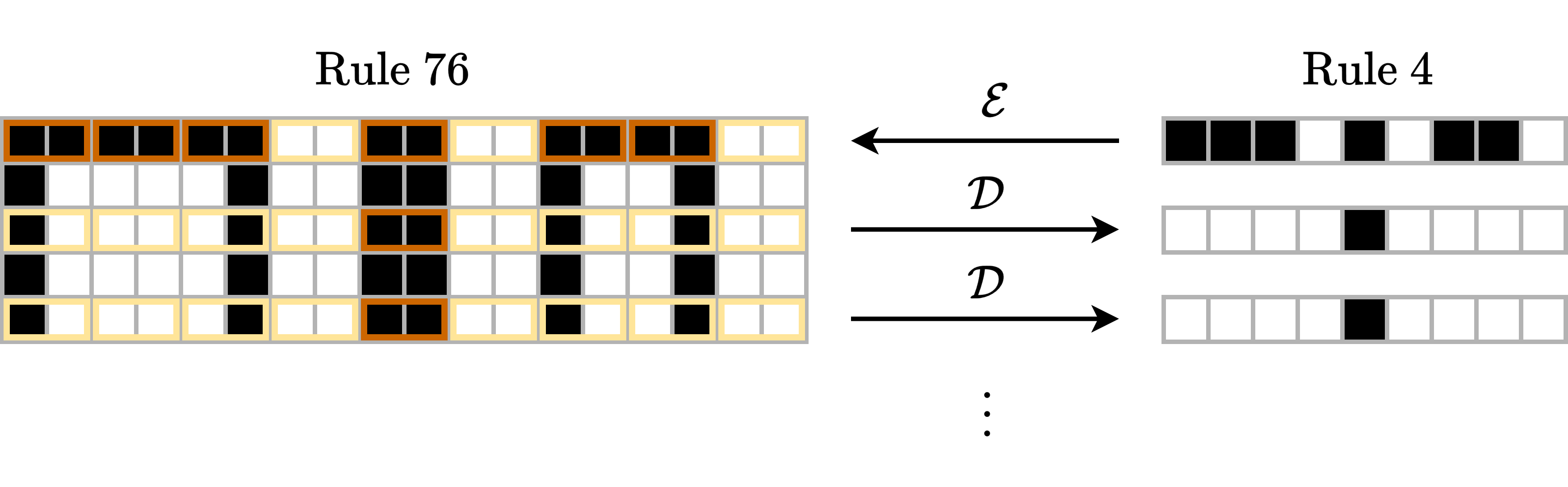}
\caption{Rule 76 globally simulating Rule 4.}\label{fig:76_sim_4}
\end{figure}

\end{example}

The three examples above showcase three types of geometrical transformations that were used to modify space-time diagrams of $\B$ to ``uncover'' the space-time diagrams of $\A$, in order to show that $\B$ globally simulates $\A$. The transformations were:
\begin{itemize}
    \item time-delays (skipping time-steps in diagrams of $\B$)
    \item shifts (shifting spatial coordinates in diagrams of $\B$)
    \item packings (aggregating together cells into larger ``supercells'')
\end{itemize}
The three transformations are not arbitrary: in \cite{bulking1} the authors show that any geometrical transformation that maps CA space-time diagrams into CA space-time diagrams while preserving time causality is a composition of these three operators.

Let us finish this section with the formal definition of global simulation.
\begin{definition}[Global CA simulation]\label{def:global_simulation}
    Let $\A = (S^{\Z^d}, F)$ and $\B = (T^{\Z^d}, G)$ be $d$-dimensional cellular automata. We say that $\B$ globally simulates $\A$ if there exists a vector $v \in \Z^d$, a time-delay $\tau \in \N$, a mapping $\E: S^{\Z^d} \rightarrow T^{\Z^d}$ and a partial mapping $\D: T^{\Z^d} \rightharpoonup S^{\Z^d}$ satisfying:
    \begin{equation}\label{eq:global_simulation}
        \D \circ \sigma_v \circ G^\tau \circ \D^{-1} (c) = F(c) \quad \text{ for all } c \in S^{\Z^d},
    \end{equation}
    where the mapping $\D$ satisfies the following three conditions:
    \begin{enumerate}
        \item The domain of $\D$ is a subshift of finite type.
        \item $\D$ is a continuous mapping.
        \item There exists a full-rank matrix $M \in \Z^{d \times d}$ such that for all $w \in \Z^d$: $\sigma_w \circ \D = \D \circ \sigma_{M \cdot w}$.
    \end{enumerate}
    Moreover, the mapping $\E$ satisfies the following three conditions:
    \begin{enumerate}
        \item $\D \circ \E = \text{\rm id}_{S^{\Z^d}}$.
        \item $\E$ is a continuous mapping.
        \item There exist two full-rank matrices $M_1, M_2 \in \Z^{d \times d}$ such that for all $w \in \Z^d$: $ \sigma_{M_2 \cdot w} \circ \E = \E \circ \sigma_{M_1 \cdot w}$
    \end{enumerate}
In such a case, we write $\A \preceq \B$. We say that $\B$ is \emph{globally universal} if it can globally simulate all $d$-dimensional cellular automata.
\end{definition}

\begin{remark}[Lower-dimensional simulation via lifting]\label{rem:lower_dim_lifting}
When we say that a $d$-dimensional CA $\mathcal{B}$ can globally simulate \emph{any} CA of the same or lower dimension, we handle the lower-dimensional case $d' < d$ through a standard lifting procedure: We embed $\mathbb{Z}^{d'}$ into $\mathbb{Z}^{d}$ via the first $d'$ coordinates, extend any $d'$-dimensional configuration to be constant along the remaining $d-d'$ coordinates, and lift the local rule to ignore these extra coordinates. This transforms any $d'$-dimensional CA $\mathcal{A}$ into a $d$-dimensional CA $\mathcal{A}'$. Thus ``$\mathcal{B}$ globally simulates $\mathcal{A}$'' means that $\mathcal{B}$ globally simulates $\mathcal{A}'$ in the sense of Definition~\ref{def:global_simulation}.
\end{remark}

We will call $\D$ the decoder and $\E$ the encoder. In the next section, we formally introduce the packing operator and show that both the encoder and decoder are ``computationally reasonable'' mappings since they essentially act as cellular automata on ``packed'' configuration spaces. This importantly ensures that one iteration of the encoder or decoder cannot perform an arbitrarily sophisticated computation. Even though the encoder does not play a direct role in condition \eqref{eq:global_simulation}, its existence guarantees that for each initial configuration $c$ of $\A$, we can efficiently recover a corresponding $c' \in \D^{-1}(c)$ of $\B$. Before we proceed with the next section, we discuss two important properties of global simulation: the relationship in \eqref{eq:global_simulation} extends to any number of iterations of the CAs, and the relation $\A \preceq \B$ between CAs of the same dimension is a preorder. We discuss this in the following lemmas.

\begin{lemma}\label{lemma:global_sim_iterations}
    Let $\A = (S^{\Z^d}, F)$ and $\B = (T^{\Z^d}, G)$ be $d$-dimensional cellular automata such that $\B$ globally simulates $\A$ via a shift vector $v \in \Z^d$, a time-delay $\tau \in \N$, an encoder $\E: S^{\Z^d} \rightarrow T^{\Z^d}$, and a decoder $\D: T^{\Z^d} \rightharpoonup S^{\Z^d}$. Then for any $n \in \N$ it holds that
    $$\D \circ (\sigma_v \circ G^\tau)^n \circ \D^{-1} (c) = F^n(c) \quad \text{ for all } c \in S^{\Z^d}.$$
\end{lemma}
\begin{proof}
     First, we notice that \eqref{eq:global_simulation} implies $\sigma_v \circ G^{\tau} (\dom \D) \subseteq \dom \D$. Thus, for any $c \in S^{\Z^d}$ and for any $n \in \N$, $\D \circ (\sigma_v \circ G^\tau)^n \circ \D^{-1} (c)$ is non-empty. Let $c \in S^{\Z^d}$ and let $c' \in \dom(\D)$. Since $ \{c' \} \subseteq \D^{-1} \circ \D (\{c'\}) $, we have that:
    \begin{align*}
        \D \circ (\sigma_v \circ G^\tau)^2 \circ \D^{-1} (c) &\subseteq \D \circ (\sigma_v \circ G^\tau) \circ \D^{-1} \circ \D \circ (\sigma_v \circ G^\tau) \circ \D^{-1} (c) \\
        &\subseteq \D \circ (\sigma_v \circ G^\tau) \circ \D^{-1} (F(c))\\
        &\subseteq \{F^2(c)\}\,.
    \end{align*}
    Since the right-hand side is a singleton, and the left-hand side is non-empty, we obtain an equality, and thus we get that
    \begin{align*}
         \D \circ (\sigma_v \circ G^\tau)^2 \circ \D^{-1} (c) = F^2(c)\,.
    \end{align*}
    One can easily finish the proof by induction on $n$.
\end{proof}

As a special case, Lemma \ref{lemma:global_sim_iterations} implies that
$$\D \circ (\sigma_v \circ G^\tau)^n \circ \E (c) = F^n(c) \quad \text{ for all } c \in S^{\Z^d}.$$
This is a weaker relationship between $F$ and $G$ that is represented in Fig.~\ref{fig:weak_global_sim} and highlights that we can compute an arbitrary iteration $F^n(c)$ using only the encoder $\E$, decoder $\D$, and iterating $G$.

\begin{figure}[htbp!]
    \centering
    \includegraphics[width=0.75\linewidth]{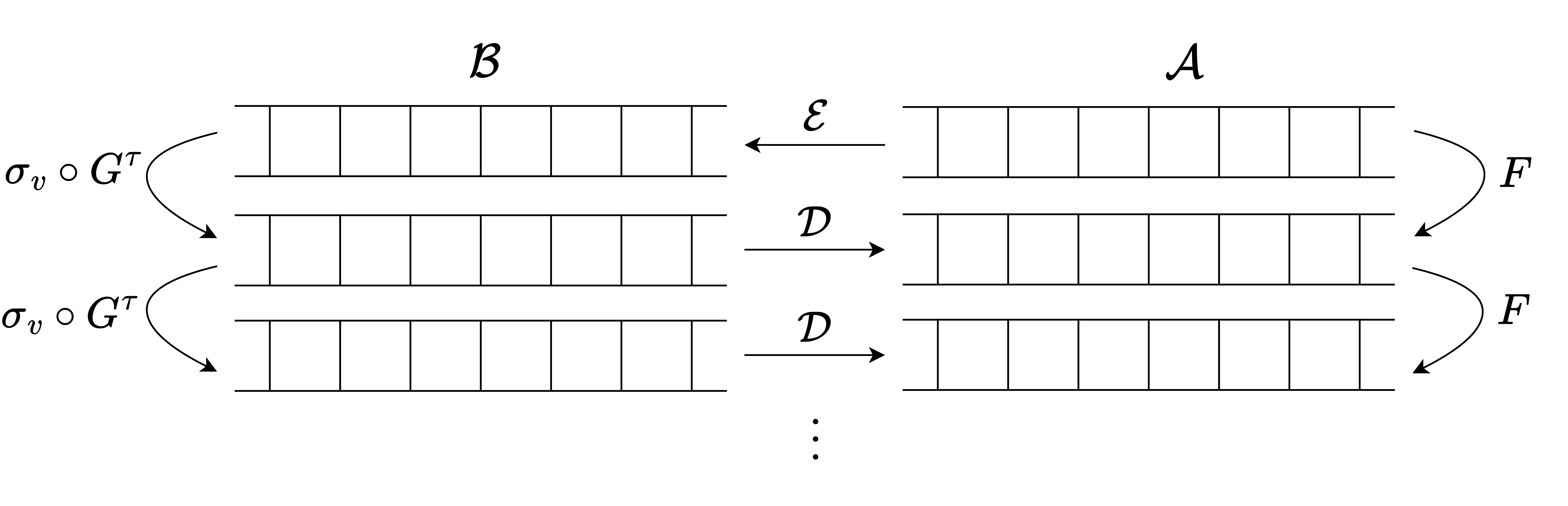}
    \caption{This diagram illustrates the relationship between $\A$ and $\B$, which holds when $\B$ globally simulates $\A$.}
    \label{fig:weak_global_sim}
\end{figure}

\begin{lemma}
    The global simulation relation $\preceq$ is a preorder.
\end{lemma}
\begin{proof}
    Let $\A = (S^{\Z^d}, F)$, $\B = (T^{\Z^d}, G)$, and $\C = (U^{\Z^d}, H)$ be $d$-dimensional cellular automata.

    \textit{Reflexivity}: It is quite clear that $\A \preceq \A$ since we can simply take the shift vector to be $v = 0$, the time-delay to be $\tau = 1$ and the encoder and decoder to be the identity mappings.

    \textit{Transitivity}: Suppose that $\A \preceq \B$ via $v \in \Z^d$, $\tau \in \N$, $\E: S^{\Z^d} \rightarrow T^{\Z^d}$ with full rank matrices $M_1, M_2$ and $\D: T^{\Z^d} \rightharpoonup S^{\Z^d}$ with a full-rank matrix $M$. And suppose that $\B \preceq \C$ via $v' \in \Z^d$, $\tau' \in \N$, $\E': T^{\Z^d} \rightarrow U^{\Z^d}$ with full rank matrices $M'_1, M'_2$ and $\D': U^{\Z^d} \rightharpoonup T^{\Z^d}$ with a full-rank matrix $M'$. Thus, we have that:
    \begin{align*}
        \D \circ \sigma_v \circ G^\tau \circ \D^{-1} &= F\\
        \D' \circ \sigma_{v'} \circ {H^\tau}' \circ \D'^{-1} &= G\,.
    \end{align*}

    Our goal is to show that $\A \preceq \C$. From Lemma \ref{lemma:global_sim_iterations}, we get that:
    \begin{align*}
        \D' \circ (\sigma_{v'} \circ H^{\tau'})^\tau \circ {\D'}^{-1} = G^\tau.
    \end{align*}
    And since $H$ commutes with any shift operator, this reduces to 
    \begin{align*}
        \D' \circ \sigma_{v'\tau} \circ H^{\tau \tau'} \circ \D'^{-1} = G^\tau.
    \end{align*}
    Using property 3 of a decoder $\D'$ and applying $\sigma_v$ to both sides from the left, we find
    \begin{align*}
        \D' \circ \sigma_{M'v+v'\tau} \circ H^{\tau \tau'} \circ \D'^{-1} = \sigma_v \circ G^\tau
    \end{align*}
    which finally yields
    \begin{align*}
       \D \circ \D' \circ \sigma_{M'v+v'\tau} \circ H^{\tau \tau'} \circ \D'^{-1} \circ \D^{-1} &= \D \circ \sigma_v \circ G^\tau \circ \D^{-1} = F.
    \end{align*}
    Thus, for the decoder $\tilde{\D} = \D \circ \D'$, shift vector $\tilde{v} = M'v+v'\tau$, and time delay $\tilde{\tau} = \tau \tau'$ it holds that $\tilde{\D} \circ \sigma_{\tilde{v}} \circ G^{\tilde{\tau}} \circ \tilde{\D}^{-1} = F.$ Naturally, we define the corresponding encoder as $\tilde{\E} = \E' \circ \E$. Lastly, we have to check that $\tilde{\E}$ and $\tilde{\D}$ satisfy all conditions of Definition \ref{def:global_simulation}. It is clear that both $\tilde{\E}$ and $\tilde{D}$ are continuous and that $\tilde{\D} \circ \tilde{\E} = \id_{S^{\Z^d}}$. Next, we show that $\tilde{E}$  commutes with certain shifts determined by a pair of full rank matrices. This would clearly be true if $M_2\cdot w \in \rng M_1'$ for all $w \in \Z^d$. This does not hold in general, but one can notice that  $\det(M'_1) M_2\cdot w \in \rng M_1'$ for all $w \in \Z^d$. Let $A$ be such that $M'_1 A = \det(M'_1) M_2$. Then, for any $w \in \Z^d$:
    \begin{align*}
        \tilde{\E} \circ \sigma_{\det(M'_1)M_1w} &= \E' \circ \E  \circ \sigma_{\det(M'_1)M_1w} = \E' \circ \sigma_{\det(M'_1)M_2w} \circ \E\\
        &= \E' \circ \sigma_{M'_1 Aw} \circ \E = \sigma_{M'_2Aw} \circ \E' \circ \E = \sigma_{M'_2Aw} \circ \tilde{\E}.
    \end{align*}
    Therefore, $\tilde{\mathcal{E}}$ satisfies condition 3 of an encoder for matrices $\tilde{M}_1 = \det(M'_1)M_1$ and $\tilde{M}_2 = M'_2A$. Verifying condition 3 for the decoder is quite straightforward. Lastly, we have to check that $\dom(\tilde{\D})$ is a subshift of finite type:
    $$\dom(\D \circ \D') = \{x \in \dom(\D') \mid \D'(x) \in \dom(\D) \} = \dom(\D') \cap \D'^{-1}(\dom(\D)).$$
    It is easy to check that the intersection of two SFTs is again a SFT. Furthermore, from the next section, it will be apparent that since $\D'$ is continuous and commutes with certain shifts, it can be viewed as a ``cellular transformaton'', and thus $\D'^{-1}$ maps SFTs to SFTs. 
\end{proof}

\subsubsection{Packing Operator}
Packings are a family of geometrical transformations of CA configurations, introduced in full generality in \cite{bulking1}, that form a key ingredient of any notion of CA global simulation.  Intuitively, given a pattern $\mathcal{P} \subset \Z^d$ of size $k$ which induces a tiling of $\Z^d$, the packing operator packs together the cells in each tile into a larger ``supercell''; transforming configurations from $S^{\Z^d}$ to $(S^k)^{\Z^d}$.

Formally, we call any finite subset $\mathcal{P} \subset \Z^d$ a \emph{pattern}. Given a vector $v \in \Z^d$, the translation of $\mathcal{P}$ by $v$ is $\mathcal{P}+v = \{p+v \mid p \in \mathcal{P} \}$. Given a sequence $V = (v_1, \ldots, v_d)$ of $d$ linearly independent vectors from $\Z^d$, we write $M_V = (v_1 | v_2| \cdots | v_d) \in \Z^{d \times d}$ to denote the matrix whose columns are given by $V$. We call $(\mathcal{P}, V)$ \emph{a tiling} of $\Z^d$ if the set $\{ \mathcal{P} + M_V \cdot z \mid z \in \Z^d \}$ is a partition of $\Z^d$. We can notice that $V$ induces an equivalence relation on $\Z^d$ defined by $z \equiv_V z'$ if and only if $z - z' \in M_v \cdot x$ for some $x \in \Z^d$. Clearly, $\equiv_V$ has only finitely many equivalence classes, their number equal to $\det(M_v)$. We can observe that $(\mathcal{P}, V)$ is a tiling of $\Z^d$ if and only if $\mathcal{P}$ contains exactly one element from each equivalence class induced by $V$. An example of a tiling of $\Z^2$ is illustrated in Fig.~\ref{fig:tiling}.

\begin{figure}[htpb!]
    \centering
    \includegraphics[width=\linewidth]{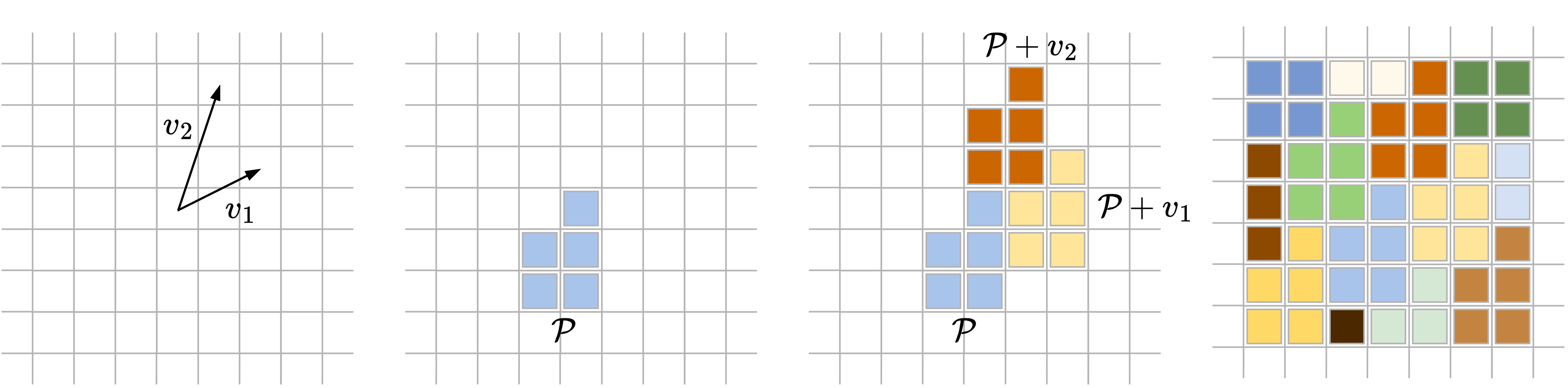}
    \caption{(left) Vectors $v_1 = (2, 1)$, $v_2 = (1, 3)$ spanning a sublattice in $\Z^2$. (middle-left) The pattern $P = \{(0,0), (1, 0), (0,1), (1, 1), (1, 2) \}$. (middle-right) Translation of $\mathcal{P}$ by $v_1$ and by $v_2$. (right) Tiling of $\Z^2$ induced by $(\mathcal{P}, (v_1, v_2)).$}
    \label{fig:tiling}
\end{figure}

Let $(\mathcal{P}, V)$ be a tiling of $\Z^d$ with $|\mathcal{P}| = m$, and let us fix the order of the elements of $\mathcal{P}$; $\mathcal{P} = (p_1, \ldots, p_m)$. The tiling $(\mathcal{P}, V)$ induces a \emph{packing operator} $o^{\langle \mathcal{P}, V \rangle}: S^{\Z^d} \rightarrow (S^m)^{\Z^d}$ for any finite set $S$, defined as follows:
 
$$o^{\langle \mathcal{P}, V \rangle}(c)(z) = \big(c(p_1+M_v\cdot z), c(p_2+M_v\cdot z), \ldots, c(p_m+M_v\cdot z)\big) \in S^m, \text{for } c \in S^{\Z^d} \text{ and } z \in \Z^d.$$
\noindent The packing operator is illustrated in Fig.~\ref{fig:packing_operator}.
 
\begin{figure}[htpb!]
    \centering
    \includegraphics[width=.8\linewidth]{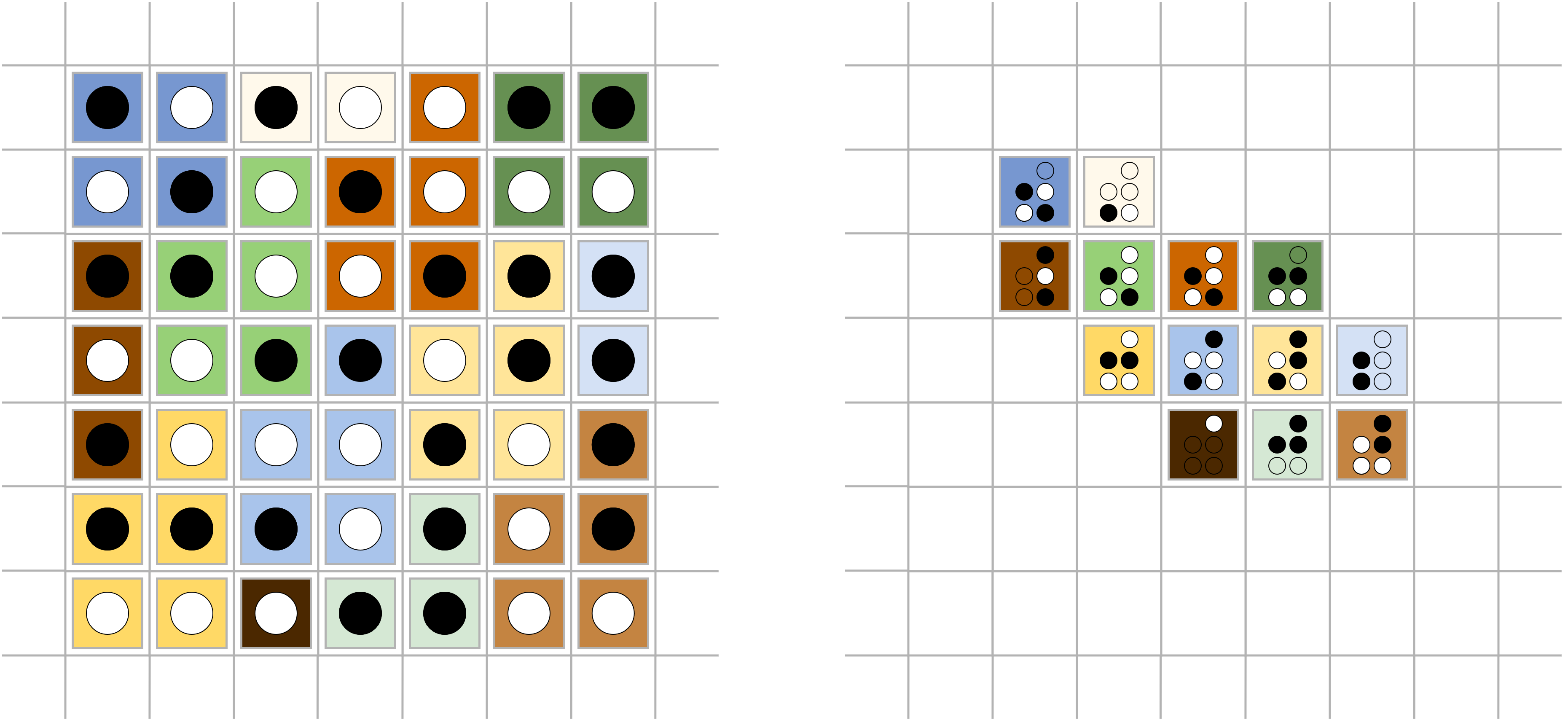}
    \caption{The tiling $(\mathcal{P}, V)$ from Fig.~\ref{fig:tiling} induces a packing operator. This figure illustrates the case of $o^{\langle \mathcal{P}, V \rangle}: \2^{\Z^2} \rightarrow (\2^5)^{\Z^2}$. (Left) Depiction of a configuration $c \in \2^{\Z^2}$. White and black circles illustrate the states 0 and 1; we keep the background coloring to show the space tiling. (Right) We display the outcome $o^{\langle \mathcal{P}, V \rangle}(c)$. The states are aggregated by groups of 5.}
    \label{fig:packing_operator}
\end{figure}

Given a cellular automaton $\B = (T^{\Z^d}, G)$ and a packing $(\mathcal{P}, V)$ of $\Z^d$ with $|\mathcal{P}|=m$, we can transform $\B$ into $\B^{\langle \mathcal{P}, V \rangle} = \left((T^m)^{\Z^d}, G^{\langle \mathcal{P}, V \rangle} \right)$ where $ G^{\langle \mathcal{P}, V \rangle} =o^{\langle \mathcal{P}, V \rangle} \circ  G \circ \left(o^{\langle \mathcal{P}, V \rangle}\right)^{-1}.$ This provides a natural way to view a given CA as operating on a larger state space. Below, we discuss that a decoder and encoder witnessing a global simulation $\A \preceq \B$ are essentially CAs operating on a packed grid space. Since their input and output state spaces differ, we generalize the notion of CAs below.

\begin{definition}[Cellular transformaton]
    Let $S$ and $T$ be finite sets and $d \in \N$. We say that $F: S^{\Z^{d}} \rightarrow T^{\Z^{d}}$ is a \emph{cellular transformaton} if there exists a neighborhood $(n_1, n_2, \ldots, n_k) \subseteq \left(\Z^{d}\right)^k$ and a local rule $f: S^k \rightarrow T$ such that for each $v \in \Z^{d}$ and for each $c \in S^{\Z^{d}}$:
    $$F(c)_v = f\left(c_{v+n_1}, \ldots, c_{v+n_k}\right).$$
\end{definition}

Cellular transformata have been studied in the context of symbolic dynamics under the name of \emph{sliding block codes} \cite{goncavales_sliding_block_codes, lind2021introduction}. Following the exact same proof strategy as in the Curtis-Hedlund-Lyndon Theorem \ref{thm:og_curtis_hedlund_lyndon}, we obtain an analogous result for cellular transformata:
\begin{lemma}[Curtis–Hedlund–Lyndon theorem for cellular transformata]\label{lemma:curtis_hedlund_lyndon_for_ct}
    Let $S, T$ be finite sets and $d \in \N$. Let $F: S^{\Z^d} \rightarrow T^{\Z^d}$. Then the following conditions are equivalent:
\begin{enumerate}[label=(\roman*)]
    \item $F$ is continuous and commutes with all shift operators.
    \item $F$ is a cellular transformaton.
\end{enumerate}
\end{lemma}

 Now we can proceed with the characterization of the encoders and decoders, using a simple generalization of the Curtis-Hedlund-Lyndon Theorem \ref{thm:og_curtis_hedlund_lyndon}.

\begin{theorem}[Packing version of the Curtis–Hedlund–Lyndon theorem for cellular transformata] \label{thm:packed_curtis_hedlund_lyndon} Let $S, T$ be finite sets and $V = (v_1, \ldots, v_d)$, $W = (w_1, \ldots, w_d)$ two sequences of $d$ linearly independent vectors in $\Z^d$; we denote their corresponding matrices by $M_V$ and $M_W$. Let $(\mathcal{P}, V)$ and $(\mathcal{Q}, W)$ be two packings of $\Z^d$. Then, the following two conditions are equivalent:
\begin{enumerate}[label=(\roman*)]
    \item $F$ is continuous and for each $v \in \Z^d$ satisfies $F \circ \sigma_{M_V \cdot v} = \sigma_{M_W \cdot v} \circ F$.
    \item $G = o^{\langle \mathcal{Q}, W \rangle} \circ F \circ \left(o^{\langle \mathcal{P}, V \rangle}\right)^{-1}$ is a cellular transformaton.
\end{enumerate}
\end{theorem}

\begin{proof} Here we provide a sketch of the proof. $(i) \implies (ii):$ It suffices to verify that $G$ is continuous and commutes with all shifts on the ``packed space''. The continuity is clear since all packing operators are continuous. Let $v \in \Z^d$. Then:
    \begin{align*}
        G \circ \sigma_v &= o^{\langle \mathcal{Q}, W \rangle} \circ F \circ \left(o^{\langle \mathcal{P}, V \rangle}\right)^{-1} \circ \sigma_v\\
        & =  o^{\langle \mathcal{Q}, W \rangle} \circ F \circ \sigma_{M_V \cdot v} \circ \left(o^{\langle \mathcal{P}, V \rangle}\right)^{-1}\\
        & =  o^{\langle \mathcal{Q}, W \rangle} \circ \sigma_{M_W \cdot v}  \circ F \circ  \left(o^{\langle \mathcal{P}, V \rangle}\right)^{-1}\\
        & = \sigma_v  o^{\langle \mathcal{Q}, W \rangle} \circ F \circ \left(o^{\langle \mathcal{P}, V \rangle}\right)^{-1} = \sigma_v \circ G.
    \end{align*}
    Thus, using Lemma \ref{lemma:curtis_hedlund_lyndon_for_ct}, $G: (S^{|\mathcal{P}|})^{\Z^d} \rightarrow (T^{|\mathcal{Q}|})^{\Z^d}$ is a cellular transformaton.\\
    \noindent $(ii) \implies (i):$ Again, we use Lemma \ref{lemma:curtis_hedlund_lyndon_for_ct} to see that $G$ commutes with all shifts and is continuous. Utilizing similar arguments as above, we get $(i)$.
\end{proof}

Hence, we can interpret each decoder witnessing a CA global simulation $\A \preceq \B$ as an operator that packs the configurations of $\B$ and transforms them into configurations of $\A$. Similarly, $\E$ packs configurations of $\A$, transforms them into packed configurations of $\B$, and unpacks those. We illustrate this by the following example, again coming from the work of Culik and Yu \cite{simple_universal_ca}. We say a CA is totalistic if its local rule is invariant to any permutation of its inputs.

\begin{example}[Totalistic simulation \cite{simple_universal_ca}]
    For each 1D cellular automaton $\A = (S^\Z, F)$ with radius $r=1$ there exists a totalistic CA $\B$ also with radius $1$ that globally simulates it.

\begin{proof}
    We simply encode the relative positions of cells into the states of $\B$ by the encoder. Suppose $\A$ has states $\{s_1, s_2, \ldots, s_n\}$ that we identify with numbers ${1, \ldots, n}$ and let $B = n+1$. The encoding is based on using the numbers 10, 100, and 1000 in $B$-ary representation, as seen in Fig.~\ref{fig:totalistic_simulation}.

 \begin{figure}[htpb!]
    \centering
    \includegraphics[width=0.85\linewidth]{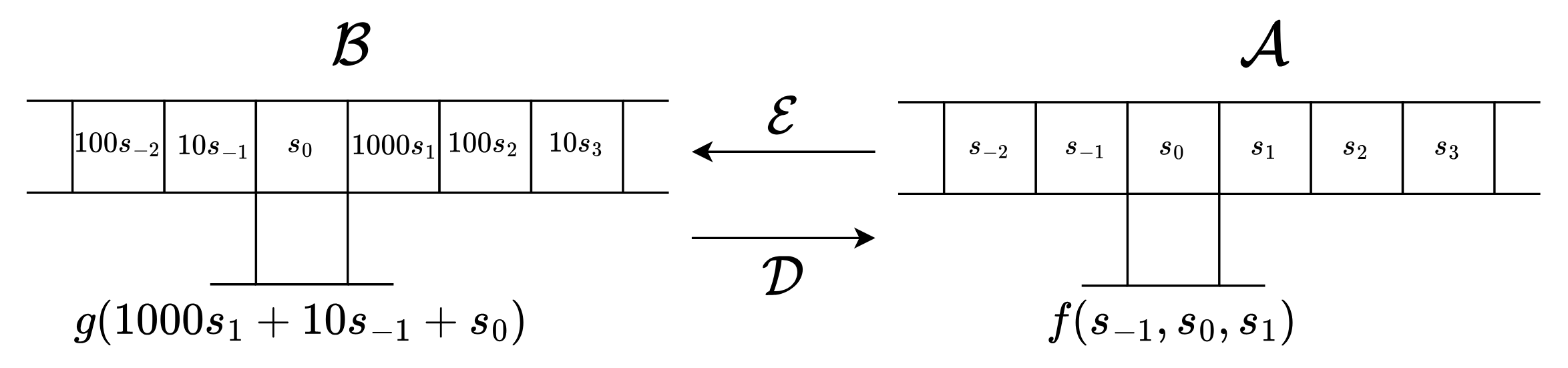}
    \caption{Encoded configuration of a totalistic CA simulating a CA with radius $r=1$.}
    \label{fig:totalistic_simulation}
\end{figure}
    It is clear that the sum of any three consecutive cells will contain exactly one 0 that determines the position of the three states in the neighborhood. This definition shows that the encoder requires information about the positions of the cells. The decoder simply ``forgets'' this additional information and sends each 4-tuple of cells $s_0,1000s_1,100s_2,10s_3 \mapsto s_0s_1s_2s_3$. 
    
    One can easily see that the encoder packs the 1D grid into blocks of four consecutive cells, and transforms them via a cellular transformaton (with radius one), to unpack them again. Formally, $\E$ satisfies for any $v \in \Z$ that $\E \circ \sigma_{4v} = \sigma_{4v} \circ \E$. The domain of decoder $\D$ contains exactly all valid encodings of configurations, i.e.~$\dom(\D) = \E(S^\Z)$. $\D$ does not need to pack the configuration space of $\B$ and can be simply implemented as a cellular transformaton with radius $1$ which forgets the positional encoding. 
    
    Since for each $c \in S^\Z$, $\D^{-1}(c) = \{\E(c)\}$, it is easy to see that $\D \circ G \circ \D^{-1} = F$.
\end{proof}
\end{example}

It is now an easy exercise to return to the examples we gave at the beginning, and to show they ``comply'' with Definition \ref{def:global_simulation} of global simulation.

\subsection{Related work}\label{appD:related_word}
Different variants of global simulation and their properties have been studied in the past; typically under the name of \emph{intrinsic simulation}. There are at least three distinct lines of work, represented by~\cite{inducing_order_on_ca_by_grouping}, \cite{bulking1}, and \cite{hudcova2024simulation} which study notions that are similar in spirit, yet differ in the details. As there is no comparison of these works in the literature, we fill this gap by providing a detailed overview, summarized in Table \ref{tab:global_sim_related_works}.

For this whole section, we consider $\B = (T^{\Z^d}, G)$ and $\A = (S^{\Z^d}, F)$ to be cellular automata in dimension $d \in \N$. All of the variants of CA simulation can be formulated in the same general framework:

$\B$ simulates $\A$ if there exists a certain transformation $\tilde{G}$ of $G$, a transformation $\tilde{F}$ of $F$ and a decoder $\D$ some computationally feasible partial mapping translating between configurations of $\B$ and $\A$ such that
$$\D \circ \tilde{G} \circ \D^{-1}(c) = \tilde{F}(c) \quad \quad \text{for all } c \in S^{\Z^d}.$$

\paragraph{Global simulation}
For global simulation, we don't allow for transformations of $F$, thus $\tilde{F} = F$. For $G$, we have $\tilde{G} =  o^{\langle \mathcal{P}, V \rangle} \circ \sigma_v \circ G^\tau \circ \left(o^{\langle \mathcal{P}, V \rangle} \right)^{-1}$ for some shift vector $v \in \Z^d$, some time-delay $\tau \in \N$, and some tiling $(\mathcal{P}, V)$ of $\Z^d$.

The decoder is a cellular transformaton $\D: (T^m)^{\Z^d} \rightharpoonup S^{\Z^d}$ whose domain is a subshift of finite type and such that we can efficiently find preimages of $\D$ via an injective cellular transformaton $\E$.

\paragraph{Grouping}
In their seminal work \cite{inducing_order_on_ca_by_grouping}, Mazoyer and Rapaport study the properties of what they call a grouping operation on 1D CAs. In our notation, $\tilde{G} =  o^{\langle \mathcal{P}, V \rangle} \circ G^m \circ \left(o^{\langle \mathcal{P}, V \rangle} \right)^{-1}$
for a simple tiling $(\mathcal{P}, V)$ where $\mathcal{P}$ is a rectangle of size $m$ and $V = (m)$. Similarly, they allow to transform $F$ into $\tilde{F} = o^{\langle \mathcal{Q}, W \rangle} \circ F^n \circ \left(o^{\langle \mathcal{Q}, W \rangle} \right)^{-1}$ again with a rectangular tiling $(\mathcal{Q}, W)$ of size $n$. Notice that here, the shift vector is $0$ and the time-delay is exactly the size of the tiling. This enforces the CAs to process information ``in real time''.

The decoder $\D$ is an injective cellular transformaton $\D: (T^m)^\Z \rightharpoonup (S^n)^\Z$ with radius 0 and a domain $(T')^\Z$ where $T' \subseteq T^m$. The authors obtain deep results about their notion of CA simulation in dimension 1 which include the following: 
 
\begin{itemize}
    \item Assuming the real-time simulation, i.e.~time-delay of $\B$ is equal to the tiling's size, the CA simulation admits no universal CA. 
    \item There exists a countable number of incomparable CA classes with respect to the relation induced by the simulation \cite{additive_cas_over_Zp_bottom}.
    \item The general problem of whether $\B$ simulates $\A$ is undecidable.
\end{itemize}

\paragraph{Bulking}
In the subsequent series of very relevant works $\cite{bulking1}$ and \cite{bulking2} by Delorme, Mazoyer, Ollinger, and Theyssier, the authors consider a generalization of the grouping operation which they call bulking. For $\B$ and $\A$ in arbitrary dimension $d \in \N$ they consider $\tilde{G} =  o^{\langle \mathcal{P}, V \rangle} \circ \sigma_v \circ G^\tau \circ \left(o^{\langle \mathcal{P}, V \rangle} \right)^{-1}$ where $v \in \Z^d$, $\tau \in \N$ and $(\mathcal{P}, V)$ is a rectangular tiling of $\Z^d$; i.e.~$\mathcal{P}$ is a rectangle, and V consists of canonical vectors, multiplied by scalars and possibly permuted. Similarly, they consider $\tilde{F} =  o^{\langle \mathcal{Q}, W \rangle} \circ \sigma_w \circ F^{\tau'} \circ \left(o^{\langle \mathcal{Q}, W \rangle} \right)^{-1}$ where $(\mathcal{Q}, W)$ is again a rectangular packing of $\Z^d$. Lastly, in \cite{bulking2}, the most general decoder the authors consider is a cellular transformaton $\D: {T^{|\mathcal{P}|}}^{\Z^d} \rightharpoonup {S^{|\mathcal{Q}|}}^{\Z^d}$ with radius 0 whose domain is a again a set $(T')^{\Z^d}$ where $T' \subseteq T^{|\mathcal{P}|}$. Among other results, in \cite{bulking1}, the authors importantly show that any geometrical transformation that maps CA space-time diagrams into CA space-time diagrams while preserving time causality is a composition of a (general) packing, time-delay, and a shift. For more details, see also the thesis of Ollinger \cite{ollinger2002automates} or Theyssier \cite{theyssier2005automates}.

\paragraph{Algebraic simulation}
In \cite{hudcova2024simulation}, Hudcová and Krásenský formulate the notion of CA simulation in algebraic language which helps them use some deeper results from universal algebra to show that certain CA classes are limited in terms of what they can simulate. Concretely, they focus on 1D CAs and consider $\tilde{F} = F$ and $\tilde{G} = \tilde{G}_1 \times \tilde{G}_2 \times \cdots \tilde{G}_k$ where $\times$ denotes the standard CA product, and each $\tilde{G}_i = o^{\langle \mathcal{P}_i, V_i \rangle} \circ G^{m_i} \circ \left(o^{\langle \mathcal{P}_i, V_i \rangle} \right)^{-1}$ where $(\mathcal{P}_i, V_i)$ is a rectangular packing of size $m_i$. The authors consider the decoder to be a cellular transformaton with radius 0 defined on a subshift of finite type. 

The introduction of products has an algebraic motivation since for every CA $\B$ the set $\{\A \mid \B \text{ can simulate } \A \}$ has a nice structure: it forms a pseudovariety. The authors show that (almost) any affine CA (local rule is an affine mapping of vector spaces) or Abelian CA (local rule is a homomorphism of an Abelian group) can only simulate CAs with the same property, thus, in particular they cannot be globally universal. A widely studied affine CA is for instance the elementary CA Rule 90. More details can be found in Hudcová's thesis \cite{hudcova2024thesis}.

We present an overview of the variety of related notions in Table \ref{tab:global_sim_related_works}.

\begin{table*}[htbp!]
   \small
   \centering
   \renewcommand{\arraystretch}{1.3} 
   \begin{tabular}{>{\centering\arraybackslash}p{2.2cm} || >{\centering\arraybackslash}p{1.2cm} | >{\centering\arraybackslash}p{4.5cm} | >{\centering\arraybackslash}p{4.5cm} | >{\centering\arraybackslash}p{2.5cm}} 
    \toprule
   \toprule
    \multicolumn{5}{c}{\large\textbf{Global Simulation Framework}} \\
    \multicolumn{5}{c}{\normalsize $\mathcal{B} = (T^{\mathbb{Z}^d}, G)$ simulates $\mathcal{A} = (S^{\mathbb{Z}^d}, F)$ if $\mathcal{D} \circ \tilde{G} \circ \mathcal{D}^{-1} = \tilde{F}$ for some decoder $\mathcal{D}$}\\
    \midrule
    \midrule
\textbf{Sim. Type} & \textbf{Dim.} & $\mathbf{\tilde{G}}$ & $\mathbf{\tilde{F}}$ & \textbf{decoder} $\bm{\D}$ \\  
\hline
\hline
grouping \cite{inducing_order_on_ca_by_grouping}
& $d = 1$ & $o^{\langle \mathcal{P}, V \rangle} \circ G^{|\mathcal{P}|} \circ \left(o^{\langle \mathcal{P}, V \rangle} \right)^{-1}$, $(\mathcal{P}, V)$ a rectangular tiling & $o^{\langle \mathcal{Q}, W \rangle} \circ F^{|\mathcal{Q}|} \circ \left(o^{\langle \mathcal{Q}, W \rangle} \right)^{-1}$, $(\mathcal{Q}, W)$ a rectangular tiling & injective CT* with radius 0 \\
\hline 
bulking \cite{bulking1, bulking2}
& $d \in \N$ & $o^{\langle \mathcal{P}, V \rangle} \circ \sigma_v \circ G^\tau \circ \left(o^{\langle \mathcal{P}, V \rangle} \right)^{-1}$, $(\mathcal{P}, V)$ a rectangular tiling & $o^{\langle \mathcal{Q}, W \rangle} \circ \sigma_w \circ F^{\tau'} \circ \left(o^{\langle \mathcal{Q}, W \rangle} \right)^{-1}$, $(\mathcal{Q}, W)$ a rectangular tiling &  \multirow{2}{*}{\shortstack[l]{CT* \\ with radius 0}}\\ 
\hline 
algebraic \cite{hudcova2024simulation}
& $d=1$ & $\tilde{G}_1 \times \tilde{G}_2 \times \cdots \times \tilde{G}_k,$ each $\tilde{G}_i$ as in grouping & $F$ & \multirow{2}{*}{\shortstack[l]{CT*\\ with radius 0}}\\
\hline
global & $d \in \N$ & $o^{\langle \mathcal{P}, V \rangle} \circ \sigma_v \circ G^\tau \circ \left(o^{\langle \mathcal{P}, V \rangle} \right)^{-1}$, $(\mathcal{P}, V)$ a general tiling & $F$ &  \multirow{2}{*}{\shortstack[l]{CT*\\ with radius r}}\\
\hline
\end{tabular}
   \caption{\centering Global Simulation: Overview of Related Works.\\ *Here, CT stands for a cellular transformaton with a local rule being a partial mapping.} 
   \label{tab:global_sim_related_works}
\end{table*}

\section{Local simulation}
\label{appsec_local_simulation}

Local simulation refers to the ability of a cellular automaton to mimic the computation of a given Turing machine. Locality refers to the fact that the Turing machine's tape is, at any given moment, just a finite set of symbols, which is contained in a bounded region of the CA's configuration via a local encoding mapping. The notion of a CA simulating a Turing machine is simple to introduce conceptually, yet in order to obtain a concrete formal definition, a myriad of technical details need to be specified. Precisely specifying such details appears to be more challenging than what could initially be expected; as Ollinger writes in his review Universalities in Cellular Automata \cite{universalities_in_cas}: ``... (this) leads to the problem of heterogeneous simulation: classical models of computation have inputs, step function, halting condition, and output. Cellular automata have no halting condition and no output. As pointed out by Durand and Róka \cite{the_game_of_life_universality_revisited}, this leads to very tricky encoding problems: their own attempt at a Turing universality based on this criterion as encoding flaw permitting us counterintuitively to consider very simple cellular automata as universal. Turing universality of dynamical systems in general and cellular automata in particular has been further discussed by Delvenne et al. \cite{decidability_and_universality_in_symbolic_dynamical_systems} and Sutner \cite{universality_and_cellular_automata}. None of the proposed definitions is completely convincing so far, so it has been chosen on purpose not to provide the reader with yet another weak formal definition.''

It is surprising that despite the rich line of works studying Turing universality within CAs, no formal definition succeeded in grasping this notion in a satisfying manner. The goal of this section is to propose a formal definition of CA Turing completeness that does not lead to pathological examples of trivial yet Turing-complete automata, and at the same time is general enough to agree with all the constructions of Turing-complete cellular automata from the literature. Before delving into the details, we first provide the conceptual overview of local simulation.

\subsection{Defining local simulation}

\paragraph{Local simulation: 1st level of resolution}
Very roughly, we say that a CA $\A$ simulates a Turing machine $\T$ if, for every initial tape configuration $c$ of $\T$ there exists a corresponding initial configuration $c'$ of $\A$ such that the space-time diagram of $\A$ contains the snapshots of all computation steps of $\T$ ($\T$ need not be a Turing machine that always halts). Further, we require that the function mapping $c \mapsto c'$ is a computationally reasonable mapping. Similarly, we require that for every time-step $t \in \N$ we can recover $t$ steps of time evolution of $c$, denoted by $T^t(c)$, from the space-time diagram of $\A$ in a computationally feasible way.

\begin{definition}[Turing machine]
\label{def:Turing0}
    A Turing machine $\T$ is given by a triple $(Q, \Sigma, \delta)$ where:
    \begin{enumerate}
        \item $Q$ is a finite set of states, with a designated $q_0 \in Q$ the \emph{initial state} and $q_{\mathrm{halt}} \in Q$ the \emph{halting state}.
        \item $\Sigma$ is a finite set representing the \emph{tape alphabet}, with a designated blank symbol $\sqcup$
        \item $\delta: Q \times \Sigma \rightarrow Q \times \Sigma \times \{L, S, R\}$ is the \emph{transition function}.
    \end{enumerate}
We define the \emph{configuration space} of $\T$ to be the set $C_\T = \Sigma^* \times Q \times \Sigma^*$. A configuration $c \in C_\T$, $c = s_1\cdots s_{m-1} q s_m\cdots s_{n}$, defines the Turing machine's tape content $s_1 \cdots s_n \in \Sigma^n$ (with all the other symbols on the bi-infinite tape being blank), the current state to be $q \in Q$ and the Turing machine head currently reading $s_m$. Thus, we can interpret $\T$ as a function $T: C_\T \rightarrow C_\T$ as follows. Let $\delta(q, s_m) = (q', s'_m, a)$ for $a \in \{\text{\rm L}, \text{\rm S}, \text{\rm R}\}$, then

\begin{equation}
    T(c) =
    \begin{cases*}
      s_1\cdots s_{m-2} q' s_{m-1} s'_m \cdots s_{n} & if $a = \text{\rm L}$ \\
      s_1\cdots s_{m-1} q' s'_m \cdots s_{n} & if $a = \text{\rm S}$ \\
    s_1\cdots s_{m-1} s'_m q' s_{m+1} \cdots s_{n} & if $a = \text{\rm R}$. 
    \end{cases*}
  \end{equation}
\end{definition}
We will identify each Turing machine $\T$ with the dynamical system $T: C_\T \rightarrow C_\T$. 

\paragraph{Local simulation: 2nd level of resolution}
Let $\A = (S^{\Z^d}, F)$ be a $d$-dimensional cellular automaton and $\T = (C_\T, T)$ a Turing machine. We say that $\A$ simulates $\T$ if there exists:
\begin{enumerate}
    \item an \emph{encoder} $\E: C_\T \rightarrow S^{\Z^d}$
    \item a \emph{decoder} $\D: S^{\Z^d} \rightarrow C_\T$
    \item a delay function $\tau: S^{\Z^d} \rightarrow C_\T$
\end{enumerate}
Such that for all $c \in C_\T$ and for all time-steps $t \in \N$ it holds that:
\begin{equation}\label{eq:local_sim_2nd_lev}
    \D \circ (F^\tau)^t \circ \E(c) = T^t(c)
\end{equation}
where $F^\tau: S^{\Z^d} \rightarrow S^{\Z^d}$ is defined as $F^{\tau}(c) = F^{\tau(c)} (c)$ for any $c \in S^{\Z^d}$. The relationship between $\A$ and $\T$ captured in equation \eqref{eq:local_sim_2nd_lev} is illustrated in Fig.~\ref{fig:local_simulaiton_scheme}.

\begin{figure}[htbp!]
    \centering
    \includegraphics[width=0.8\linewidth]{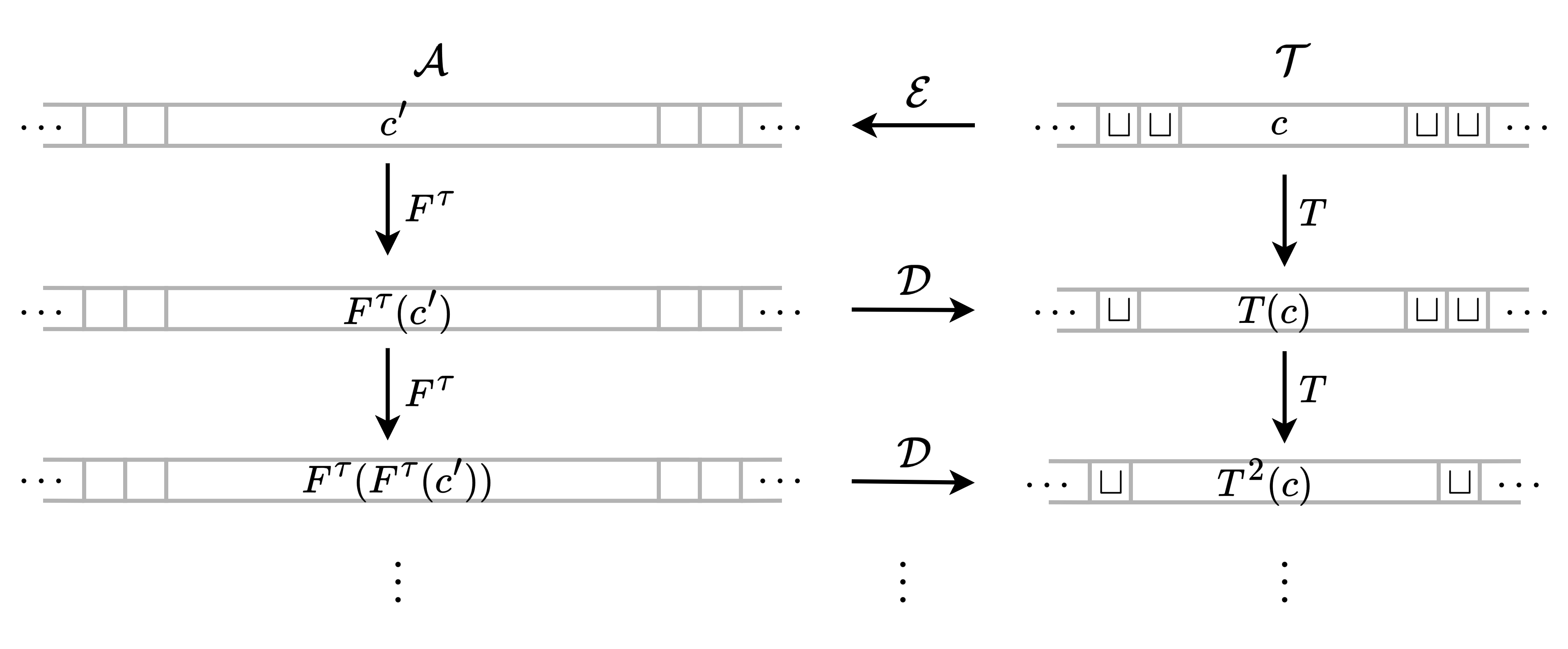}
    \caption{A general scheme illustrating the concept of local simulation: cellular automaton $\A$ simulates the Turing machine $\T$ via an encoder $\E$, decoder $\D$, and a delay function $\tau$.}
    \label{fig:local_simulaiton_scheme}
\end{figure}

Crucially, the mappings $\E, \D$, and $\tau$ have to be ``computationally reasonable''. Intuitively, we want to prevent cases where $\E$, $\D$, or $\tau$ could perform an arbitrarily sophisticated computation in a single iteration. However, specifying what ``computationally reasonable'' means exactly is where the technical challenge lies. First, the mappings relate finite configurations of $\T$ with infinite configurations of $\A$. To deal with this, it is natural to require the following. Let $c \in C_\T$ with $|c| = n$ and let $\D(c') = c$.
\begin{enumerate}
    \item $\E$ maps $c$ to a finite region of size proportional to $n$, centred around 0. The configuration outside this region is initialized with some ``admissible background'' (a natural choice is, for instance, the quiescent state if the CA has one).
    \item When decoding $c'$, $\D$ has access only to a finite region centred around 0 whose size is proportional to $n$. 
    \item Similarly to $\D$, when $\tau$ computes the time-delay for $c'$, it only has access to a finite region centred around 0 whose size is proportional to $n$.
\end{enumerate}
Thus, we have reduced the mappings $\E, \D$ and $\tau$ to ones between finite configurations. As such, we can require them to be realized by computable functions of bounded complexity (concretely, we will bound their time complexity by some computable function $\beta(n)$).

Before we proceed with the formal definition, we discuss two examples of what can go wrong, if we impose no complexity restrictions on the mappings $\D, \E$, and $\tau$.

\begin{example}[Unrestricted Decoder] 
We adapt an example from \cite{computational_dynamical_systems} to show that when the decoder's computational complexity is not restricted, even a trivial CA can simulate a universal Turing machine.

Let $\T_{\mathrm{univ}} = (C_{\T_{\mathrm{univ}}}, T_{\mathrm{univ}})$ be a universal Turing machine given by $(Q, \Sigma, \delta)$. We define a trivial 1D CA $\A$ with states $S = Q \cup \Sigma \cup \{0, 1\}$ and radius 1. The local rule $f: S^3 \rightarrow S$ acts as the identity of the central cell except for the following three cases:
$$f(0q_0\sqcup) = 1, \quad f(1q_0\sqcup) = 1, \quad  f(q_0\sqcup\sqcup) = q_0.$$
Let $c \in C_{\T_{\mathrm{univ}}}$ with length $n$, we define $\restr{\E(c)}{[0, n+2]} = c \, 0 \,q_0$ and to be $\sqcup$ everywhere else. We define $\tau$ to be a constant 1, and we define the decoder 
$$\D(\cdots\sqcup\sqcup c \, 0 \underbrace{11\cdots 1}_{k} q_0 \sqcup\sqcup \cdots) \coloneqq T^k_\mathrm{univ}(c).$$ 

Now it is easy to see that if we iterate $\A$ on $\E(c)$, it keeps all the information about $c$ and $q_0$ intact and simply tracks the number of iterations by a growing sequence of 1s. Therefore, we see that $\A$ simulates $\T_{\mathrm{univ}}$ as it satisfies equation \eqref{eq:local_sim_2nd_lev}.  
\end{example}

\begin{example}[Unrestricted Encoder]
We adapt an example from \cite{the_game_of_life_universality_revisited} to show that if we impose no complexity restrictions on the encoder and delay function, even a trivial CA can simulate a universal Turing machine.

Let $\T_{\mathrm{univ}} = (C_{\T_{\mathrm{univ}}}, T_{\mathrm{univ}})$ be a universal Turing machine given by $(Q, \Sigma, \delta)$. We define a trivial 1D CA $\A$ with states $S = Q \cup \Sigma$ and radius 1, which simply shifts every configuration by one cell to the left. Let $c \in C_{\T_{\mathrm{univ}}}$, we define the encoder to already ``precompute'' all iterations of $\T_{\mathrm{univ}}$ on $c$:
$$\E(c) \coloneqq \cdots \sqcup \sqcup . c \sqcup T_{\mathrm{univ}}(c) \sqcup T^2_{\mathrm{univ}}(c) \sqcup \cdots$$
where $c$ starts at position 0. Whenever we have a CA configuration of the form $\cdots \sqcup \sqcup . u_1 \sqcup u_2 \sqcup \cdots$ for $u_1, u_2, \ldots \in S^*$ we define $\tau(\cdots \sqcup \sqcup . u_1 \sqcup u_2 \sqcup \cdots) = |u_1|+1$. Lastly, we define the decoder $\D(\cdots \sqcup \sqcup . u_1 \sqcup u_2 \sqcup \cdots) = u_1.$ Again, it is easy to check that $\A$ simulates $\T_{\mathrm{univ}}$ via $\E, \D$, and $\tau$ while all the ``computation'' is hidden in the encoder.
\end{example}

As we will see when discussing the Turing completeness of Rule 110, in order for it to comply with our definition of local universality, we have to further generalize the form of the CA's background (quiescent state is too restrictive) and adjust equation \eqref{eq:local_sim_2nd_lev} by appropriate shifts, since the information in the CA's space-time diagram can, in general, be shifted around by the dynamics.   

\paragraph{Local simulation: formal definition}

\begin{definition}[Admissible background] \label{def:admissible_background}
    Let $\A = (S^{\Z^d}, F)$ be a CA. We say that $B \in S^{\Z^d}$ is a \emph{basic admissible background} if there exists a packing $(\mathcal{P}, V)$ of $\Z^d$ with $|\mathcal{P}| = m$ and a cellular transformaton $G: \{0\}^{\Z^d} \rightarrow (S^m)^{\Z^d}$ such that $B = (o^{\langle \mathcal{P}, V\rangle} )^{-1} \circ G(0^{\Z^d})$. We say that $B \in S^{\Z^d}$ is an \emph{admissible background} if there exists a finite partition of $\Z^d$ by convex polyhedral regions $R_1, \ldots, R_k$ delimited by linear inequalities with rational coefficients and a sequence of basic backgrounds $B_1, \ldots, B_k$ such that for each $i \in \{1, \ldots, k\}$ it holds that $\restr{B}{R_i} = \restr{B_i}{R_i}$.
\end{definition}
An example of an admissible background in dimension 2 is illustrated in Fig.~\ref{fig:admissible_background}.

\begin{figure}[htbp!]
    \centering
    \includegraphics[width=0.9\linewidth]{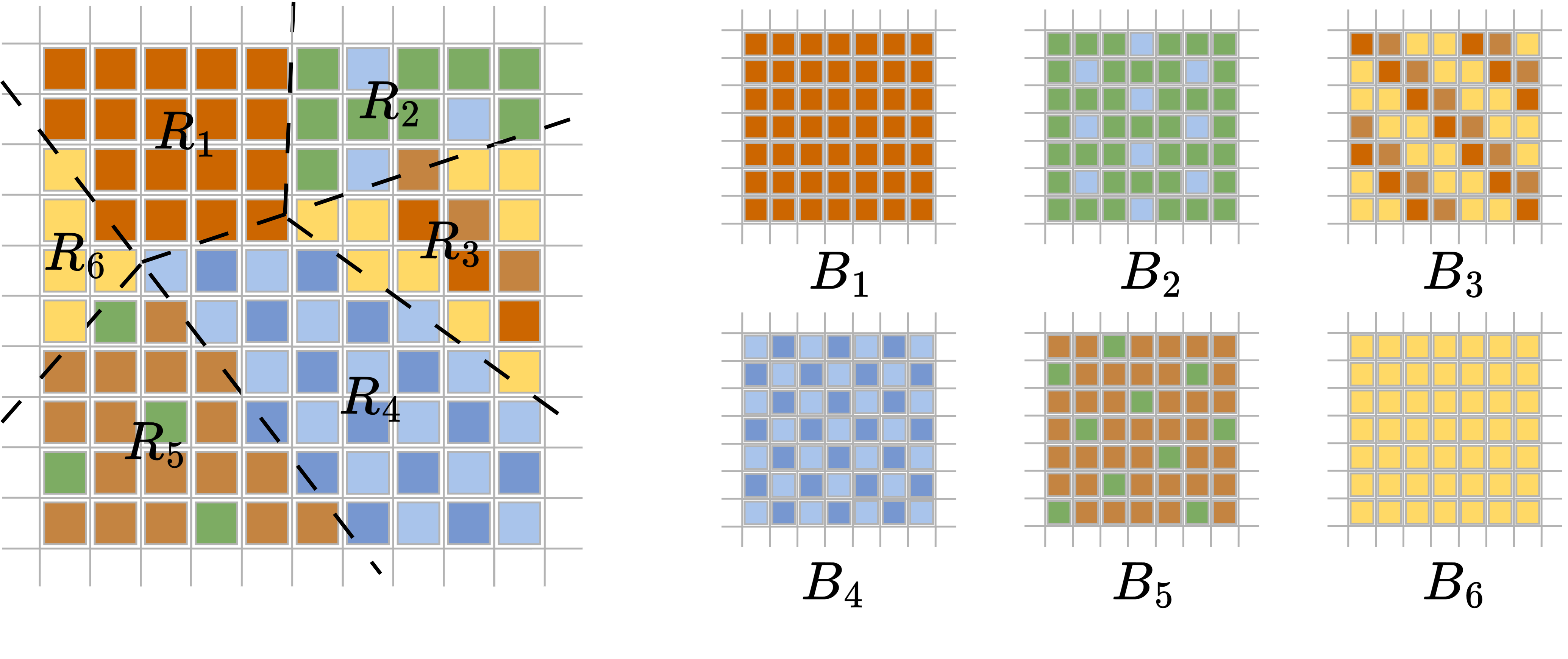}
    \caption{Example of an admissible background in $\Z^2$ delimited by six convex regions (left) and corresponding basic backgrounds (right).}
    \label{fig:admissible_background}
\end{figure}

\begin{definition}[Local simulation]\label{def:local_simulation}
    Let $\A = (S^{\Z^d}, F)$ be a CA and let $\T: C_\T \rightarrow C_\T$ be a Turing machine. We say that $\A$ simulates $\T$ if there exists a finite collection of admissible backgrounds $\{B_1, \ldots, B_k \}$, a mapping $\E: C_\T \rightarrow S^{\Z^d}$, a partial mapping $\D: S^{\Z^d} \rightharpoonup C_\T$, a partial mapping $\tau: S^{\Z^d} \rightarrow \N$, and two vectors $v, w \in \Z^d$ such that:
    \begin{equation}\label{eq:local_simulation}
        \D \circ \big(\sigma_w \circ (\sigma_v \circ F)^\tau\big)^n \circ \E = T^n \quad \quad \quad \text{for every } n \in \N.
    \end{equation}
    The mappings $\E$ (encoder), $\D$ (decoder), and $\tau$ (delay function) must satisfy the following:
    \begin{itemize}
        \item The encoder $\E: C_\T \rightarrow S^{\Z^d}$ is implemented by a Turing machine with $k+2$ tapes. The first tape contains the input configuration $c \in C_\T$. The second tape is empty, this is where the output will be written. The tapes with indices $3, \ldots, k+2$ each contain the descriptions of admissible backgrounds $\{B_1, \ldots, B_k \}$. The Turing machine has one head that can move across the tapes and halts for every input configuration. Upon halting, the output tape will contain an $R$-pattern over $S$ where $R = [-k_c, k_c]^d$ for some $k_c \in \N$ and further, the output tape contains the description of an index $i \in \{1, \ldots, k\}$ of an admissible background. The output yields a configuration $\E(c) \in  S^{\Z^d}$ defined by $\restr{\E(c)}{R} = c'$ and $\restr{\E(c)}{\Z^d \setminus R} = B_i$. 
        \item The decoder $\D: S^{\Z^d} \rightharpoonup C_\T$ is implemented by a Turing machine with two tapes -- the input and output tape. The input tape contains the configuration $c' \in  S^{\Z^d}$ with the head positioned at $0 \in \Z^d$. Upon halting, $\D$ outputs on the second tape a configuration $c \in C_\T$. If $c'$ is not in the domain of $\D$, then the Turing machine need not halt.
        \item The delay function $\tau$ has the same domain as $\D$. It is also implemented by a Turing machine with two tapes -- the input and output tape. Upon halting, the output tape contains the binary representation of the integer $\tau(c') \in \N$.
    \end{itemize}
    Moreover, we require that there exists a computable function $\beta: \N \rightarrow \N$ such that $\E(c)$ runs in time $O\big(\beta(|c|)\big)$, and $\D(c')$ and $\tau(c')$ run in time $O\big(\beta(|\D(c)|)\big)$. In such a case, we say that the simulation runs in time $O(\beta(n))$.
\end{definition}

We notice that due to the time-complexity bound on the encoder, for the central pattern encoding $c$, it holds that $k_c \leq \beta(n)$. Similarly, the time-complexity bound on the decoder implies that $\D$ only depends on a finite region of the input $c'$ whose size is also bounded by $\beta(n)$. 

\begin{definition}[Local universality]
    We say that a CA is \emph{locally universal} if it can locally simulate a universal Turing machine.
\end{definition}

\begin{remark}
    We remark that there might appear future constructions of Turing-complete CAs that do not agree with Definition \ref{def:local_simulation} exactly. For instance, the notion of an admissible background could be relaxed to allow for more general geometrical partitions and patterns. Alternatively, one may require a CA $\A$ to only simulate every $k$-th step of a Turing machine rather than every step exactly. However, we argue that Definition \ref{def:local_simulation} can readily be generalized to account for such examples without having to undergo a fundamental change. 
\end{remark}

One important generalization we will use later is that of \emph{fractional shifts}: for a 1D system and $p, q \in \N$, let $\sigma_{\frac{p}{q}}$ denotes a mapping that every $q$ steps shifts a given configuration by $p$ cells to the left. This notion can be generalized to higher dimensions in a straightforward way.

Verifying that Definition \ref{def:local_simulation} is compatible with every single construction of a Turing-complete CA from the literature is clearly outside of the scope of this paper. In what follows, we focus on one of the most elaborate works: the proof of the Turing completeness of Rule 110 \cite{cook_og} -- a CA whose local rule is so compact that in order to show its computational universality, the capacities of the corresponding encoder and decoder have been stretched much further compared to other constructions. Showing directly that Rule 110 can simulate a universal Turing machine would be quite cumbersome and instead, the proof in \cite{cook_og} follows a series of reductions, arriving at a much simpler, yet universal model of computation called a cyclic tag system. Before we proceed with a detailed discussion of Rule 110's Turing completeness, we introduce a series of computational systems and discuss their universality.

\subsection{Computational systems and their simulation}
We briefly introduce all ``intermediate'' computational systems used in the proof of Rule 110's universality, and we define the notion of one computational system simulating another.

\begin{definition}[Clockwise Turing machine]
A clockwise Turing machine is like a standard single-tape Turing machine except for the following details: 
\begin{enumerate}[label=(\roman*)]
    \item the tape is finite and circular;
    \item the head moves only clockwise on the tape;
    \item before moving to the next cell, the head can either write a single symbol into the current cell, or split the current cell in two and write two new symbols (thus prolonging the tape).
\end{enumerate}
Analogously to regular Turing machines, for each clockwise Turing machine $\T$ we define its configuration space $C_\T$ and associate $\T$ with a dynamical system $(C_\T, T)$. 
\end{definition}

\begin{definition}[Tag system]\label{def:tag_system}
    A tag system is defined by a finite set of symbols $\Sigma$ and a transition function $f: \Sigma \rightarrow \Sigma^*$. Given a tape configuration $s_1 \cdots s_n$ of symbols from $\Sigma$, the tag system updates it as follows:
\begin{enumerate}[label=(\roman*)]
    \item the symbols $s_1, s_2$ are deleted from the start of the tape;
    \item the word $f(s_1)$ is appended at the end of the tape.
\end{enumerate}
The configuration space of a tag system $\T$ with alphabet $\Sigma$ is defined as $C_\T = \Sigma^*$ and the local transition function $f$ induces the map $T: \Sigma^* \rightarrow \Sigma^*$. Again, we can associate with $\T$ the dynamical system $(C_\T, T)$.
\end{definition}

\begin{definition}[Cyclic tag system]
    A cyclic tag system operates on binary configurations (with symbols $\text{\rm Y}$ and $\text{\rm N}$) and is defined by a list of appendants $(\alpha_0, \alpha_1, \ldots, \alpha_{z-1})$ with each $\alpha_i \in \{\text{\rm Y},\text{\rm N} \}^*$. Given a tape configuration $s_1, \ldots, s_n$ of symbols from $\{\text{\rm Y},\text{\rm N} \}$, the cyclic tag system at time $t$ updates it as follows:
\begin{enumerate}[label=(\roman*)]
    \item $s_1$ is removed from the start of the tape
    \item if $s_1 = \text{\rm Y}$ then the word $\alpha_{t \bmod z}$ is appended at the end of the tape, otherwise, nothing gets appended.
\end{enumerate}
After the update, time increases to $t+1$. Initially, the system starts at time $t = 0$.
The configuration space of a cyclic tag system $\T$ with appendant list $(\alpha_0, \alpha_1, \ldots, \alpha_{z-1})$ is defined as $C_\T = \{\text{\rm Y}, \text{\rm N}\}^* \times \{0\} \times \{\text{\rm Y}, \text{\rm N}\}^*$ where a configuration $c = c_10c_2$ represents the tape content by $c_1$ and current appendant by $c_2$. The update rule induces the map $T: C_\T \rightarrow C_\T$ which simply updates the tape and the new appendant. Again, we can associate with $\T$ the dynamical system $(C_\T, T)$.
\end{definition}

Analogously to local simulation, below we define what it means for a computational system to simulate another.

\begin{definition}\label{def:tm_simulation}
    Let $\T_1 = (C_{\T_1}, T_1)$ and $\T_2 = (C_{\T_2}, T_2)$ be computational systems (each can be a Turing machine, a clockwise Turing machine, a tag system or a cyclic tag system). We say that $\T_2$ simulates $\T_1$ if there exists a mapping $\E: C_{\T_1} \rightarrow C_{\T_2}$, a partial mapping $\D: C_{\T_2} \rightharpoonup C_{\T_1}$, and a partial mapping $\tau: C_{\T_2} \rightharpoonup \N$ such that:
    \begin{equation}\label{eq:tm_simulation}
        \D \circ  (T_2^\tau)^n \circ \E = T_1^n \quad \quad \quad \text{for every } n \in \N.
    \end{equation}
    The mappings $\E$ (encoder), $\D$ (decoder), and $\tau$ (delay function) must satisfy the following:
    \begin{itemize}
        \item The encoder $\E: C_\T \rightarrow S^{\Z^d}$ is implemented by a Turing machine with two tapes. The first tape contains the input configuration $c \in C_{\T_1}$. The second tape is initially empty, and upon halting, it contains $\E(c) \in C_{\T_2}$.
        \item The decoder $\D: C_{\T_2} \rightharpoonup C_{\T_1}$ is implemented by a Turing machine with two tapes -- the input and output tape. The input tape contains the configuration $c' \in  S^{\Z^d}$, if $c'$ is in the domain of $\D$, the Turing machine halts with the initially empty output tape containing the configuration $\D(c') \in C_{\T_1}$. If $c'$ is not in the domain of $\D$, then the Turing machine need not halt.
        \item The delay function $\tau$ has the same domain as $\D$. It is also implemented by a Turing machine with two tapes -- the input and output tape. Upon halting, the output tape contains the binary representation of the integer $\tau(c') \in \N$.
    \end{itemize}
    Moreover, we require that there exists a computable function $\beta: \N \rightarrow \N$ such that $\E(c)$ runs in time $O\big(\beta(|c|)\big)$, and $\D(c')$ and $\tau(c')$ run in time $O\big(\beta(|\D(c')|)\big)$. In such a case, we say that the simulation runs in time $O(\beta(n))$.
\end{definition}

\begin{remark}
    Notice that in Definition \ref{def:local_simulation} of local simulation we have not used any particular property of a Turing machine other than the fact that it induces a dynamical system $(C_\T, T)$ where every configuration $c \in C_\T$ is a finite sequence of symbols from a finite alphabet. Thus, the definition of local simulation can be readily generalized to cover the cases of a CA locally simulating any computational system defined above.
\end{remark}

Now, we are ready to prove two straightforward, yet very useful results regarding the composition of simulating systems.
\begin{lemma}\label{lemma:composing_tm_sim}
    Let $\T_1 = (C_{\T_1}, T_1)$, $\T_2 = (C_{\T_2}, T_2)$, and $\T_3 = (C_{\T_3}, T_3)$ be computational systems (each of them can be a Turing machine, a clockwise Turing machine, a tag system or a cyclic tag system). If $\T_3$ simulates $\T_2$ in time $O(\beta_2(n))$ and $\T_2$ simulates $\T_1$ in time $O(\beta_1(n))$, then also $\T_3$ simulates $\T_1$ in time $O\left(\max\{\beta_1(n),\beta_2(n)\}\right)$.
\end{lemma}
\begin{proof}
    Let $\E_1, \D_1, \tau_1$ witness that $\T_2$ simulates $\T_1$ and let $\E_2, \D_2, \tau_2$ witness that $\T_3$ simulates $\T_2$. It is straightforward to verify that $\T_3$ simulates $\T_1$ with the encoder $\tilde{\E} =\E_2 \circ \E_1$, the decoder $\tilde{\D} =\D_1 \circ \D_2$ and with the delay function $\tilde{\tau} =\tau_2 \circ T_1^{\tau_1}$.
\end{proof}

\begin{lemma}\label{lemma:composing_local_sim}
    Let $\T_1 = (C_{\T_1}, T_1)$, and $\T_2 = (C_{\T_2}, T_2)$ be computational systems (each of them can be a Turing machine, a clockwise Turing machine, a tag system or a cyclic tag system) and let $\A = (S^{\Z^d}, F)$ be a cellular automaton. If $\T_2$ simulates $\T_1$ in time $O(\beta_2(n))$ and $\A$ simulates $\T_2$ in time $O(\beta_1(n))$, then also $\A$ simulates $\T_1$ in time $O\left(\max\{\beta_1(n),\beta_2(n)\}\right)$.
\end{lemma}
\begin{proof}
    Let $\E_1, \D_1, \tau_1$ witness that $\T_2$ simulates $\T_1$ and let $\E_2, \D_2, \tau_2$ and $v, w \in \Z^d$ witness that $\A$ simulates $\T_2$. It is straightforward to verify that $\A$ simulates $\T_1$ with the encoder $\tilde{\E} =\E_2 \circ \E_1$, the decoder $\tilde{\D} =\D_1 \circ \D_2$, with the delay function $\tilde{\tau} =\tau_2 \circ T_1^{\tau_1}$, and with shift vectors $v, w$.
\end{proof}

\subsection{Local universality of Rule 110}
 We demonstrate that the definition of local CA simulation agrees with previous constructions on one of the most elaborate examples: the Turing completeness of Rule 110. This requires going over the technical details in Cook's proof: whereas he specified the encoder in \cite{cook2009concrete}, we will fill in the details about the delay function, decoder, and shifts. 

 Cook's proof strategy of Rule 110's Turing completeness entails two important results:
\begin{enumerate}
    \item There exists a universal cyclic tag system (i.e.~a cyclic tag system that can simulate a universal Turing machine).
    \item Rule 110 can simulate a large class of cyclic tag systems that contains a universal one.
\end{enumerate}
In order to show that Cook's proof complies with our definition of local universality, we have to show that the cyclic tag system can simulate a universal Turing machine according to Definition \ref{def:tm_simulation} and that Rule 110 can simulate the universal tag system according to Definition \ref{def:local_simulation}. Then, using Lemma \ref{lemma:composing_local_sim} we obtain that Rule 110 can locally simulate a universal Turing machine.

\subsubsection{Universal Cyclic Tag System}
In this section, we briefly discuss the first result. Subsequently, our main focus is the second result, which requires a detailed understanding of a variety of gliders and their collisions within Rule 110's dynamics.

Let $\A \mapsto \B$ denote that system $\B$ can simulate system $\A$. In \cite{cook_og}, Cook shows the following series of reductions:
$$\text{Turing machine with alphabet } \{0,1\} \mapsto \text{tag system} \mapsto \text{cyclic tag system}\,.$$
This, however, leads to an exponential slowdown -- a problem later solved by Neary and Woods in \cite{p_completeness_of_rule110} who showed that a cyclic tag system can actually simulate a Turing machine in polynomial time (a result noted as surprising by Cook \cite{cook2009concrete}). The authors prove the following series of reductions:
$$\text{binary Turing machine} \mapsto \text{binary clockwise Turing machine} \mapsto \text{cyclic tag system}\,.$$
Here, a binary Turing machine is a Turing machine with only two tape symbols. Since a binary Turing machine can simulate an arbitrary Turing machine with at most a constant factor increase in time, it follows that there exists a universal binary Turing machine. We briefly summarize the results of Neary and Woods in the following two lemmas.

\begin{lemma}[Universal clockwise Turing machine \cite{p_completeness_of_rule110}]
    For each Turing machine $\T$ given by $(Q, \Sigma, \delta)$ that runs in time $t(n)$ there exists a clockwise Turing machine $\T'$ that simulates it and runs in time $O(t^2(n))$.
\end{lemma}
\begin{proof}
We give a proof sketch: Let $\T$ be a Turing machine with arbitrary alphabet $\Sigma$ and $m$ states. Then, the clockwise Turing machine $\T'$ operates as follows:
\begin{enumerate}
    \item Right moves of $\T$, when visiting a non-blank symbol, correspond to clockwise moves of $\T'$.
    \item When $\T$ moves right and visits a blank symbol, $\C_\T$ splits its current cell into two, writing a special placeholder symbol $r$ into the second cell, and proceeds by cycling around the whole tape to end up in $r$.
    \item Left moves of $\T$, when reading a non-blank symbol, correspond to $\T'$ leaving a placeholder symbol $l$ in its current cell and cycling around its whole tape while shifting each symbol one cell clockwise.
    \item When $\T$ moves left and visits a blank symbol, $\T'$ splits its current cell into two, writing a special symbol $l'$ into the first cell, and proceeds by cycling around the whole tape until it reaches the symbol $l'$.
\end{enumerate}
In this case, both the encoder $\E$ and decoder $\D$ are simply the identity mappings. Together with the delay function $\tau$, they run in time $O(n)$.
\end{proof}

\begin{lemma}[Universal cyclic tag system \cite{p_completeness_of_rule110}]\label{lemma:universal_cyclic_tag}
    For each binary clockwise Turing machine $\T$ given by $(Q, \{a, b\}, \delta)$ that runs in time $t(n)$ there exists a cyclic tag system that simulates it and runs in time $O(|Q|t^2(n)\log t(n))$.
\end{lemma}
\begin{proof}
This is an elaborate proof and we give a rough sketch; for details, see \cite{p_completeness_of_rule110}. Let us fix a binary clockwise Turing machine $\T$ with tape symbols $\{a,b \}$ and $m$ states. Let $z = 30m+61$. We will construct a cyclic tag system $\T'$ with appendants $(\alpha_0, \alpha_1, \ldots, \alpha_{2z-1})$; each $\alpha_i \in \{0,1 \}^*$. Given an initial tape configuration of $\T'$ consisting of the current state $q_i$, read symbol $\sigma_1$ and tape contents $\sigma_1 \cdots \sigma_s \in \{a, b \}^s$ we encode this as a configuration of $\B$ as follows:
$$\underline{\alpha_0}, \alpha_1, \ldots, \alpha_{2z-1} \quad \quad \langle 1, q_i \rangle \langle \sigma_1\rangle \cdots \langle \sigma_s\rangle \mu^{s'}. $$
The left part shows the appendant list with the current index underlined. The right part shows the cyclic tag system's configuration with $\mu = 10^{z-1}$, $s' = 2^{\lceil \log_2(s) \rceil}$ and $\langle 1, q_i \rangle, \langle \sigma_1\rangle, \ldots, \langle \sigma_s\rangle \in \{0,1 \}^*$.

The goal is to ``separate'' the encoded read tape symbol $ \langle \sigma_1\rangle$ from the rest of the symbols such that $ \langle \sigma_1\rangle$ is the only symbol that ``knows it has to update itself into a new tape symbol and new state'' and it also ``knows about the current state $q_i$''. 
This separation will be achieved by two stages of the simulation process and will make use of the $\mu$ symbols as counters. After the separation is finished, the system will arrive into the third stage:
$$\alpha_0, \ldots, \underline{\alpha_{z+10}}, \ldots, \alpha_{2z-1} \quad \quad \langle 3, q_i \rangle \langle \sigma_1\rangle \langle \ssig_2\rangle \cdots \langle \ssig_s\rangle \smu^{s'} $$
Here, the marked symbols are all binary sequences of length $2z$ and represent tape data and counter symbols that ``know they should not update themselves in any way''. The only uncrossed symbol is $\sigma_1$ which will make the update. But how does $\langle \sigma_1 \rangle$ know that $\T'$ is currently in state $q_i$? Crucially, this is ensured by the encoding $\langle 3, q_i \rangle$: its length is $z + i + 20$ so after $z + i + 20$ steps, the cyclic tag system ends up in state
$$\alpha_0, \ldots, \underline{\alpha_{30i+30}}, \ldots, \alpha_{2z-1} \quad \quad \langle \sigma_1\rangle \langle \ssig_2\rangle \cdots \langle \ssig_s\rangle \smu^{s'} a$$
where $a \in \{0,1 \}^*$ marks the new appendant (its precise form is not important at this level of detail). The main idea is that the update rule of the cyclic tag system now leverages two pieces of information: the current read symbol is encoded in $\langle \sigma_1\rangle$ and the current state of the simulated Turing machine is encoded in the index of the appendant: $30i+30$. Suppose for simplicity that the Turing machine, as it moves clockwise, writes only one symbol $\sigma_{s+1} \in \{a , b \}$ and changes to state $q_k$. Then, the cyclic tag system mimics this transition by reading all $2z$ symbols of $\langle \sigma_1\rangle$ and adding the appendant $\langle \sigma_{s+1} \rangle \langle 1, q_k \rangle$. To finish, stage 3 cycles through the remaining symbols while ``unmarking them''. Lastly, the length of appendant $a$ ensures that after $a$ was read, the index of the tag system returns to 0. Thus, at the end of stage 3, the new configuration of $\B$ will be:

$$\underline{\alpha_0}, \alpha_1, \ldots, \alpha_{2z-1} \quad \quad \langle 1, q_i \rangle \langle \sigma_2\rangle \cdots \langle \sigma_s\rangle \mu^{s'} \langle \sigma_{s+1} \rangle\,. $$
Apparently, the next read symbol is $\sigma_2$. The only difference with the initial encoding is that now, the encoded symbol $\langle \sigma_{s+1} \rangle $ is placed after the counter symbols $\mu$. It is apparent from the details of the simulation that the relative position of encoded tape symbols and counter symbols does not matter, as long as the order of the encoded tape symbols is correct. This sums up the main idea of the construction. If there are currently $n$ symbols on the Turing machine's tape, simulating one stage of the cyclic tag will take $O(n)$ time, and the stages are executed $O(\log n)$ times. Thus, to simulate one step of the clockwise Turing machine, the cyclic tag needs $O(n \log n)$ steps.

We briefly discuss the time and space complexities of the encoder, decoder and delay function. Each encoded object $\langle * \rangle$ has length at most $3z$ where $z$ only depends on the number of the Turing machine's states. Thus, the encoder and decoder time-complexities are in $O(n)$ (with respect to the current length $n$ of the simulated Turing machine's tape), only needing constant space. The time-delay function simply needs to iterate the cyclic tag system until the end of stage 3, outputting the number of time-steps needed to reach a decodable cyclic tag configuration, thus running in time $O(n\log n)$ (with respect to the current length $n$ of the simulated Turing machine's tape) using at most $O(n)$ space.
\end{proof}

\begin{remark}
From the detailed proof in the paper of Neary and Woods \cite{p_completeness_of_rule110}, it is clear that the cyclic tag system constructed as above satisfies a property that will later become important: when simulating the Turing machine, after every full cycle through the appendant list, at least one non-empty appendant gets appended.
\end{remark}

\subsubsection{Rule 110 Simulating Cyclic Tag Systems}
In what follows, we discuss Cook's second, much more elaborate result: Rule 110 can simulate a large class of cyclic tag systems that contains a universal one and moreover, Rule 110 only needs a polynomial time delay to do so. Our aim is to show that this result complies with Definition \ref{def:local_simulation} of local simulation.

\begin{definition}[Admissible tag systems]
    We say that a cyclic tag $\T$ is \emph{admissible} if the first appendant is nonempty and every nonempty appendant's length is a multiple of 6. We say that a configuration $c \in C_\T$ is \emph{admissible}, if while iterating $\T$ on $c$, we have a guarantee that each pass through the whole appendant list, at least one non-empty appendant gets appended.
\end{definition}

From Cook's construction, it follows that Rule 110 can simulate only admissible cyclic tag systems from admissible initial configurations. Crucially, the results of Cook, and Neary and Woods, give the following lemma.

\begin{lemma}\label{lemma:universal_admissible_cyclic_tag}
    There exists a universal cyclic tag system $\U$ that can simulate a universal Turing machine $\T_{\mathrm{univ}}$ with a polynomial time delay that satisfies the following: Let $\E$ be the encoder witnessing the simulation. Then, for every configuration $c \in C_{\T_{\mathrm{univ}}}$, $\E(c)$ and $\U$ are admissible.
\end{lemma}
\begin{proof}
    In the introduction of \cite{cook2009concrete}, Cook describes a simple way of converting an arbitrary tag system into one that can simulate it and whose appendant lengths are multiples of 6: we expand each appendant by adding the symbol N five times after every symbol in the appendant, and we expand the tape similarly. Lastly, we expand the list of appendants by adding the empty appendant five times after every original appendant. We illustrate it on an example. A tag system with initial tape containing $\text{Y}$ and appendants $(\text{Y}, \text{N})$ can be transformed into a tape $\text{YNNNNN}$ with appendants $(\text{YNNNNN}, \emptyset, \emptyset, \emptyset, \emptyset, \emptyset, \text{NNNNNN}, \emptyset, \emptyset, \emptyset, \emptyset, \emptyset)$; it is easy to check that the computational dynamics of the two systems will be equivalent. Moreover, it is easy to see that if the configuration $c$ was admissible for the original cyclic tag, then the expanded configuration $c'$ is admissible for the expanded tag.

    Let $\U$ be the universal tag system constructed by Neary and Woods which simulates a universal Turing machine via an encoder $\E$. From the details in the work of Neary and Woods (Lemma \ref{lemma:universal_cyclic_tag}), it is easy to see that $\E(c)$ is admissible for $\U$. However, it is not the case that every appendant's length is a multiple of 6. It suffices to use the above transformation, to obtain a universal cyclic tag $\U'$ which is admissible, and for which every encoded tape of the universal Turing machine it simulates is also admissible. Moreover, $\U'$ simulates the universal Turing machine in a polynomial time, with encoder, decoder and delay function complexities being polynomial in time and linear in space.
\end{proof}
 
The general idea of Cook's construction is illustrated in Fig.~\ref{fig:rule110_sim_scheme}. In the diagram, each row represents a configuration of Rule 110 with time progressing downwards; and each object (tape symbols, ossifiers, appendants) is represented by a group of gliders in Rule 110.

\begin{figure}[htbp!]
    \centering
    \includegraphics[width=0.8\linewidth]{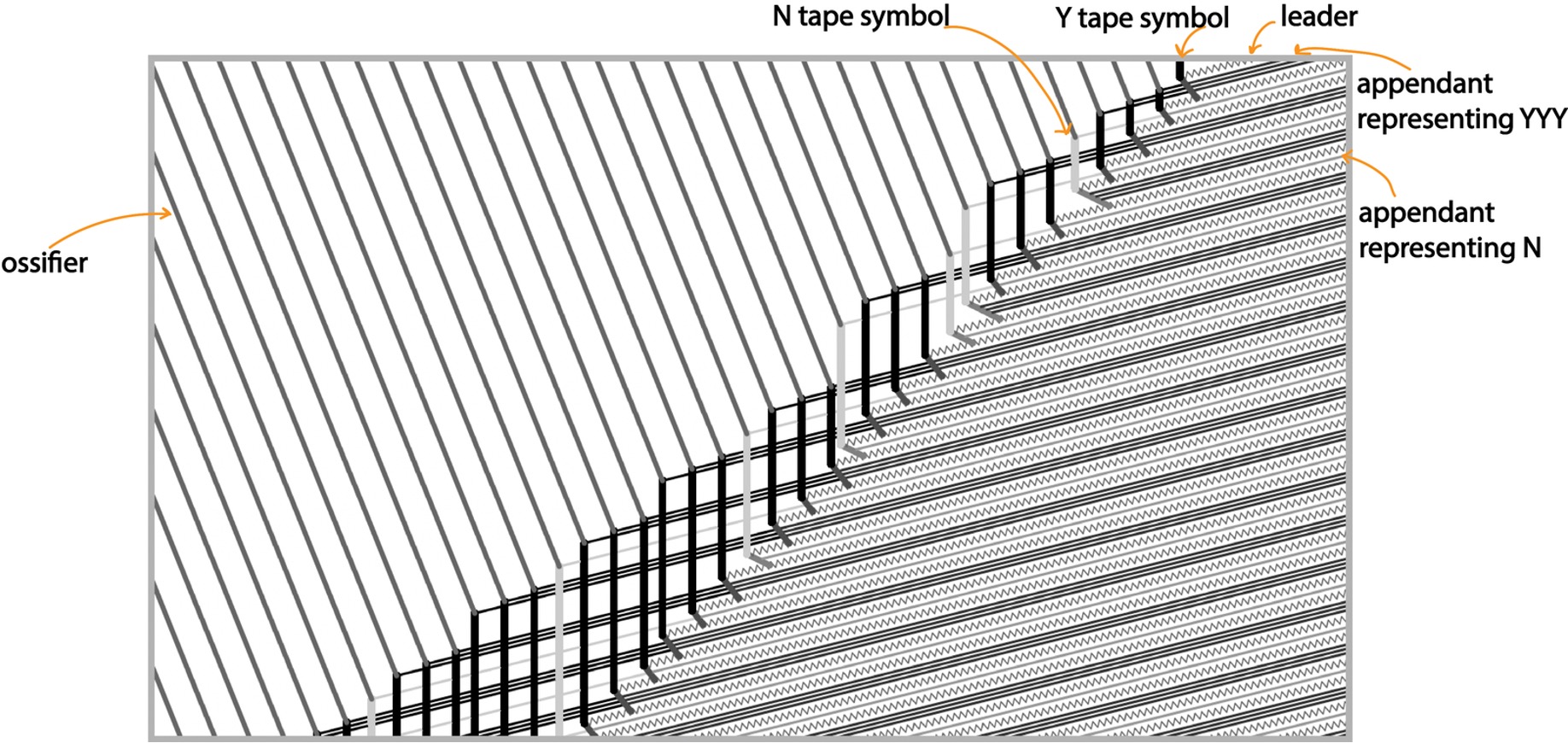}
    \caption{General Scheme of Rule 110 simulating a cyclic tag with appendant list \{YYY, N\} and with initial tape symbol Y. (Such a cyclic tag is non-admissible and this scheme is only illustrative). Figure is adapted from \cite{cook_og}.}
    \label{fig:rule110_sim_scheme}
\end{figure}

\begin{itemize}
    \item The tape data is encoded in a finite region, and the symbols Y and N represented by static gliders. We note the tape is represented in a ``reversed order'', with the first symbol on the cyclic tag's tape being the rightmost one.
    \item To the right of the tape data, there is an infinitely repeating periodic sequence which represents the ever repeating cyclic tag's appendant list. Each appendant is represented by a block of gliders moving to the left, each block starts with a ``leader'' glider, and is followed by gliders representing the actual data.
    \item Once a leader hits a tape symbol, the symbol is erased by the collision. If it was an N, then the leader gets modified and wipes out all the following appendant data, upon being annihilated by the next leader. If it was a Y, then the leader transforms the gliders into ``moving data'', which pass through the whole tape region, leaving it intact.
    \item To the left of the tape data, there is an infinitely repeating periodic sequence of ``ossifiers''. An ossifier is a glider moving towards the right. Its purpose is to collide with moving data gliders: this collision annihilates both the ossifier and moving data glider, and it creates a corresponding static data symbol. 
\end{itemize}

We note that if the ossifier collides with a tape symbol before hitting a moving data glider, both get annihilated and the construction will not work. Thus, we need a guarantee that if the ossifiers are spaced far enough apart, they will not collide with the tape data. This is exactly where the second condition of a cyclic tag's admissibility comes in.

The gliders representing the different objects of the construction are represented in Figures \ref{fig:enc_central_tape}, \ref{fig:enc_left_tape}, and \ref{fig:enc_right_tape}. One can see that the gliders have various temporal periods. Representing them as ``puzzle pieces'' ensures that the gliders are connected exactly at such a phase which makes them compatible. This elegant way of representing the building blocks of the encoder was introduced in Cook's paper \cite{cook2009concrete}.

\subsubsection{Encoder}
We first describe the construction of the encoder, closely following Cook's approach \cite{cook2009concrete} with slight modifications for soundness and more clarity. The encoded tape consists of three parts: the central part representing the tape symbols, the left part containing a periodic sequence of ossifiers, and the right part containing the periodically repeated representation of the cyclic tag system.

\begin{figure*}[thpb!]
\centering
     \begin{subfigure}[t]{0.22\textwidth}
        $$\mathbb{Y}$$
        \centering
        \includegraphics[height=.7in]{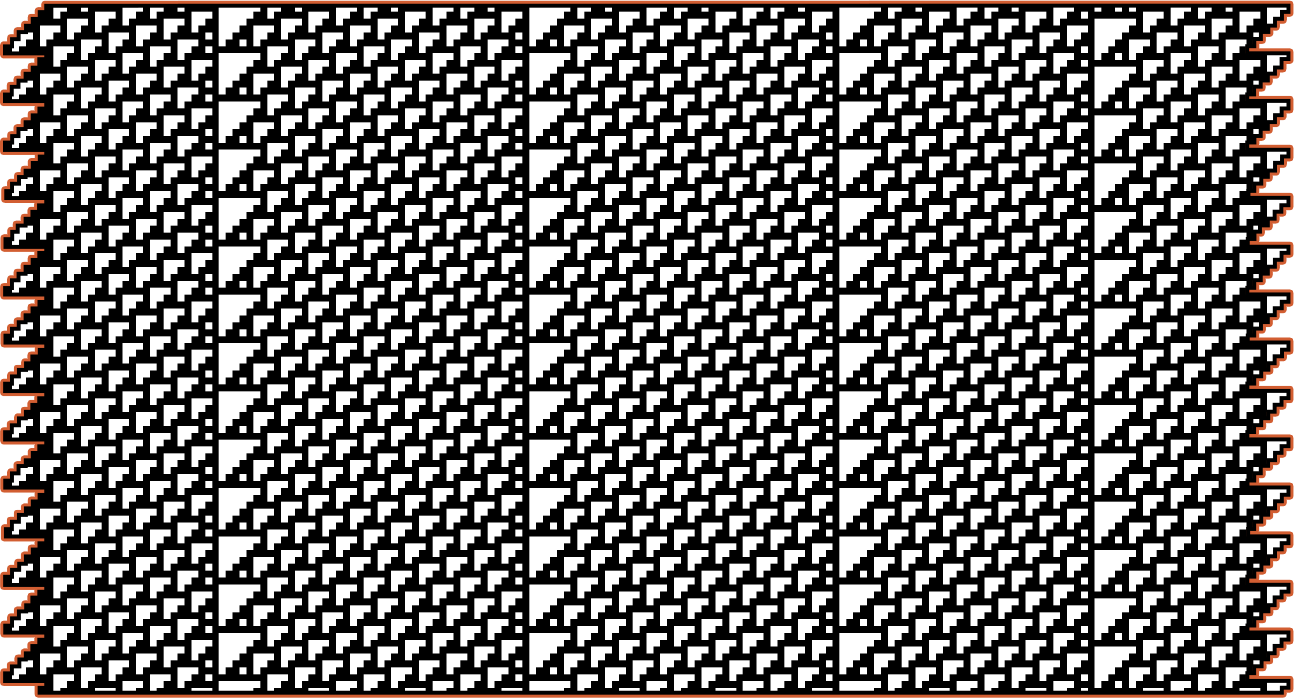}
        \caption{Y tape symbol}
        \label{fig:y}
    \end{subfigure}
    \begin{subfigure}[t]{0.3\textwidth}
        $$\mathbb{N}$$
        \centering
        \includegraphics[height=.7in]{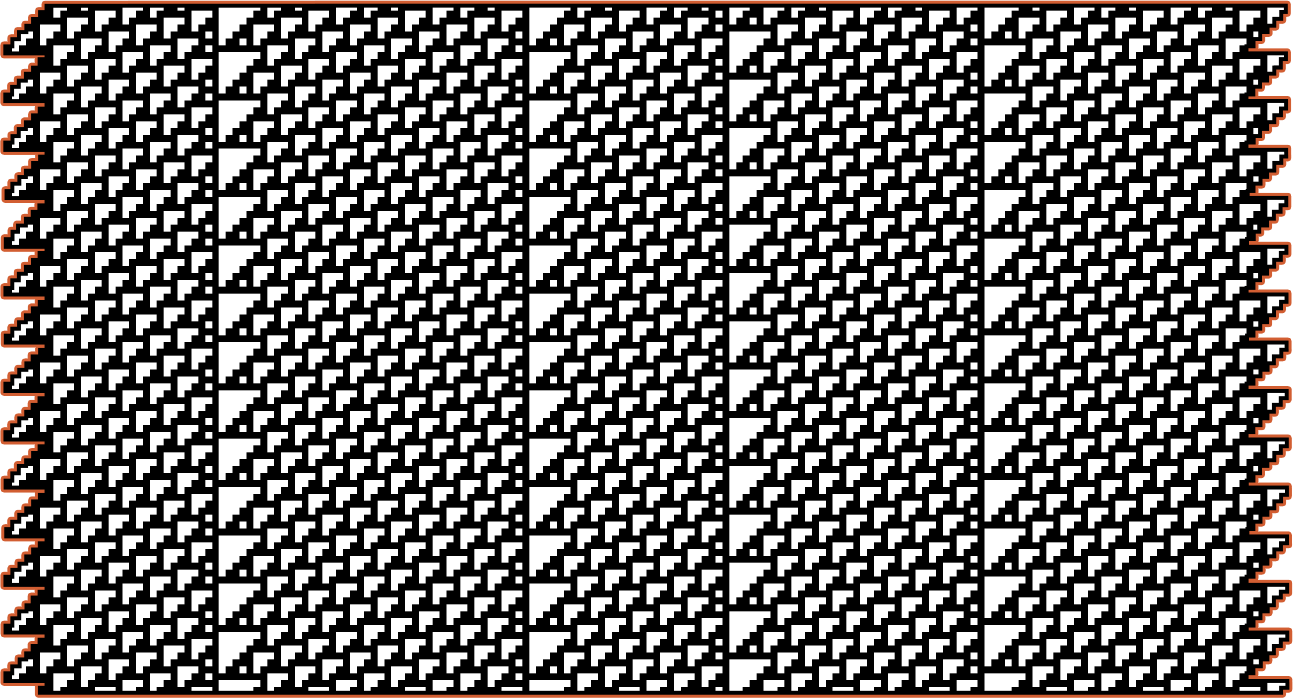}
        \caption{N tape symbol}
        \label{fig:n}
    \end{subfigure}
    \begin{subfigure}[t]{0.33\textwidth}
        $$\mathbb{L}$$
        \centering
        \includegraphics[height=.7in]{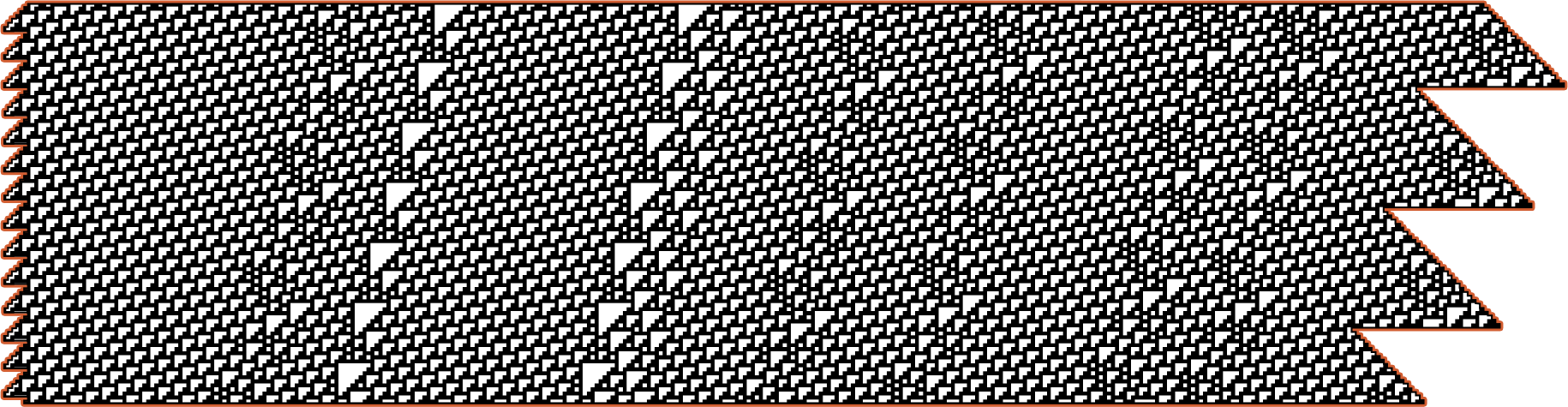}
        \caption{prepared leader}
        \label{fig:leader}
    \end{subfigure}%
    \caption{Central Part of the Tape}
    \label{fig:enc_central_tape}
\end{figure*}

\begin{example}[Central Part of the Tape]
    Encoder transforms the tape symbols $YNN$ into a finite configuration concatenating the gliders $\mathbb{N} \mathbb{N} \mathbb{Y} \mathbb{L}$ (the order of symbols in the encoded configuration is reversed). The rightmost glider is a leader which will be followed by an infinitely repeated pattern encoding the cyclic tag system's appendant list. On the left, the configuration will be prolonged by a repeated pattern of ossifiers.
\end{example}

\begin{figure*}[thpb!]
\centering
     \begin{subfigure}[t]{0.22\textwidth}
        $$\mathbb{E}$$
        \centering
        \includegraphics[height=.7in]{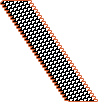}
        \caption{ether}
        \label{fig:ether}
    \end{subfigure}
    \begin{subfigure}[t]{0.3\textwidth}
        $$\mathbb{O}$$
        \centering
        \includegraphics[height=.7in]{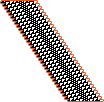}
        \caption{ossifier}
        \label{fig:oss}
    \end{subfigure}
    \begin{subfigure}[t]{0.33\textwidth}
        $$\mathbb{T}$$
        \centering
        \includegraphics[height=.7in]{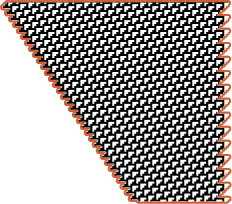}
        \caption{ether to tape symbols transition}
        \label{fig:ether_to_y}
    \end{subfigure}%
    \caption{Left Part of the Tape}
    \label{fig:enc_left_tape}
\end{figure*}

\paragraph{Left part of the tape}
The ether $\mathbb{E}$ forms a background pattern through which all gliders move. The sequence $\mathcal{L} = [\mathbb{E}]^\nu \mathbb{O} [\mathbb{E}]^{13} \mathbb{O} [\mathbb{E}]^{11} \mathbb{O}  [\mathbb{E}]^{12} \mathbb{O}$ encodes four ossifiers; upon colliding with four moving data gliders, they will produce a static tape symbol. The constant $\nu$ ensures that the groups of ossifiers are spaced far enough from each other to never collide with tape data by mistake. Explicitly, Cook \cite{cook2009concrete} computes $\nu$ as follows:
\begin{align*}
v = \quad \,\, &76 \cdot (\text{the total number of } Y\text{s in all appendants}) \\
+ \,\, &80 \cdot (\text{the total number of } N\text{s in all appendants}) \\
+ \,\, &60 \cdot (\text{the number of nonempty appendants}) \\
+ \,\, &43 \cdot (\text{the number of empty appendants})
\end{align*}

The transition pattern $\mathbb{T}$ simply connects the static tape data to the ossifiers. Lastly, we have to make sure the first group of ossifiers is spaced far enough not to hit the tape data by mistake; this is accounted for by the constant $\nu'$ (a conservative estimate is to make it larger than $\nu$ + 14 $\cdot$ (the total number of tape symbols)).

\begin{example}[Left Part of the Tape]
The left part of the tape, together with the tape symbols $YNN$ is encoded as $\cdots \mathcal{L} \mathcal{L} [\mathbb{E}]^{\nu'}\mathbb{T}\mathbb{N} \mathbb{N} \mathbb{Y} \mathbb{L}$ where $\mathcal{L} = [\mathbb{E}]^\nu \mathbb{O} [\mathbb{E}]^{13} \mathbb{O} [\mathbb{E}]^{11} \mathbb{O}  [\mathbb{E}]^{12} \mathbb{O}$. We can assume that the last symbol of $\mathbb{L}$ is exactly at the position 0.
\end{example}

\begin{figure*}[thpb!]
\centering
\begin{subfigure}[t]{0.4\textwidth}
        $$\mathbb{R}$$
        \centering
        \includegraphics[height=.65in]{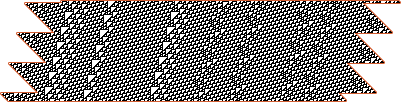}
        \caption{raw leader}
        \label{fig:rl}
    \end{subfigure}
    ~
    \begin{subfigure}[t]{0.28\textwidth}
        $$\mathbb{I}$$
        \centering
        \includegraphics[height=.65in]{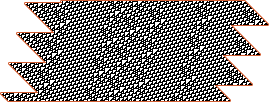}
        \caption{standard component I}
        \label{fig:scI}
    \end{subfigure}
    \begin{subfigure}[t]{0.28\textwidth}
    $$\mathbb{J}$$
        \centering
        \includegraphics[height=.65in]{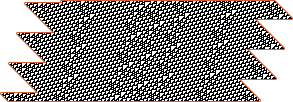}
        \caption{standard component J}
        \label{fig:scJ}
    \end{subfigure}\\
     ~
        \begin{subfigure}[t]{0.3\textwidth}
        $$\mathbb{P}$$
        \centering
        \includegraphics[height=.65in]{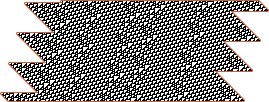}
        \caption{principal component}
        \label{fig:pc}
    \end{subfigure}
    \begin{subfigure}[t]{0.3\textwidth}
        $$\mathbb{S}$$
        \centering
        \includegraphics[height=.65in]{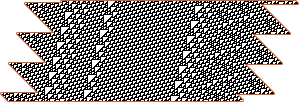}
        \caption{short leader (primary)}
        \label{fig:sl1}
    \end{subfigure}%
    ~ 
    \begin{subfigure}[t]{0.33\textwidth}
        $$\mathbb{S'}$$
        \centering
        \includegraphics[height=.65in]{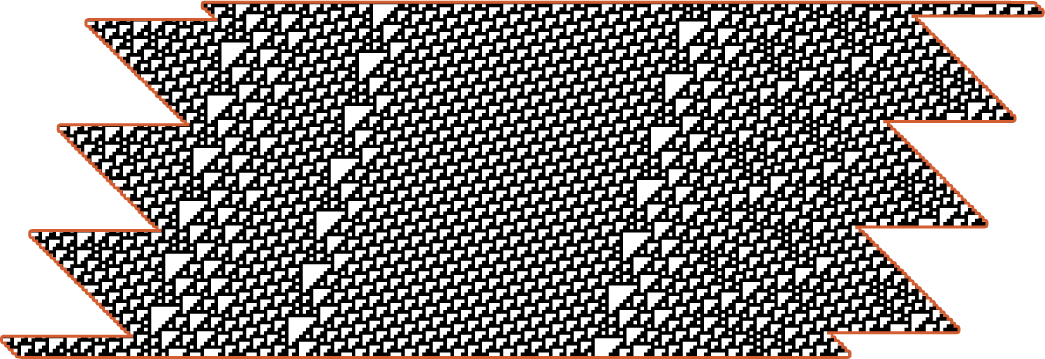}
        \caption{short leader (secondary)}
        \label{fig:sl2}
    \end{subfigure}
    \caption{Right Part of the Tape}
    \label{fig:enc_right_tape}
\end{figure*}

\paragraph{Right part of the tape}
The gliders in Fig.~\ref{fig:enc_right_tape} are used to encode the cyclic tag system's appendant list using the following rules. Each $Y$ symbol is represented by the gliders $\mathbb{I} \mathbb{I}$, each $N$ by $\mathbb{I}\mathbb{J}$. In the first symbol of each appendant, the first standard component $\mathbb{I}$ is replaced by primary component $\mathbb{P}$. There is a raw leader (either long or short) between every two appendants. Let $a_i, a_j$ be two consecutive appendants; we use raw (long) leader $\mathbb{R}$ if $a_j \neq \emptyset$, we use $\mathbb{S}$ if $a_i \neq \emptyset$ and $a_j = \emptyset$ and we use $\mathbb{S'}$ if $a_i = \emptyset$ and $a_j = \emptyset$ \footnote{We highlight a difference with Cook's construction: he only introduces the short leader $\mathbb{S}$. However, from our experiments, it appears this does not suffice since placing two components $\mathbb{S}\mathbb{S}$ next to each other representing two empty appendants does not give the correct dynamics of the simulation.}.

\begin{example}[Right Part of the Tape]
    A cyclic tag system with appendant list given by $\{\text{\rm YNN}, \text{\rm NYYN}, \emptyset, \emptyset, \emptyset\}$ is encoded as: $\mathcal{R} = \underbrace{\mathbb{P}\mathbb{I}\mathbb{I}\mathbb{J}\mathbb{I}\mathbb{J}}_{\text{\rm YNN}}\mathbb{R}\underbrace{\mathbb{P}\mathbb{J}\mathbb{I}\mathbb{I}\mathbb{I}\mathbb{I}\mathbb{I}\mathbb{J}}_{\text{\rm NYYN}}\mathbb{S}\mathbb{S'}\mathbb{S'}\mathbb{R}$.
\end{example}

\begin{lemma}
    There exists an encoder $\E$ satisfying the constraints of local simulation definition, which transforms the tape of an admissible cyclic tag into an initial configuration of Rule 110 that agrees with Cook's construction. Moreover, the encoder's complexity is linear in time and constant in space.
\end{lemma}
\begin{proof}
Let $\T$ be an admissible cyclic tag system and let $C \subseteq C_\T$ be a set of admissible configurations. We construct an encoder $\E: C \rightarrow \2^\Z$ which satisfies the constraints of local simulation Definition \ref{def:local_simulation} as follows.

Intuitively, the role of the encoder is clear: the Turing machine that implements it transforms the input $s_1 \cdots s_n \in C$ into the central part of Rule 110 configuration using the gliders in Fig.~\ref{fig:enc_central_tape}. Then, this finite sequence is embedded into an admissible background consisting of two regions: to the left we have a periodic repetition of the sequence $\mathcal{L}$ of ossifiers, to the right there is an ever repeating sequence $\R$ representing $\T$'s appendant list. This concludes the idea of the encoder.

Each glider is constant in size, and we only need a constant number of gliders to represent each symbol of the cyclic tag. For each glider, the encoder has to choose a compatible phase, which takes constant amount of time. Thus, the encoder works in time $O(n)$ and in constant space.
\end{proof}

\subsubsection{Delay Function and Decoder}

Given an admissible cyclic tag system $\T$ and an admissible initial configuration $c \in C_\T$, Cook's outline of the proof in \cite{universalities_in_cas} guarantees that once we iterate Rule 110 on the encoded configuration $\E(c)$, the collisions of the Rule 110's gliders will mimic the dynamics of $\T$, as presented in Fig.~\ref{fig:rule110_sim_scheme}. This becomes intuitively clear when observing the dynamics of Rule 110 from an encoded configuration. We illustrate the main types of glider interactions in Fig.~\ref{fig:glider_interactions}. To build a deeper intuition for the dynamics, we suggest that the reader explore the interactive library we implemented in order to visualize Rule 110 simulating arbitrary admissible tag systems available \href{https://github.com/frotaur/PyCA}{here}. 

\begin{figure*}[thpb!]
\centering
    \begin{subfigure}[t]{0.48\textwidth}
        \centering
        \includegraphics[height=2.2in]{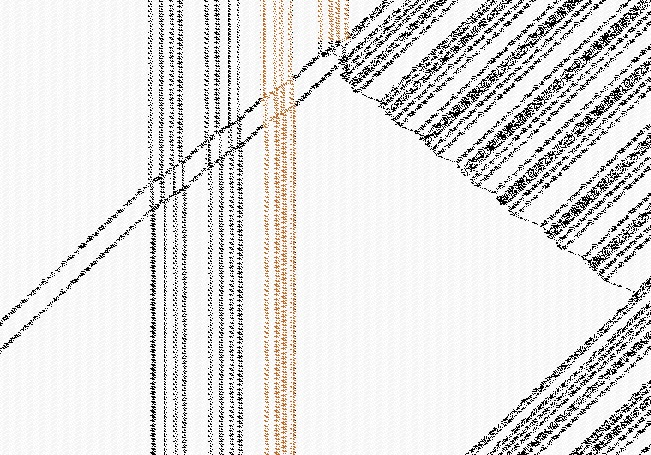}
        \caption{Tape contains NNYY; as the leader hits the symbol N, it emits a ``rejector'' that erases the following appendant and then gets absorbed by the next raw leader. The collision of the leader with a tape symbol produces a pair of ``garbage gliders'' that propagate to the left and pass through all ossifiers without interfering with them.}
        \label{fig:n_collision}
    \end{subfigure}
    \begin{subfigure}[t]{0.48\textwidth} 
        \centering
        \includegraphics[height=2.2in]{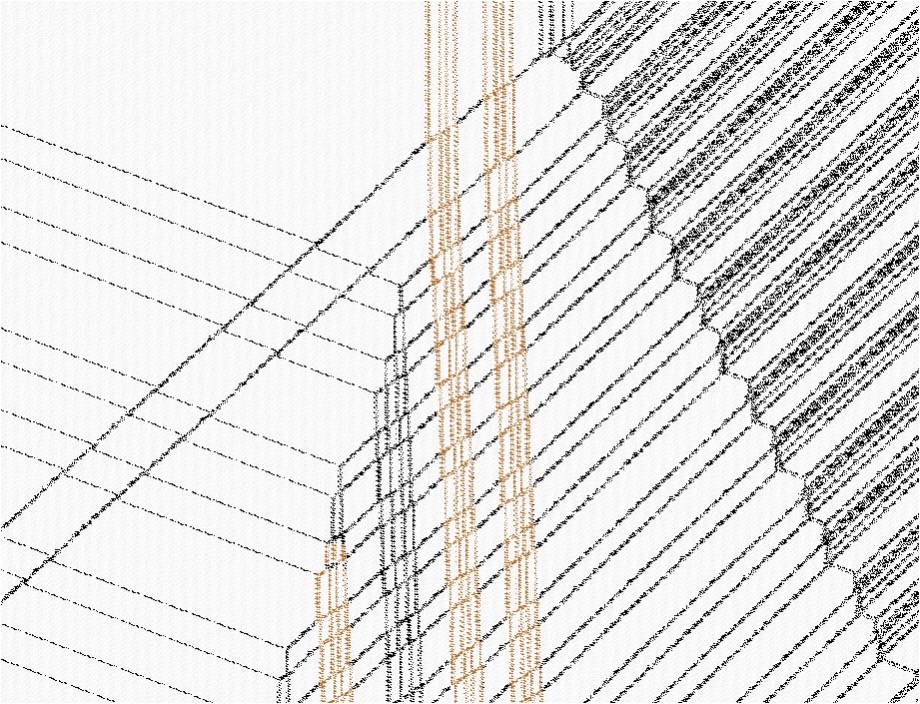}
        \caption{Tape contains YNN; as the leader hits the first Y, it emits an ``acceptor'' that transforms the following appendant into moving data. The moving data pass through the static tape data, leaving it intact. Eventually, they collide with ossifiers, each four-tuple transformed into a static tape symbol. Again, a pair of ``garbage gliders'' is produced.}
        \label{fig:y_collision}
    \end{subfigure}%
    \caption{Main types of glider interactions in Rule 110's simulation of a cyclic tag system. Both figures were generated by the PyCA Rule 110 library we implemented \href{https://github.com/frotaur/PyCA}{here}. For better visibility, the ether background is abstracted away, static tape gliders representing Y are colored black, gliders representing N are colored orange.}
    \label{fig:glider_interactions}
\end{figure*}

In Cook's construction, it is not specified what are the times at which to read out the current cyclic tag system's tape symbols nor what is the particular form of the decoder. We describe both in more detail below.

\paragraph{Delay function $\tau$}

The main idea behind the algorithm computing $\tau$ is simple: given a configuration $c$ which represents the current tape content and a current appendant about to hit the current first tape symbol, $\tau$ should ensure that once we iterate Rule 110 $\tau(c)$ times from $c$, the current appendant will have interacted with the first tape symbol, resulting in either a complete transformation of the appendant into moving data (if the symbol was Y) or a complete deletion of the appendant (if the symbol was $N$). This ensures that after $\tau(c)$ iterations, all symbols from the update tape are present in the Rule 110's configuration as either static tape data or moving data. Then, the decoder simply maps the static and moving tape data into the corresponding tape symbols. We portray a schematic illustration of $\tau$ in Fig.~\ref{fig:scheme_tau}. 

\begin{figure}[htbp!]
    \centering
    \includegraphics[width=0.8\linewidth]{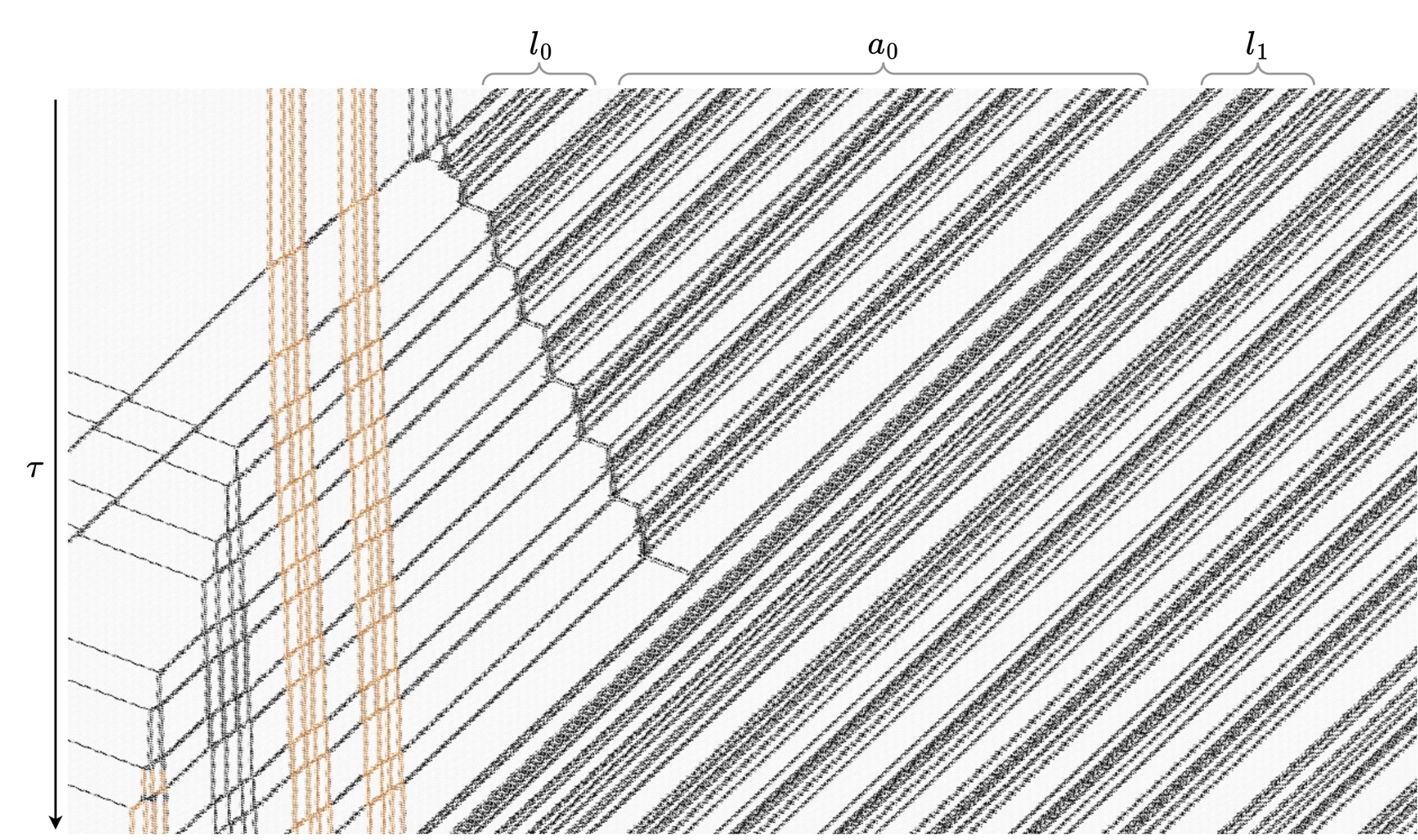}
    \caption{Space-time diagram of Rule 110 illustrating the delay function $\tau$. The first row shows a finite region of a configuration $c \in \2^\Z$ where the leftmost leader $l_0$ is about to hit the rightmost tape symbol Y. The last row shows a portion of the configuration $F^\tau(c)$, where the next leader $l_1$ is about to hit the next tape symbol N. We have guarantee that in $F^\tau(c)$ the whole appendant $a_0$ has been processed and was either deleted, or transformed into moving data that are located to the left of the current rightmost tape symbol.}
    \label{fig:scheme_tau}
\end{figure}

We present a very simplified outline of $\tau$'s implementation below. 

\begin{algorithm}[h!]
\footnotesize
    \KwData{configuration $c \in \{0,1 \}^\Z$ where the distance between the current leftmost leader $l_t$ and the rightmost tape symbol is less than 300.\\}
    \KwResult{$\tau \in \N$ the number of iterations such that $c' = F^\tau(c)$ is the configuration where the distance between the current leftmost leader $l_{t+1}$ and the rightmost tape symbol is less than 300.}
    $\tau = 0$\;
    \While{dist(next tape symbol, leader $l_{t+1}$) $> 300$}
       {$\tau = \tau + 1$\;
       $c = F(c)$\;
       }
    return $\tau$\;
    \caption{Algorithm for the delay function.}
    \label{alg:tau}
\end{algorithm}

\paragraph{Decoder}
The algorithm implementing the decoder is straightforward: $\D$ receives a configuration where all current tape symbols are present as either static data or moving data. The decoder simply has to first sweep the tape right to left and decode all static data into either Y or no symbols, and afterwards, it has to sweep the tape from left to right, decoding all moving data gliders. A particularly simple implementation is suggested below:

\begin{algorithm}[h!]
\footnotesize
    \KwData{configuration $c \in \{0,1 \}^\Z$ where the distance between the current leftmost leader $l_t$ and the rightmost tape symbol is less than 300.\\}
    \KwResult{The corresponding tape content $s_1\ldots s_n \in C_\T$ of the tag system $\T$.}
    find $i \in \N$ the position of the leader $l_t$\;
    find $j \in \N$ the position of the rightmost ossifier in $c$\;
    crop the configuration $c' = c[j,i]$\;
    \While{there are moving data present in $c'$}
       {\If{$c'$ contains no ossifiers}{$c' = \mathcal{L}c'$}
       $c' = F(c')$\;
       }
    reading $c'$ from right to left, match every pattern $\mathbb{Y}$ and $\N$\;
    return the corresponding symbols
    \caption{Algorithm for the decoder $\D$.}
    \label{alg:decoder}
\end{algorithm}

\begin{lemma}\label{lemma:bounded_window}
   Let $c \in C_\T$ be a configuration encoding $n$ symbols of thy cyclic tag's tape, to the right of which is a leader $l_1$ about to hit the first tape symbol. Let $i$ be the index delimiting the start of $l$. Then, both $\tau$ and the decoder only need access to $c_{[i-k, i+k]}$ where $k \in O(n)$.  
\end{lemma}
\begin{proof}
    To the left of $i$, both $\tau$ and the decoder need access to all the static and moving tape data. Let $\nu \in \N$ be the constant number of ether blocks between two consecutive ossifier groups. This guarantees that two consecutive static data symbols are no more than $\nu$ ether blocks apart. Thus, the distance between two static tape symbols is bounded by $K_0 = \nu\cdot |\mathbb{E}|$. And the distance from $i$ to the first ossifier is bounded by $K_0\cdot (n+2)$. To the right, $\tau$ and the decoder need access to the next leader $l_{t+1}$. The distance from $l_t$ to $l_{t+1}$ is bounded by the total length of the encoded cyclic tag pattern, which is some constant $K_1$.
     
    Therefore, both $\tau$ and the decoder need access to the finite part of the configuration delimited to the left by an ossifier and to the right by the leader $l_{t+1}$: $c_{[i-K_0\cdot (n+2), i+K_1]}$.
\end{proof}

The last and most cumbersome detail to discuss is the following: We have already argued that both the delay function and decoder only need to read a finite portion of the configuration whose size linearly depends on $n$. This is clear if we center the read-out window around the leftmost's leader position. However, in the definition of local simulation, the readout window is centered around 0. Thus, we have to ensure there exists appropriate shift vectors $v, w$ such that whenever we look at the first leader in the configuration $\sigma_w \circ (\sigma_v \circ F)^\tau$, its distance from 0 is bounded by some constant (independent of the total number of iterations).

We discuss that this is possible in detail below. The leader, as well as all gliders in the right periodic part of the encoded configuration, moves at a constant speed of $v = \frac{4}{15}$; i.e.~every 15 time-steps it moves 4 cells to the left. Let $K_a = |\mathcal{R}|$ be the length of the finite sequence encoding the cyclic tag's appendant list and let $z$ be the number of appendants in the list. We denote by $w = \frac{K_a}{z}$ the average appendant length. 

\begin{lemma}\label{lemma:centering_readout_window}
    Let $\T$ be an admissible tag system and let $c \in C_\T$ be an admissible initial configuration. Let $c_0 = \E(c)$ be the encoded configuration of Rule 110 where $\E$ is the encoder described above. Let $K_a = |\mathcal{R}|$ be the length of the finite sequence encoding the cyclic tag's appendant list and let $z$ be the number of appendants in the list. We let $w = \frac{K_a}{z}$ be the average appendant length and set $v = -\frac{4}{15}$. Then, for each iteration $t \in \N$, the configuration $c_t = \sigma_w \circ (\sigma_v \circ F)^\tau \circ \E(c)$ satisfies the following:
    \begin{enumerate}
        \item the leftmost leader $l_t$ in $c_t$ is ``about to hit'' the rightmost static tape symbol; i.e., their distance is less than 300 
        \item the index $i_t$ delimiting the start of $l_t$ satisfies $|i_t| \leq 2K_a$. 
    \end{enumerate}
\end{lemma}

The lemma implies that if the decoder or delay function receive the configuration $c_t$, in order to decode or compute the time delay, they only require access to a finite window of the configuration centered around 0, whose size is bounded by the current number of tape symbols.

\begin{proof}
    Both points $1$ and $2$ clearly hold for $c_0$ from the definition of the encoder with $i_0 = 0$. From the definition of $\tau$, it is clear that $c_1$ encodes the situation when the next leader $l_1$ as about to hit the current rightmost tape symbol, so clearly $1.$ is satisfied. In the configuration $F^\tau(c_0)$ the leader $l_2$ travelled with the speed of 4 cells to the left every 15 time-steps. Thus, shifting the configuration by $(\sigma_v \circ F)^\tau(c_0)$ ensures that the position of $l_1$ in $c_0$ and $c_1$ is identical. If the distance between $l_0$ and $l_1$ was exactly $w = \frac{K_a}{z}$ then the final shift $\sigma_w \circ (\sigma_v \circ F)^\tau(c_0)$ would ensure that the position of $l_1$ in $c_1$ is again exactly 0. This might not be the case in general, resulting in the more general bound $|i_1| \leq 2 K_a$. However, after iterating through all $z$ appendants, we have a guarantee that $i_z = 0$. Thus, every $z$ iterations of the simulation, the current leader's position returns to 0. This finishes the outline of the proof.
\end{proof}

For our exact implementation of the encoder, $\tau$, and decoder in Python, see \href{https://github.com/barahudcova/rule110_local_simulation}{rule110decoder}. 

\begin{lemma}\label{lemma:enc_dec_complexity}
    The encoder, decoder, and time-delay witnessing that Rule 110 simulates a universal cyclic tag system have polynomial time complexity, and linear space complexity. 
\end{lemma}
\begin{proof}
The encoder's complexity was discussed above. Suppose the configuration $c$ encodes $n$ tape symbols of the cyclic tag and is ready to be decoded. 

The decoder, starting at index 0, requires $O(n)$ time to find the leftmost leader and rightmost ossifier, delimiting a finite portion of the configuration $c_{[i, j]}$ it needs for its computation whose size is in $O(n)$. Next, the decoder as to transform all moving data into static data. This happens in $O(n)$ iterations of Rule 110. Each iteration takes time $O(n)$ to simulate. Finally, the decoder has to sweep the whole configuration, matching gliders to decode into tape symbols. This takes $O(n)$ time-steps. Thus, the overall time-complexity of the decoder is $O(n^2)$. The decoder only needs to store the current configuration throughout its simulation of Rule 110, resulting in the space complexity of $O(n)$.

Finally, the time-delay function, starting at index 0, also requires $O(n)$ time to find the second leftmost leader and rightmost ossifier, delimiting a finite portion of the configuration $c_{[i, j]}$ it needs for its computation whose size is in O(n). Then, the delay function iterates Rule 110 on this finite region until this leader is about to hit the next static tape symbol, this happens in $O(n)$ iterations. Therefore, the overall time-complexity of the time delay is $O(n^2)$, and the space complexity is $O(n)$.
\end{proof}

We can summarize the previous discussion in the following result.
\begin{theorem}
    Rule 110 is locally universal. The time required to simulate one step of a universal Turing machine is polynomial. The encoder, decoder, and time-delay functions witnessing this relation require polynomial time and linear space.
\end{theorem}
\begin{proof}
    From the results of Cook \cite{cook2009concrete} and Woods and Neary \cite{p_completeness_of_rule110} we have a guarantee that there exists an admissible tag system $\U$ which simulates a universal Turing machine using a subset of its configurations $C \subseteq C_\U$ which are all admissible.

    Further, Cook's results guarantee that any $c \in C$ can be encoded into Rule 110's configuration such that Rule 110 iterated on $\E(c)$ mimics the dynamics of $\U$ according to the diagram in Fig.~\ref{fig:rule110_sim_scheme}. From our detailed discussion above, it is easy to verify that $\E$ satisfies the encoder conditions in the definition of local simulation.

    We specify the delay function and decoder by pseudocodes \ref{alg:tau} and \ref{alg:decoder}, showing both functions can be implemented by a Turing machine. Lemma \ref{lemma:bounded_window} shows that both $\tau$ and $\D$ only need access to a finite part of their input configuration whose size is proportional to the current number of tape symbols. Further, Lemma \ref{lemma:centering_readout_window} shows that the readout window can be centered around 0 by choosing appropriate shift vectors for the simulation. 

    From Lemma \ref{lemma:enc_dec_complexity} and previous discussions, we know that along the series of reductions shown in Fig.~\ref{fig:reductions}, at every step, the encoder, decoder and time-delay have polynomial time-complexity and linear space complexity. Therefore, this also holds for their composition.
\end{proof}

\begin{figure}[htbp!]
    \centering
    \includegraphics[width=0.75\linewidth]{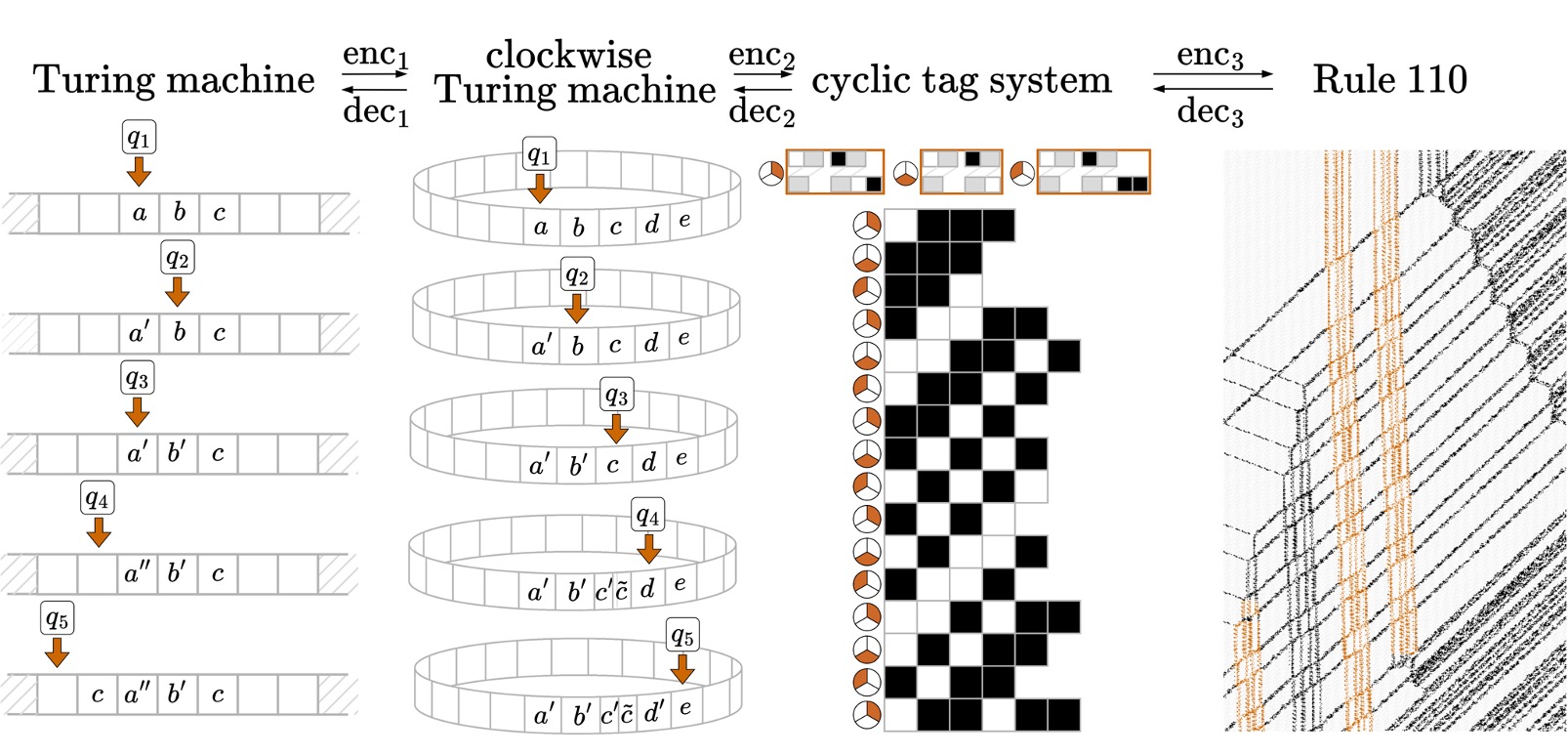}
    \caption{A Scheme illustrating the series of reductions all the way from a universal Turing machine to Rule 110.}
    \label{fig:reductions}
\end{figure}

\subsection{Local vs.~global Universality}
The goal of this section is to show that global universality is a strictly stronger condition than local universality. First, we show that each globally universal CA is also locally universal. Second, we highlight an example of a 1D reversible CA which is locally universal but not globally universal.

\begin{lemma}\label{lemma:local_global_relation}
    Let $\B = (T^{\Z^d}, G)$ and $\A = (S^{\Z^d}, F)$ be cellular automata and $\T =(C_\T, T)$ a computational system (a Turing machine, cyclic tag, etc.). Assume that $\B$ globally simulates $\A$ and $\A$ locally simulates $\T$. Then, $\B$ locally simulates $\T$.
\end{lemma}
\begin{proof}
        Let $\E_1: C_\T \rightarrow S^{\Z^d}$ be the encoder and $\D_1: S^{\Z^d} \rightharpoonup C_\T$ the decoder that witness $\A$ locally simulating $\T$ with time-delay function $\tau: S^{\Z^d} \rightharpoonup \N$ and shift vectors $v, w \in \Z^d$; we have:
    \begin{equation}\label{eq:local_witness}
       \D_1 \circ \big( \sigma_w \circ (\sigma_v \circ F)^\tau \big)^n \circ \E_1(s) = T^n(s) \quad \quad \text{ for all } n \in \N \text{ and all } s \in C_\T. 
    \end{equation}
    
    Let $\E_2: S^{\Z^d} \rightarrow T^{\Z^d}$ be the encoder and $\D_2: T^{\Z^d} \rightharpoonup S^{\Z^d}$ the decoder witnessing that $\B$ globally simulates $\A$ with delay constant $t \in \N$ and shift vector $u \in \Z^d$; we have:
    \begin{equation}\label{eq:global_witness}
        \D_2 \circ \sigma_u \circ G^t \circ \D^{-1}_2(c) = F(c) \quad \quad \text{ for all } c \in S^{\Z^d}.
    \end{equation}

    Then, \eqref{eq:global_witness} implies:
    \begin{equation*}
        \D_2 \circ \sigma_{M\cdot v + u} \circ G^t \circ \D^{-1}_2(c) = \sigma_v \circ F(c) \quad \quad \text{ for all } c \in S^{\Z^d}.
    \end{equation*}
    And further, due to Lemma \ref{lemma:global_sim_iterations}:
    \begin{equation*}
        \D_2 \circ \big( \sigma_{M\cdot v + u} \circ G^t \big)^{\tau(c)} \circ \D^{-1}_2(c) = \sigma_v \circ F^{\tau(c)}(c) \quad \quad \text{ for all } c \in S^{\Z^d}.
    \end{equation*}
    We define a new delay-function $\tau': T^{\Z^d} \rightharpoonup \N$ as $\tau'(c') = \tau(\D_2(c'))$ for $c' \in T^{\Z^d}$. Then
    \begin{equation*}
        \D_2 \circ \big( \sigma_{M\cdot v + u} \circ G^t \big)^{\tau'} \circ \D^{-1}_2(c) = (\sigma_v \circ F)^{\tau}(c) \quad \quad \text{ for all } c \in S^{\Z^d}\,,
    \end{equation*}
    and 
    \begin{equation*}
        \D_2 \circ \sigma_{M\cdot w} \circ ( \sigma_{M\cdot v + u} \circ G^t )^{\tau'} \circ \D^{-1}_2(c) = \sigma_w \circ (\sigma_v \circ F)^{\tau}(c) \quad \quad \text{ for all } c \in S^{\Z^d}.
    \end{equation*}
    Again, iterating the argument from Lemma \ref{lemma:global_sim_iterations} we have that for all $n \in \N$:
    \begin{equation*}
        \D_2 \circ \big( \sigma_{M\cdot w} \circ ( \sigma_{M\cdot v + u} \circ G^t )^{\tau'} \big)^n \circ \D^{-1}_2(c) = \big( \sigma_w \circ (\sigma_v \circ F)^{\tau}\big)^n(c) \quad \quad \text{ for all } c \in S^{\Z^d}.
    \end{equation*}
    And thus, as a special case, it also holds for all $n \in \N$ that:
    \begin{equation*}
        \D_2 \circ \big( \sigma_{M\cdot w} \circ ( \sigma_{M\cdot v + u} \circ G^t )^{\tau'} \big)^n \circ \E_2(c) = \big( \sigma_w \circ (\sigma_v \circ F)^{\tau}\big)^n(c) \quad \quad \text{ for all } c \in S^{\Z^d}.
    \end{equation*}
    Substituting the last expression into \eqref{eq:local_witness} we get:
    \begin{equation*}
        \D' \circ \big( \sigma_{M\cdot w} \circ ( \sigma_{M\cdot v + u} \circ G^t )^{\tau'} \big)^n \circ \E'(s) = T^n(s) \quad \quad \text{ for all } s \in C_\T\,,
    \end{equation*}
    where $\E' = \E_2 \circ \E_1$ and $\D' = \D_1 \circ \D_2$. 

    Lastly, it remains to verify the following:
    \begin{enumerate}
        \item Let $B_1, \ldots, B_k$ be the list of admissible backgrounds for the encoder $\E_1$. Then, $\E_2(B_1), \ldots, \E_2(B_k)$ is also a list of admissible backgrounds; this follows from the fact that $\E_2$ is a cellular transformaton.
        \item $\E'$ is a mapping implementable by a Turing machine in the sense of Definition \ref{def:local_simulation} with admissible backgrounds $\E_2(B_1), \ldots, \E_2(B_k)$. This follows from the fact that $\E_2$, as a cellular transformaton, can be implemented by a Turing machine. 
        \item $\D'$ and $\tau'$ can be implemented by a Turing machine in the sense of Definition \ref{def:local_simulation}. This again follows from the fact that $\D_2$ as a cellular transformaton, can be implemented by a Turing machine. 
    \end{enumerate}
\end{proof}

An important corollary of the previous lemma is the following relationship between local and global universality.

\begin{theorem}\label{thm:local_global_relation}
    Each globally universal CA is also locally universal.
\end{theorem}
\begin{proof}
    It is apparent that there exists a CA $\A_d = (S^{\Z^d}, F)$ which can simulate a universal Turing machine $\T_U$ for every dimension $d \in \N$. Let $\B = (T^{\Z^d}, G)$ be a globally universal CA in dimension $d \in \N$. Since $\B$ can simulate $\A_d$, using Lemma \ref{lemma:local_global_relation} we get that $\B$ can also simulate $\T_U$ and thus that $\B$ is also locally universal.
\end{proof}

In what follows, we show that global universality is a strictly stronger condition than local universality in every dimension. To be precise, the following statement holds:
\begin{theorem}
    For each $d \in \N$ there exists a $d$-dimensional cellular automaton which is locally universal but is not globally universal.
\end{theorem}

The argument is relatively simple, and uses the known fact that reversible cellular automata can be locally but not globally universal. We note that a more technical and general version of this statement was presented in \cite{hertling1998embedding}. We begin by introducing reversible automata.

A cellular automaton with global rule $F: S^{\Z^d} \rightarrow S^{\Z^d}$ is reversible if $F$ is a bijection. In this case, since $F^{-1}$ is continuous and commutes with all shifts, the Curtis-Hedlund-Lyndon Theorem \ref{thm:og_curtis_hedlund_lyndon} implies that it is also the global rule of a CA. While there is an algorithm that decides the reversibility of each 1D CA, this problem is proven to be undecidable for dimensions $d \geq 2$ by Kari \cite{kari1990reversibility}; a general overview of results regarding reversible CAs is \cite{kari2005reversible}. Reversible CAs have been widely studied in the context of their computational capabilities; we list a selection of results \cite{intrinsic_universality_of_reversible_1D_ca, reversible_spacetime_simulation_of_cas, morita1989computation, imai2000computation}.

\begin{theorem}[Reversible CAs are locally universal \cite{toffoli1977computation, morita1995}]\label{thm:reversible_cas_local_universality}
    For every $d \in \N$ there exists a $d$-dimensional reversible cellular automaton which is locally universal.
\end{theorem}
\begin{proof}
    Toffoli famously showed that for any $d$-dimensional CA can be simulated by a $d+1$-dimensional reversible CA \cite{toffoli1977computation}, which proves the result for $d \geq 2$. Later, in 1995, Morita showed the result for $d=1$, \cite{morita1995}. To be precise, one has to verify the constructions from the two works are compliant with our Definition \ref{def:local_simulation} of local simulation, which is a straightforward yet technical exercise.
\end{proof}

We proceed by showing that any reversible CA is limited in terms of what it can simulate \emph{globally}.

\begin{theorem}[Reversible CAs cannot be globally universal]
    Let $\mathcal{R}$ be a $d$-dimensional reversible CA with global rule $G: S^{\Z^d} \rightarrow S^{\Z^d}$ and let $\A$ be a CA with states $\2 \coloneqq \{0,1 \}$ whose global rule $F$ is a constant mapping sending everything to $0^{\Z^d}$. Then, $\mathcal{R}$ cannot globally simulate $\A$.
\end{theorem}
\begin{proof}
    For contradiction, we will assume that there exists a shift vector $v \in \Z^d$, a time delay $\tau \in \N$, an encoder $\E: \2^{\Z^d} \rightarrow S^{\Z^d}$, and a decoder $\D: S^{\Z^d} \rightharpoonup \2^{\Z^d}$ witnessing that $\mathcal{R}$ globally simulates $\A$. In particular, since $\E$ satisfies the conditions from Definition \ref{def:global_simulation}, there exist sequences $V = (v_1, \ldots, v_d)$ and $W = (w_1, \ldots, w_d)$ of linearly independent vectors from $\Z^d$ such that $\sigma_{w_i} \circ \E =\E \circ \sigma_{v_i}$ for all $i \in \{1, \ldots, d\}$. We list a series of observations:
    \begin{enumerate}
        \item Since $\mathcal{R}$ is reversible, $G$ is injective, and $\sigma_v \circ G^\tau$ is also injective.
        \item Because of the condition $\D \circ \E = \mathrm{id}_{\2^{\Z^d}}$, $\D$ cannot be a constant mapping (this is important to see that $\D$ cannot be the one computing the constant function of $\A$).
        \item Let $c \in \2^{\Z^d}$ be a $V$-periodic configuration. Then, $\E(c)$ is a $W$-periodic configuration. We denote by $S_W$ the set of all $W$-periodic configurations in $S^{\Z^d}$. Since $W$ spans a $d$-dimensional sublattice of $\Z^d$, the set $S_W$ is finite.
        \item Clearly, the mapping $\sigma_v \circ G^\tau$ preserves the $W$-periodicity of a configuration, since it is a CA.
    \end{enumerate}
Now, we can finish the proof. Let $c \in \2^{\Z^d}$ be an arbitrary $V$-periodic configuration such that $c \neq 0^{\Z^d}$. Since $\sigma_v \circ \restr{G^\tau}{S_W}: S_W \rightarrow S_W$ is an injective mapping on a finite set, there exists some $0 < n < |S_W|$, such that $(\sigma_v \circ G^\tau)^n \circ \E(c) = \E(c)$. This gives the contradiction since on one hand $\D \circ \E(c) = c \neq 0^{\Z^d}$ and on the other hand $\D \circ \E(c) = \D \circ (\sigma_v \circ G^\tau)^n \circ \E(c) = F^n(c) = 0^{\Z^d}$.
\end{proof}

\section{One-dimensional universal self-replicator}
\label{appsec_1d_sr}

In von Neumann's construction of a universal self-replicator \cite{neumann_vnsr}, the new copy of the machine is built with a vertical offset. Our goal is to modify this construction so that the copies are built at the exact same height as the original machine. This will allow us to compress the dynamics into a single dimension, resulting in a 1D CA capable of universal self-replication. The general idea is illustrated in Fig.~\ref{fig:vnsr_general_idea}.

\begin{figure*}[htbp!]
    \centering
    \begin{subfigure}[t]{0.5\linewidth}
        \centering
        \includegraphics[height=1.4in]{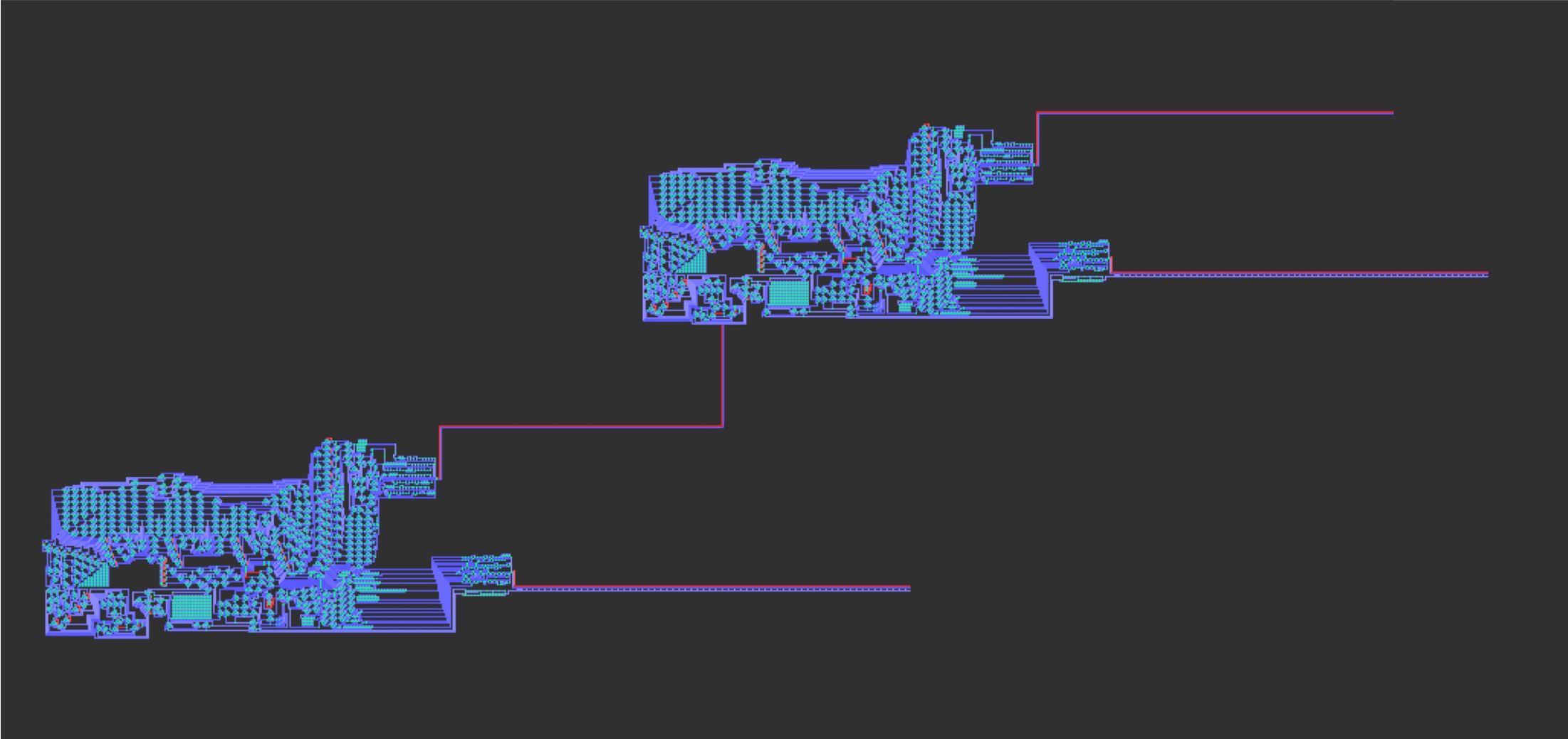}
        \caption{}
        \label{fig:vnsr_general_idea_a}
    \end{subfigure}%
    \begin{subfigure}[t]{0.5\linewidth}
        \centering
        \includegraphics[height=1.4in]{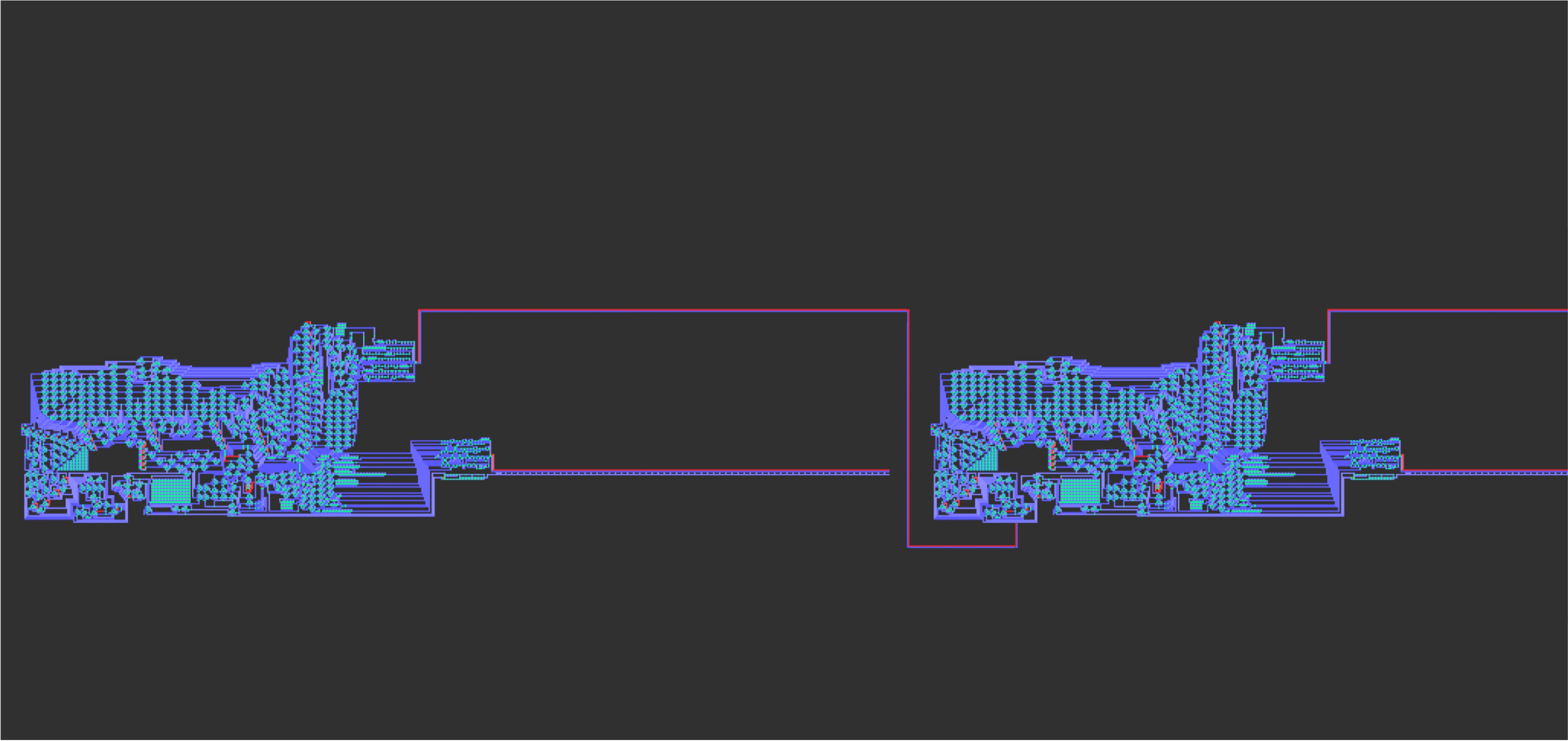}
        \caption{}
    \end{subfigure}
    \caption{General idea behind the self-replicator modification. (a) Scheme of the original universal self-replicator, as implemented by Buckley  \cite{buckley2008signal} in the software Golly \cite{golly}. (b) Illustration of the modified self-replicator we propose that has no vertical offset.}
    \label{fig:vnsr_general_idea}
\end{figure*}

Concretely, we will modify the construction of J. W. Thatcher \cite{thatcher_construction} who provides a simplification of von Neumann's universal self-replicator.

\subsection{Overview of the original construction}

The Von Neumann Self-Replicator (VNSR) consists of three functional blocks: the supervisory unit, the tape unit, and the constructing unit; as illustrated in Thatcher's Fig.~\ref{fig:vnsr_scheme}. The tape unit is responsible for reading binary data from the tape (which is indexed by positive integers, and extends to the right) and transmitting the corresponding signals to the information channel. The constructing unit operates a constructing arm which can move around the plane and, on the basis of incoming inputs, construct an arbitrary array of (inactive) cell states. The supervisory unit represents the ``universal computer'' that communicates via the information channel both with the tape unit -- giving the reading head instructions to read a symbol or move; and with the constructing unit -- sending instructions to the construction arm to move to a certain location and to create a new state there. The supervisory unit consists of components which represent internal states of the machine and also encode transitions between the states based on current state and potentially, on the tape symbol read.

\begin{figure}[htbp!]
        \centering
        \includegraphics[width=0.9\linewidth]{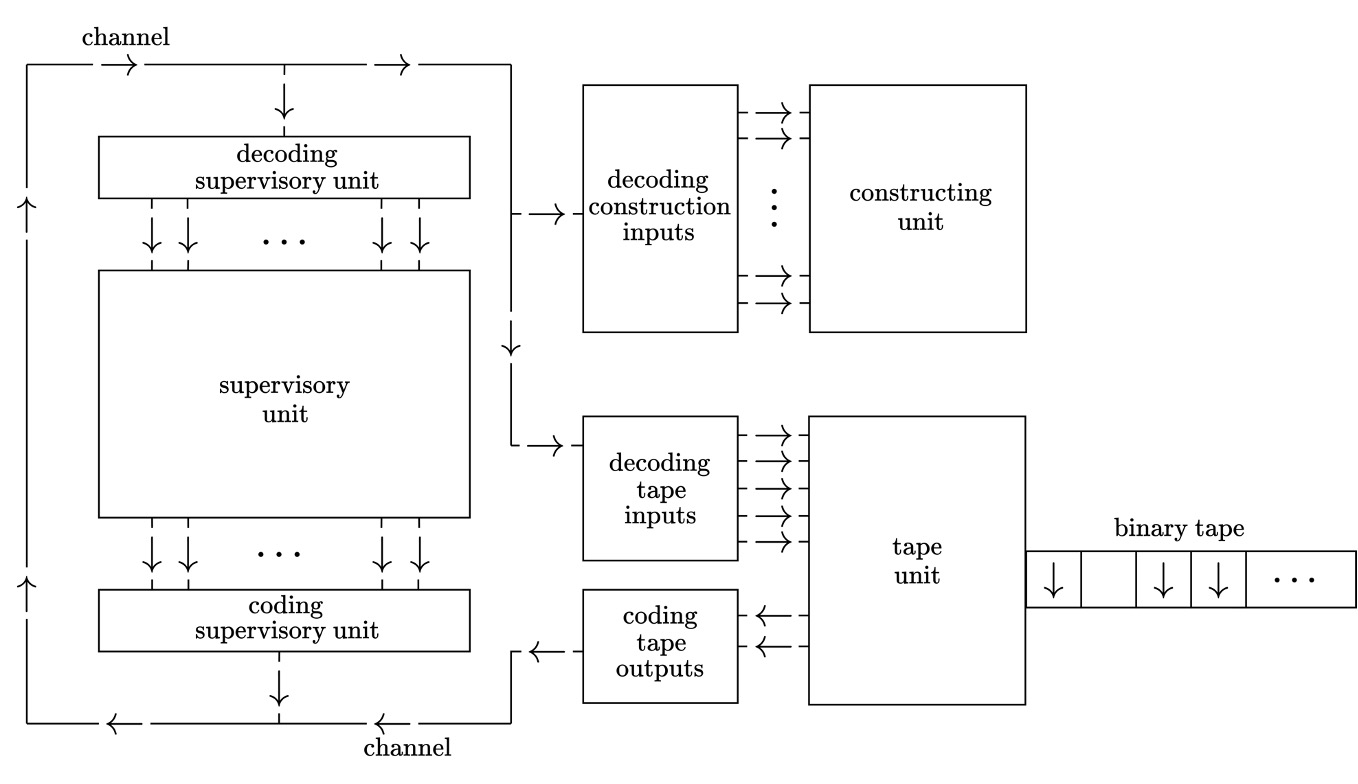}
        \caption{General organization of the VNSR adapted from~\cite{thatcher_construction}.}
        \label{fig:vnsr_scheme}
    \end{figure}

The original VNSR can be made to have the following modes of operation:
\begin{enumerate}
    \item \textbf{Computation:} VNSR performs computation based on its tape instructions. If it halts, we assume the reading head returns to the first symbol of the tape and VNSR moves to phase 2.
    \item \textbf{Self-Replication:} the VNSR scans the tape; upon recognizing the self-replication instructions it starts the self-replication process.
\end{enumerate}

The construction arm can move upwards and to the right from its original position, and can also retract to the left and downwards. As such, it can operate in a quadrant of the plane where the VNSR's copy is built, as illustrated in Fig.~\ref{fig:construction_arm_restriction}. This creates a vertical and horizontal offset, as illustrated in Fig.~\ref{fig:vnsr_general_idea_a}. Our goal is to modify Thatcher's construction such that we get rid of the vertical offset, which will result in creating the ``offsprings'' at the exact same height as the original machine.

\begin{figure}[htbp!]
        \centering
        \includegraphics[width=0.7\linewidth]{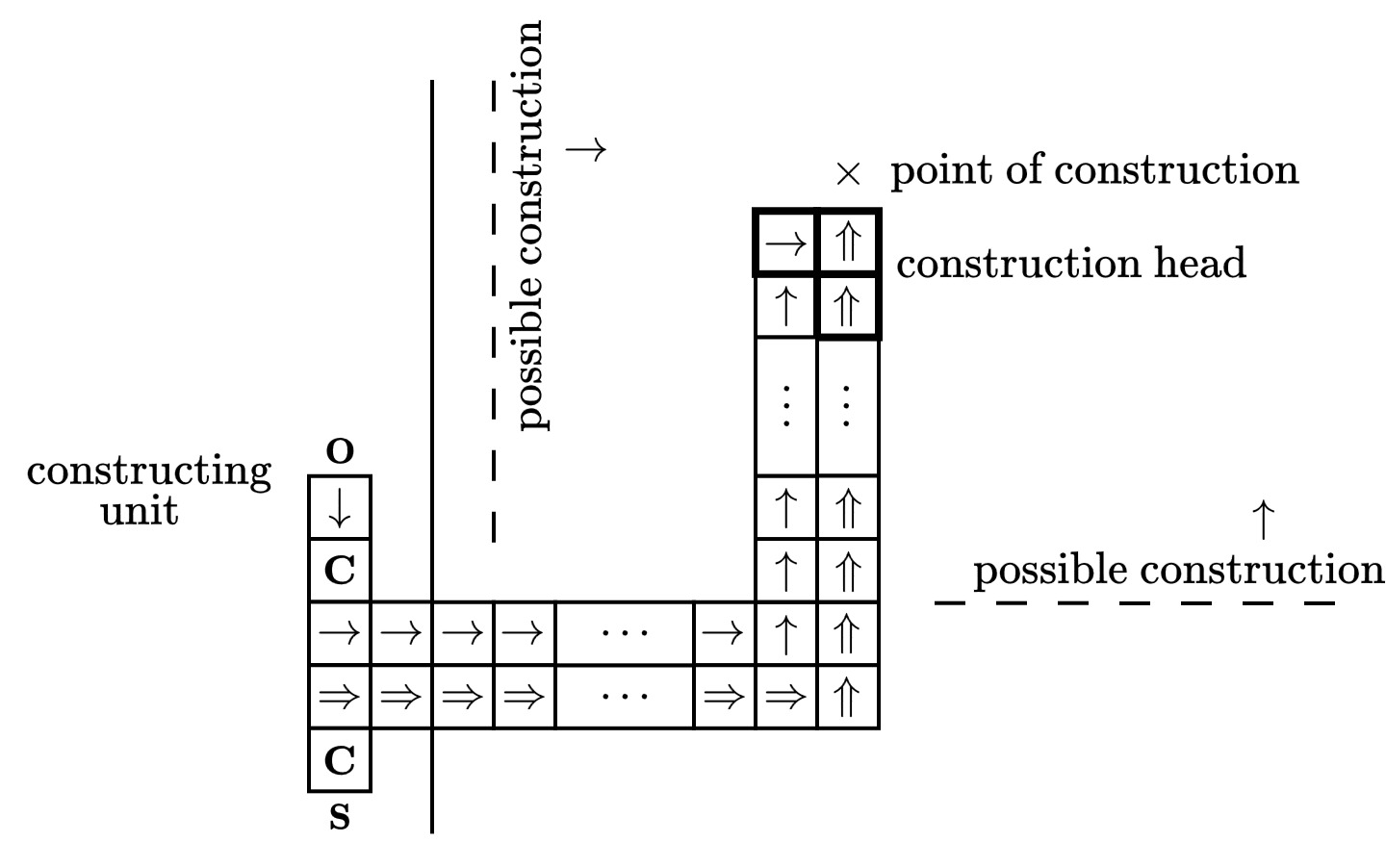}
        \caption{Spatial restriction of the construction arm's movement, adapted from~\cite{thatcher_construction}.}
        \label{fig:construction_arm_restriction}
    \end{figure}
    
\subsection{Modification for the one-dimensional setting}

We modify the original design so that at the beginning of the machine's self-replication phase, the constructing arm has been prolonged so that its initial position is shifted ``right'' and ``downward'' to allow the newly constructed machine to be at the exact same height as its ``parent machine''. Our design of the universal self-replicator will have the following modes of operation:
\begin{enumerate}
    \item \textbf{Computation:} Analogous to the original.
    \item \textbf{Arm extension:} The tape reader reads each tape symbol one by one, progressively moving to the right, upon reaching and recognizing a stop sequence $\omega$ marking the end of the non-empty tape. For each symbol read, the modified supervisory unit instructs the constructing arm to be prolonged by one cell to the right. Once the process is over, we have a guarantee that the constructing arm's head is far enough to the right not to interfere with the tape's content. Afterwards, a special sequence is emitted to the construction arm that prolongs the arm downwards by a constant number of cells, ensuring the constructing arm's new initial position is at the exact height needed.
    \item \textbf{Self-replication:} Analogous to the original.
\end{enumerate}

Concretely, the modification we suggest is quite straightforward and follows the strategy below:
\begin{itemize}
    \item Lemma \ref{lemma:start_of_extension_phase}: The supervisory unit can be modified to accommodate for the extension of the construction arm in Phase 2.
    \item Lemma \ref{lemma:construction_arm_position}: The construction arm can be modified so that at Phase 2., it is possible for it to move to the right and downwards. 
    \item Lemma \ref{lemma:bounded_height} The modifications summarized by the lemmas above result in a modified VNSR whose height is bounded by 1000.
\end{itemize}

Below, we present the strategy in more detail.

\begin{lemma}\label{lemma:start_of_extension_phase}
    The supervisory unit of the VNSR can be modified in such a way that after the computation phase, the VNSR enters a new phase called the extension phase. During this phase, the supervisory unit uses the tape reader to scan the whole tape, and for each symbol, the supervisory unit sends a specific sequence $\sigma$ to the construction unit. Subsequently, the modified supervisory unit returns the reading head to the first tape symbol.
\end{lemma}
\begin{proof}
    We will assume the beginning of the tape is delimited by the sequence $\omega_1$ and the end of the non-empty tape is delimited by $\omega_2$; similar to the original construction by von Neumann \cite{neumann_vnsr} p.~115. Further, we assume the reading head is positioned on the first tape symbol.

    Thatcher's architecture of the supervisory unit is shown in Fig.~\ref{fig:supervisory_unit_scheme}. Each component $C_i$ contains a decoder which, upon recognizing the encoding of its corresponding index $i$ in the channel, activates the component $C_i$. Upon being activated, the component $C_i$ is able to either:

    \begin{figure}[htbp!]
        \centering
        \includegraphics[width=0.6\linewidth]{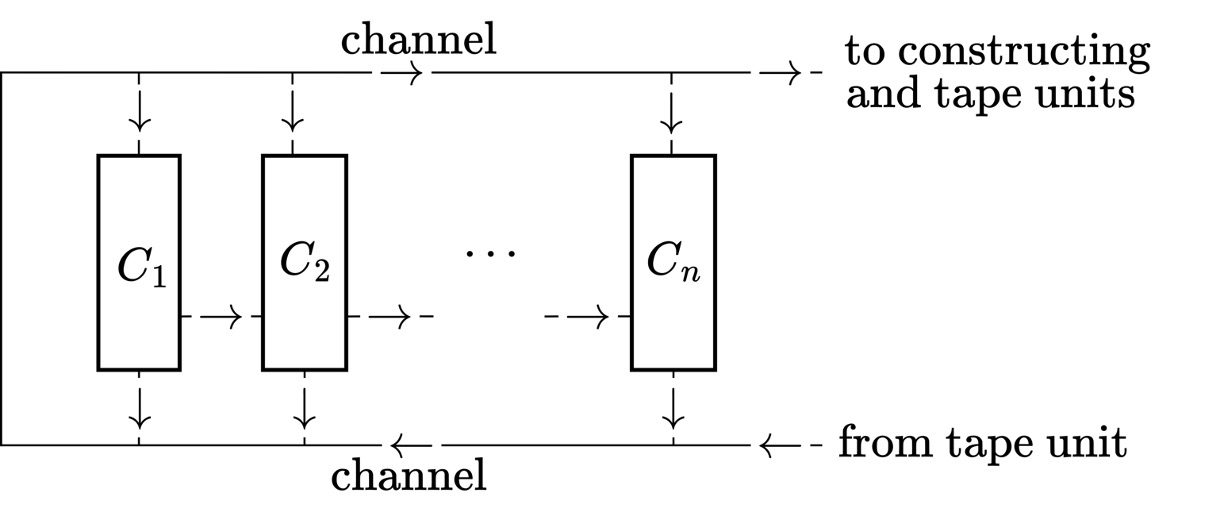}
        \caption{Scheme of the internal organization of the supervisory unit.}
        \label{fig:supervisory_unit_scheme}
    \end{figure}

\begin{itemize}
    \item Send an instruction to the tape unit to read the next symbol, and conditional on the result, broadcast the index $j$ of the next component that should be activated.
    \item Send an instruction to the construction arm, and activate the consecutive unit with index $i+1 \bmod n$.
\end{itemize}
We propose adding three more components, $C_{n+1}$, $C_{n+2}$, $C_{n+3}$ where:
\begin{itemize}
    \item $C_{n+1}$ is activated after the computation phase is finished, and prompts the reading head at position $i$ to read $k$ consecutive symbols from the tape $a_1 \ldots a_k$ where $k = |\omega_2|$, and to return to the tape index $i+1$. $C_{n+1}$ contains a recognizer deciding whether $a_1 \ldots a_k = \omega_2$ marks the end of the tape. If not, $C_{n+1}$ activates the component $C_{n+2}$. If yes, it activates $C_{n+3}$.
    \item $C_{n+2}$ sends the sequence $\sigma$ to the construction unit, which encodes the instruction to extend the construction arm one step to the right. Then, it activates $C_{n+1}$.
    \item $C_{n+3}$ emits a signal to the construction unit with the instruction to extend the construction arm $k_1$ steps to the right, then $k_2$ steps downwards and finally, three steps to the right (where $k_1$ is a constant ensuring the arm will not interfere with the original machine's tape and $k_2$ is a constant corresponding to the height of the original machine). Afterwards, $C_{n+3}$ activates the component $C_1$ which proceeds with returning the reading head to the first tape symbol and with the self-replication phase.
\end{itemize}
Each component $C_i$ consists of a limited number of the following units:
\begin{itemize}
    \item the recognizer $R(\underline{u})$ recognizing the encoding of sequence $u$ form the information channel: either to recognize the sequence $\underline{i}$ marking the activation of the component $C_i$ or to recognize the tape symbols currently read) (each of height 7; \cite{thatcher_construction} p.~33)
    \item the pulser $P(\underline{u})$ which upon activation outputs a series of instructions $u$ into the channel: either for the construction arm or for the reading head (each of height 3; \cite{thatcher_construction}, p.~28) 
    \item A periodic pulser $PP(1)$ which ensures that an activated component $C_i$ stays activated while waiting for the tape symbols being read (height 6, \cite{thatcher_construction}, p.~35)
    \item A delay that ensures proper operation of the tape and constructing units (height at most 300; \cite{thatcher_construction}, p.~95)
\end{itemize}
This allows us to bound the height of the supervisory unit by 500.
\end{proof}

\begin{lemma}\label{lemma:construction_arm_position} At the beginning of the extension phase, the construction arm's initial position can be modified so that the arm can grow down and to the right and afterwards, can create the arm's head to the right of the original VNSR's end of the tape.
\end{lemma}
\begin{proof}
    We modify the initial configuration of the constructing arm from the original Fig.~\ref{fig:construction_arm_init} a) to the new Fig.~\ref{fig:construction_arm_init} b). 

    \begin{figure}[htbp!]
        \centering
        \includegraphics[width=0.35\linewidth]{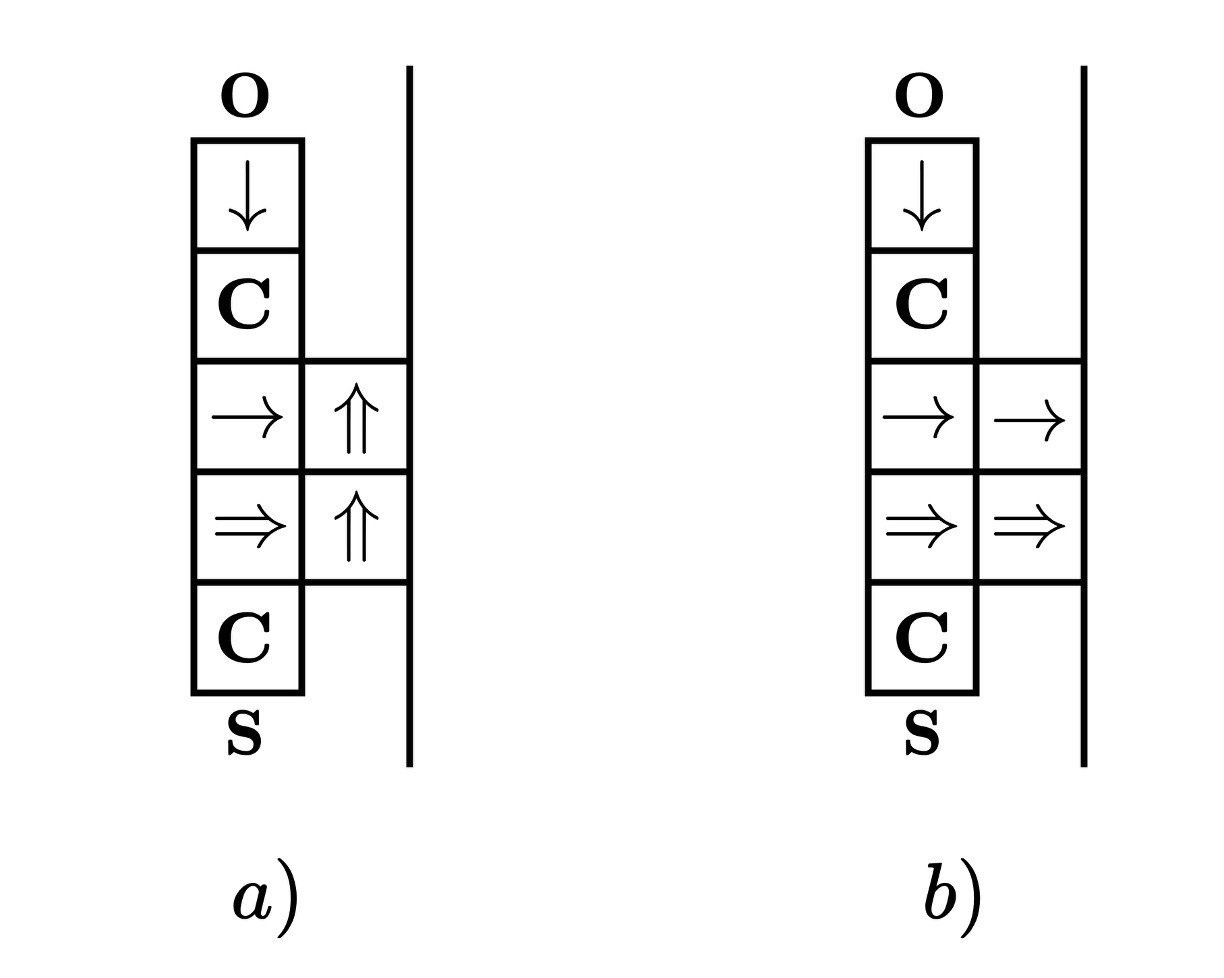}
        \caption{a) Original initial configuration of the constructing arm. b) Modified initial configuration that will result in a shifted constructing arm.}
        \label{fig:construction_arm_init}
    \end{figure} 

    Once the extension process starts, we assume that each non-blank symbol on the tape being read gets converted by the control unit to a signal sequence $\sigma$ arriving at the constructing unit into a pulser generating the sequence 10000 for the ordinary transmission channel input \textbf{O} and the sequence 1011 for the special transmission channel input \textbf{S}, resulting in the arm being prolonged by one cell to the right. After all the non-blank tape symbols are read, the state of the constructing arm is shown in Fig.~\ref{fig:construction_arm_final} a). Once the whole tape is read, a special signal is transmitted to the construction unit, encoding the final horizontal shift of the constructing arm by a constant number of cells, and subsequently prolonging it by a constant number of arrows downwards and to the right, as depicted in Fig.~\ref{fig:construction_arm_final} b). This ensures that the new machine will be built at the exact same height as its parent.
           
    \begin{figure}[htbp!]
        \centering
        \includegraphics[width=0.66\linewidth]{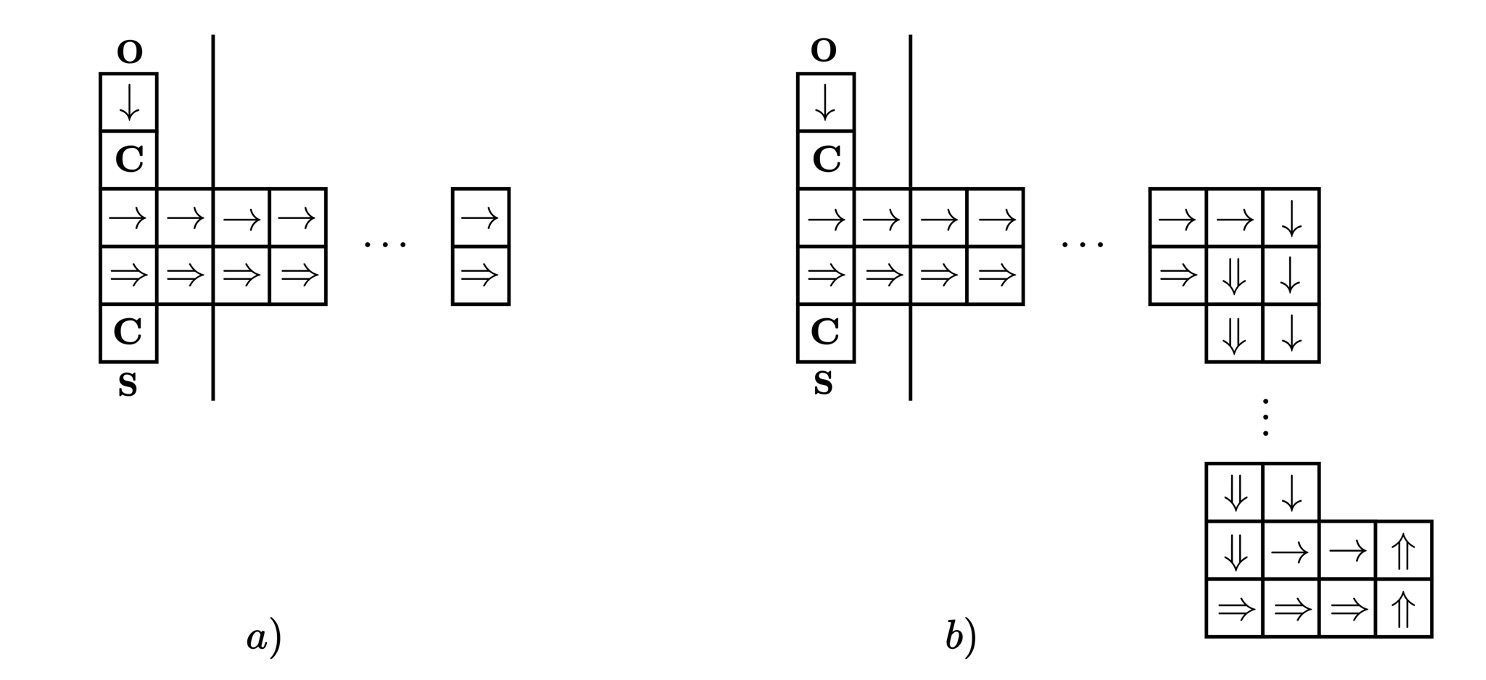}
        \caption{a) State of the constructing arm after all non-blank symbols are read; the length of the arm is proportional to the length of the tape. b) Final state of the shifted constructing arm.}
        \label{fig:construction_arm_final}
    \end{figure} 

    One important aspect of the original construction is the following: if the ordinary sequence of instructions is $i_1 i_2 \ldots i_m$ and the special sequence is $j_1 j_2 \ldots j_m$, then $i_1$ and $j_1$ will arrive at \textbf{O} and \textbf{S} at the same time. In addition, $i_1$ and $j_1$ will also arrive at the corresponding inputs to the constructing arm's head at the same time. We note that this confluence is preserved for the shifted constructing arm, as can be easily verified from Fig.~\ref{fig:construction_arm_final} b). This allows for the rest of the self-replication process to be logically organized in exactly the same way as in the original.
\end{proof}

\begin{remark} In most constructions, the `coding' part of the tape is read-only or otherwise not extended during self-replication.  In this case, the length that we need to extend the arm is the length of the coding part of the tape plus a fixed constant. Otherwise in constructions for which the coding part of the tape is manipulated during self-replications and is extended to a constant multiple of its original length, this is readily accommodated by extending the arm to a constant multiple of the coding tape length plus a universal constant.  (If the coding region is extended in a non-linear fashion this can also be readily accommodated in the extension of the arm.)
\end{remark}

\begin{lemma}\label{lemma:bounded_height} The VNSR resulting from the modification of previous lemmas stays within a strip of height at most 2000.
\end{lemma}
\begin{proof}
    First, we discuss an upper bound for the original VNSR's height. In~\cite{thatcher_construction}, Thatcher specifies the following sizes: 
\begin{itemize}
    \item tape control unit: height of 83 units, length of 107 units
    \item constructing unit: height of 33 units, length of 139 units
\end{itemize}
Though the supervisory unit's size is not specified, we discussed its upper bound to be 500 in Lemma \ref{lemma:start_of_extension_phase}. Thus, the original VNSR's height can be upper bounded by 1000.

Lemma \ref{lemma:start_of_extension_phase} requires the modification of the supervisory unit by adding extra components. As those are chained horizontally, this does not change the supervisory unit's height. Since the newly added components contain a few extra pulsers and recognizers to accommodate for their extended functionality, we propose an upper bound for the modified supervisory unit's height to be 1000.
Lemma \ref{lemma:construction_arm_position} does not require any change of the internal organization of the constructing unit, apart from the modification of the initial constructing arm state, which does not alter its height.

To summarize, a generous upper bound for the height of the modified VNSR that we propose is 2000.
\end{proof}

\begin{theorem} \label{thm:1d_universal_self_replicator}
There exists a 1D universal self-replicator.
\end{theorem}
\begin{proof}
 From Lemma \ref{lemma:bounded_height} we have the existence of a 2D CA with 29-states (the von Neumann CA) and an initial pattern of height 2000 which represents a universal self-replicator. From the construction, it is apparent that during the whole process of self-replication, any cell outside of the strip of height 2000 stays within the quiescent state. Thus, we can define a 1D CA with $29^{2000}$ states and nearest neighbors where:
 \begin{itemize}
     \item each state represent a vertical column of height 2000 in the von Neumann CA;
     \item the local rule $f(a, b, c)$ mimics one step of the von Neumann CA applied to a $2000\times3$ rectangle with columns represented by $a, b$, and $c$ surrounded by the quiescent state, and outputs the outcome in the state corresponding to the middle column of the updated rectangle.
 \end{itemize}
\end{proof}

\section{Local self-replicating condition and its generalizations}
In this section, our goal is to formalize necessary conditions for universal self-replication. The conditions we provide are by no means sufficient; indeed, giving a formal characterization of non-trivial self-replicating CAs is a challenging open problem. However, these conditions will be very useful for the next section where we show the existence of a locally universal CA that cannot be universally self-replicating.

Informally, we first identify which CA patterns could be perceived as ``organisms'': very loosely, an organism is any sequence of finite patterns that define the organism's canonical activity. For instance, the most classic glider in Game of Life periodically traverses four different shapes as it moves through the space; see Fig.~\ref{fig:gol_glider}. 

\begin{figure}[htbp!]
    \centering
    \includegraphics[width=0.5\linewidth]{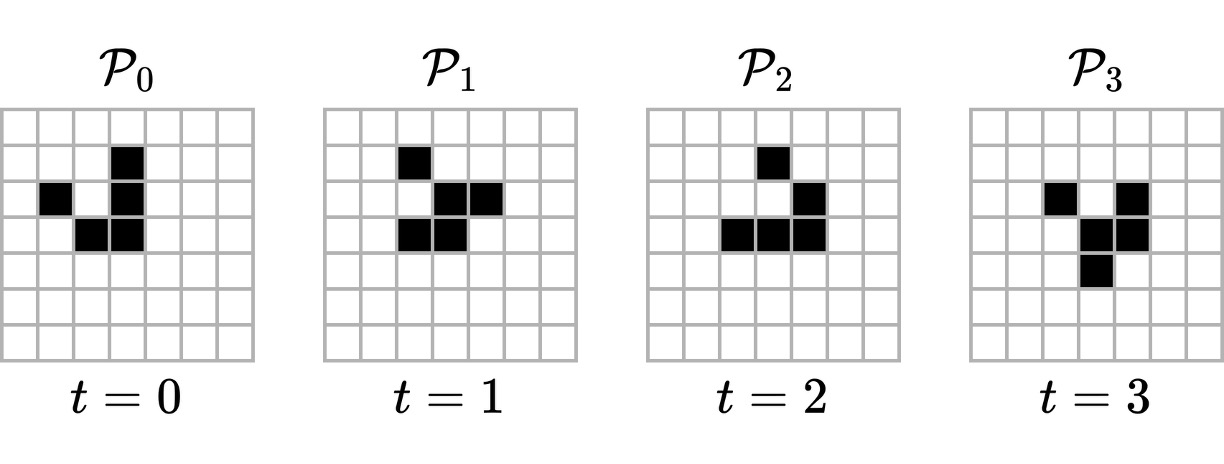}
    \caption{A classical glider in Game of Life traverses four basic shapes $(\mP_0, \mP_1, \mP_2, \mP_3)$ as it moves.}
    \label{fig:gol_glider}
\end{figure}

We say that such an ``organism'' is self-replicating, if throughout the CA's iterations, it generates arbitrarily many copies of itself. To exclude cases of trivial patterns, such as a single black cell on a white background growing progressively in all directions, we further require the ``organism'' to contain a component that has a temporal period of length at least two. To further build the intuition behind non-trivial self-replication, we introduce two famous examples from the literature.

\subsection{Examples of non-trivial self-replication}

\begin{example}[Langton's loop]
Langton's loop  \cite{langton1984self} is a seminal construction of a non-universal self-replicating CA structure which is nevertheless non-trivial. The loop consists of an external sheath colored red in Fig. \ref{fig:langton_loop}. Inside the sheath there is a circulating sequence of states encoding the loop itself. Whenever there is enough space next to the loop, the whole pattern self-replicates. Once the loop runs out of space, it ``dies out'' and becomes a static pattern. 

\begin{figure}[htbp!]
    \centering
    \includegraphics[width=0.75\linewidth]{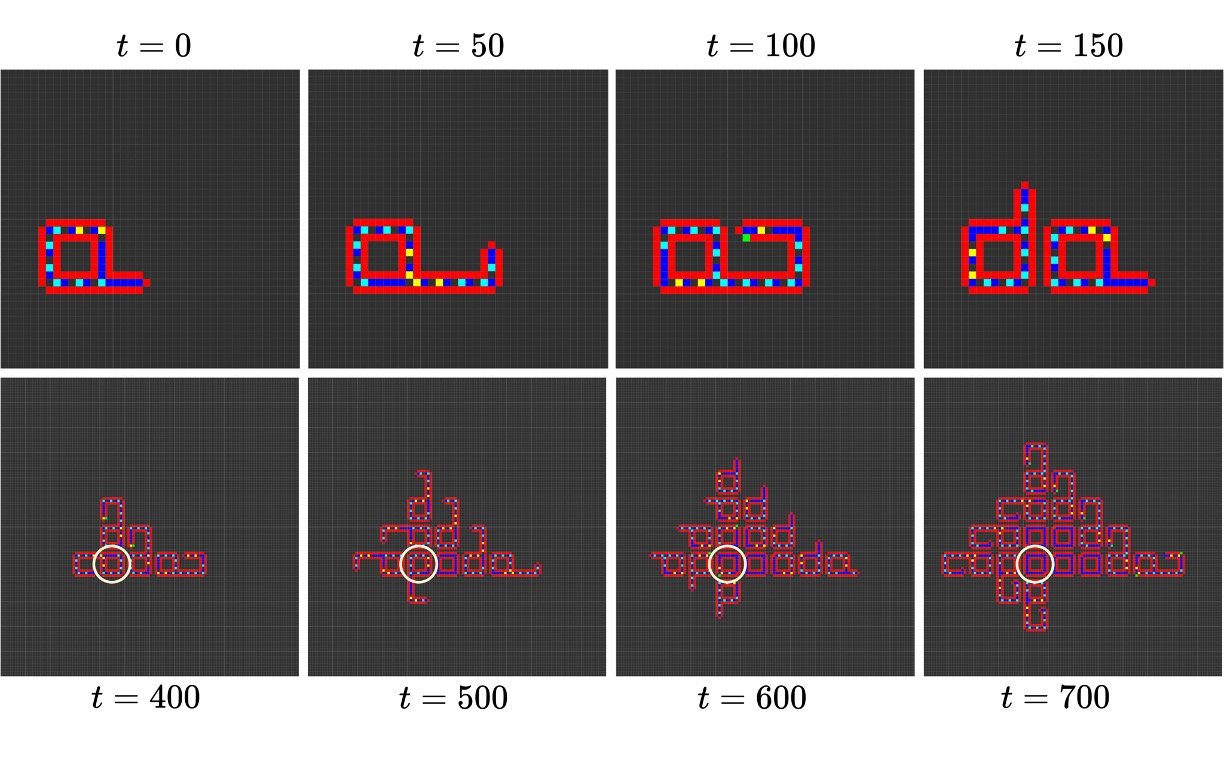}
    \caption{Langton's loops \cite{langton1984self}. (First row) Initial process of self-replication. (Second row) Original loop is highlighted by a circle, and eventually becomes static.}
    \label{fig:langton_loop}
\end{figure}
\end{example}

\begin{example}[Byl's loop]
Byl's loop  \cite{byl1989self} is a further simplification of Langton's construction. By removing the inner sheath, Byl reduced the loop's size from 86 to 12, and the replication period from 151 to 25. Once the loop runs out of surrounding space, it enters a temporal period of size 4.
    \begin{figure}[htbp!]
    \centering
    \includegraphics[width=0.75\linewidth]{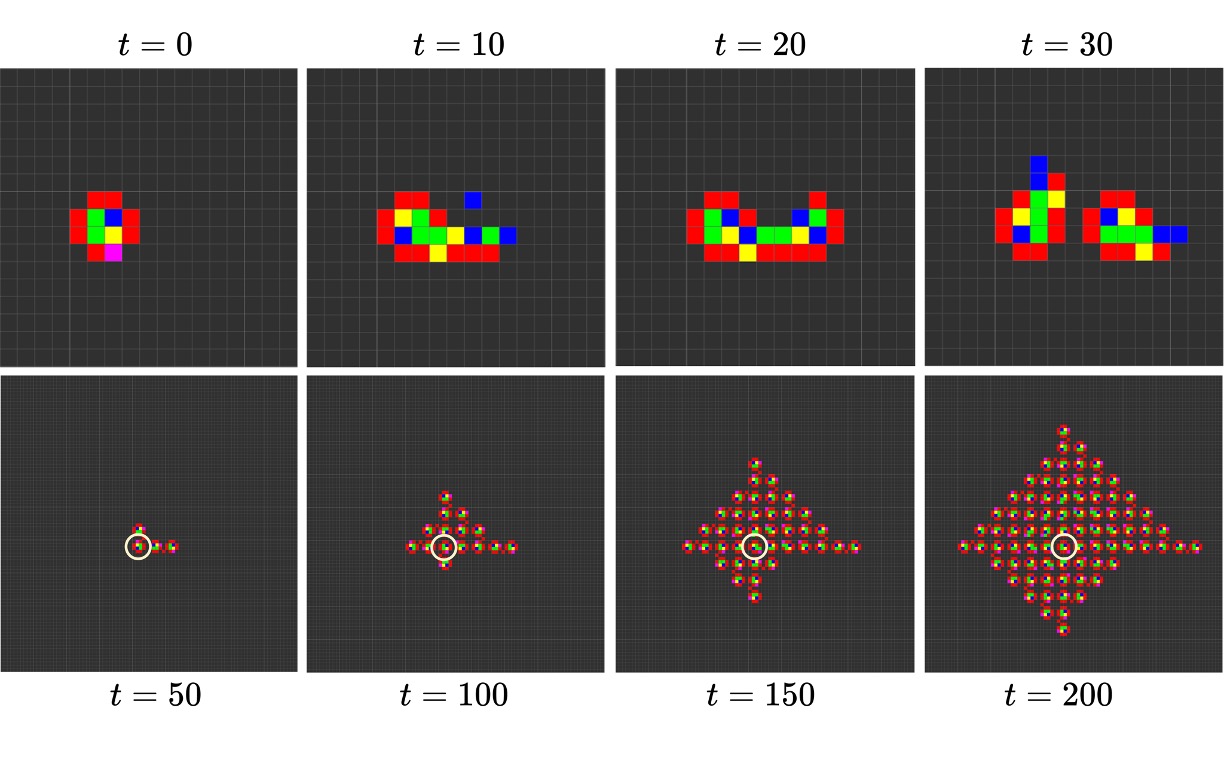}
    \caption{Byl's loops \cite{byl1989self}. (First row) Initial process of self-replication. (Second row) Original pattern is highlighted by a circle, it eventually reaches a temporal cycle of size 4.}
    \label{fig:byl_loop}
\end{figure}
\end{example}

\noindent We also comment on:

\begin{example}[von Neumann's universal self-replicator]
As previously discussed, von Neumann's universal self-replicator~\cite{neumann_vnsr} can be programmed to replicate itself (and then subsequently perform additional computation, if these instructions are appended to the self-replication instructions).  Then the subsequent copies of the machine will likewise self-replicate in an identical fashion.
\end{example}

\subsection{Formulating necessary conditions for self-replication}

First, we say that each ``CA organism'' is in part characterized by a sequence of patterns $\bm{\mP} = \left (\mP_0, \mP_1, \ldots, \mP_{T-1} \right)$ which describe its ``canonical form of activity''. In order for the organism to self-replicate, this activity (e.g.~the regular motion of a subsystem of the organism, or even the dynamics of self-replication itself) has to be present in a growing number of disjoint CA regions throughout the CA's iterations. 

For this, we have to address the issue of how to count the number of an organism's copies in a CA configuration. To account for various examples of non-trivial self-replication, it becomes clear that one should allow basic geometrical transformations, such as rotations, when identifying two organisms as identical. This becomes apparent, for instance, from the example of Langton's loops. We formally define this below.

\begin{definition}[Equivalence of local patterns]
Let $\A = (S^{\Z^d}, F)$ be a cellular automaton, $R \subset \Z^d$ a finite region, and $T \in \N$. Consider $\bm{\mP} = \left (\mP_0, \mP_1, \ldots, \mP_{T-1} \right)$ a sequence of pairwise distinct $R$-patterns; i.e.~for each $0 \leq i < T$ we have $\mP_i: R \rightarrow S$. We will refer to $\bm{\mP}$ as a local pattern sequence with temporal length $T$, or just a local pattern.

Let $R' \subset \Z^d$ and let $\bm{\mP}' = \left (\mP'_0, \mP'_1, \ldots, \mP'_{T-1} \right)$ be another sequence of disjoint $R'$-patterns with the same temporal length $T$. We say that $\bm{\mP}$ and $\bm{\mP}'$ are equivalent, if there exists an isometry $\varphi: \Z^d \rightarrow \Z^d$ and an integer $0 \leq \sigma < T$ such that:
    $$ \mP_i = \mP'_{\sigma + i \bmod T} \circ \varphi  \quad \quad \quad \text{for all } 0 \leq i < T.$$

Let $c \in S^{\Z^d}$ be a configuration of $\A$ and let $k \geq 1$ be an integer. We say that $c$ contains $k$ copies of $\bm{\mP}$ if there exist $k$ disjoint regions $R_1, \ldots, R_k \subset \Z^d$ such that $(\restr{c}{R_i}, \restr{F(c)}{R_i}, \ldots, \restr{F^{T-1}(c)}{R_i})$ defines a local pattern equivalent to $\bm{\mP}$ for all $1 \leq i \leq k$. We say that $c$ contains exactly $k$ copies of $\bm{\mP}$, if it contains $k$ copies of $\bm{\mP}$ and it does not contain $k+1$ copies of $\bm{\mP}$.
\end{definition}

Now we are ready to present the definition of a self-replicating local pattern. 

\begin{definition}[Self-replicating local pattern]\label{def:selfreplocal1}
Let $\A = (S^{\Z^d}, F)$ be a cellular automaton, $R \subset \Z^d$ a finite region, and $T \in \N$. Consider $\bm{\mP} = \left (\mP_0, \mP_1, \ldots, \mP_{T-1} \right)$ a local pattern. Let $c \in S^{\Z^d}$ be a configuration of $\A$. We say that $\bm{\mP}$ self-replicates with configuration $c$ if there exist two sequences of positive integers:
\begin{alignat*}{3}
    n_1 &< n_2 &&< n_3 &&< \cdots\\
    k_1 &< k_2 &&< k_3 &&< \cdots
\end{alignat*}
such that $F^{n_i}(c)$ contains precisely $k_i$ copies of $\bm{\mP}$. Moreover, we require $c$ to be a ``computationally feasible'' configuration; more precisely, we require that $c$ agrees with an admissible background $B$ (according to Definition \ref{def:admissible_background}) outside of a finite region.
\end{definition}

We note that infinite regions of an admissible background may evolve into spatially periodic, time-periodic regimes, in which a given periodic local pattern appears infinitely many times in parallel. Such patterns should not be regarded as genuine self-replicators, since they violate the requirement that the number of copies increase in a controlled manner. The following lemmas makes this precise.

\begin{lemma}[Eventual spatiotemporal periodicity of periodic configurations]\label{lemma:periodic}
Let $\A = (S^{\Z^d}, F)$ be a $d$-dimensional cellular automaton. Suppose the initial configuration $c \in S^{\Z^d}$ is totally spatially periodic, i.e.~there exist linearly independent vectors $v_1, \ldots, v_d \in \Z^d$ such that $\sigma_{v_i}(c) = c$ for all $i$. Then the trajectory of $c$ under $F$ eventually becomes periodic in both space and time. In particular, there exist integers $T \geq 1$ and $\widetilde{T} \geq 0$ such that
\begin{align}
F^{t+T}(c) = F^t(c) \quad \text{for all }\, t \geq \widetilde{T},
\end{align}
and so every $F^t(c)$ remains totally spatially periodic with respect to $v_1,\dots,v_d$ and is eventually periodic in time.
\end{lemma}

\begin{proof}
Since $F$ commutes with translations, total spatial periodicity is preserved along the orbit: if $\sigma_{v_i}(c)=c$ for all $i$, then $\sigma_{v_i}(F(c))=F(c)$ for all $i$. Hence the orbit of $c$ lies inside the set $K$ of configurations with the fixed spatial periods determined by the $v_i$’s. Choosing a fundamental domain for this periodic lattice, each configuration in $K$ is determined by finitely many cells, so $K$ is finite. The restriction of $F$ to $K$ is therefore a self-map of a finite set, and the orbit of $c$ in $K$ is eventually periodic. That is, there exist integers $\widetilde T \geq 0$ and $T \geq 1$ with $F^{t+T}(c)=F^t(c)$ for all $t \geq \widetilde T$. This proves the claim.
\end{proof}

\begin{lemma}\label{lemma:nobackgroundPs}
Let $B$ be an admissible background in the sense of Definition~\ref{def:admissible_background}, realized by a finite partition $\mathbb Z^d=\bigsqcup_{i=1}^k R_i$ and basic backgrounds $B_1,\dots,B_k$ with $\restr{B}{R_i}=\restr{B_i}{R_i}$. Assume some $R_i$ is unbounded and the CA dynamics maps $B_i$ after a transient time $\widetilde T$ into a configuration that is totally spatially periodic and time-periodic with finite period $T$ (as guaranteed by Lemma~\ref{lemma:periodic}). Let $c$ be a configuration that agrees with $B$ outside a finite region $Q \subset \Z^d$. Let $\bm{\mathcal P}$ be any local pattern that occurs within this time-periodic regime and reappears every $T$ steps. Then $\bm{\mathcal P}$ cannot be a self-replicating local pattern for $c$. 
\end{lemma}

\begin{proof}
Let $r$ be the interaction radius of the CA. Since $R_i$ is convex and unbounded, for any $m \in \mathbb N$ there exists an unbounded convex subset $R_i^{(m)} \subset R_i$ which is at $\ell_\infty$-distance at least $rm$ from both $\partial R_i$ and from the finite region $Q$. Take $m=\widetilde T+\ell$ for $\ell \in \{0,1,\dots,T-1\}$. Then for every $t \geq \widetilde T+\ell$ with $t \equiv \widetilde T+\ell \pmod{T}$, the dynamics inside $R_i^{(m)}$ coincides with the evolution one would see if $R_i$ were the whole lattice and there were no perturbation in $Q$, since no influence can reach from $\partial R_i$ or from $Q$ within $t$ steps. 

By hypothesis, after time $\widetilde T$ the restriction of $B_i$ is spatially periodic and time-periodic with period $T$, so at each such time $t$ the unbounded region $R_i^{(m)}$ contains infinitely many disjoint occurrences of $\bm{\mathcal P}$. Therefore, for each residue class modulo $T$ and all sufficiently large times, the configuration contains infinitely many copies of $\bm{\mathcal P}$. 

However, Definition~\ref{def:selfreplocal1} requires an infinite increasing sequence $n_1 < n_2 < n_3 < \cdots$ with the configuration containing exactly $k_j < \infty$ copies of $\bm{\mathcal P}$ at time $n_j$. Since after time $\widetilde T$ every (large) time in each residue class contains infinitely many copies, only finitely many times remain with finitely many copies. Thus no such infinite sequence $n_1 < n_2 < n_3 < \cdots$ can exist, and $\bm{\mathcal P}$ is not self-replicating.
\end{proof}

We highlight several crucial details of Definition~\ref{def:selfreplocal1}. Our definition aims to partially capture a certain causality: the patterns present at time $n_i$ should be the ones that give rise to the patterns at time $n_{i+1}$, or at least have a common progenitor or constructor. This is in contrast with the undesirable situation where infinitely many organisms are already encoded in the background $c$ and appear at predetermined times. Our definition addresses this concern in two ways: by imposing that at time $n_i$ there are precisely $k_i$ organisms present (and no more than that), and by requiring the background $c$ to be an admissible configuration. The latter requirement significantly mitigates the possibility of pre-encoded organisms, as $c$ is simple enough not to encode a complicated infinite set of instructions to progressively generate more and more patterns throughout the CA's iterations.  In particular, Lemma~\ref{lemma:nobackgroundPs} shows that periodic motifs which develop in certain infinite regions of the admissible background cannot qualify as self-replicating local patterns. At last, we present the necessary condition of self-replication below. Therein, we further exclude the case of static organisms with temporal length 1.

\begin{definition}[Local self-replicating condition]\label{def:local_selfrep_cond}
    We say that a cellular automaton $\A = (S^{\Z^d}, F)$ satisfies the locally self-replicating condition if there exists a local pattern $\bm{\mP} = \left (\mP_0, \mP_1, \ldots, \mP_{T-1} \right)$ with $T \geq 2$ which self-replicates.
\end{definition}

\begin{remark}
It is now easy to check that both Langton's and Byl's loops satisfy the local self-replicating condition, as does von Neumann's universal self-replicator (e.g.~when it is programmed to replicate).
\end{remark}

We conclude by providing additional justification for our definitions and highlighting their key features. The local self-replicating condition represents a biologically plausible model for non-trivial asexual reproduction, as identical organisms will have dynamical motifs in common.  A crucial feature of our definition is that it does not track individual organisms or require their perpetual existence; organisms may die, and we only require that the population proliferates across arbitrary time horizons under suitable conditions in infinite volume. This infinite-volume assumption is a mathematical convenience that allows for unbounded organism proliferation, reflecting the principle that any organism should theoretically produce arbitrarily many offspring given sufficient resources. Finally, we emphasize again that our local self-replicating condition is necessarily satisfied by universal self-replicators; the reason is that their offspring will themselves replicate in an identical fashion, and so the process of replication itself forms a local pattern sequence.

These considerations suggest that more stringent necessary conditions for non-trivial self-replication could be formulated. While such additional conditions are not required for our construction of a universal CA that cannot sustain universal self-replication, we briefly discuss several possibilities below.

\subsection{Comments on more general local self-replicating conditions}

Our self-replicating condition imposes only mild constraints on the dynamical behavior of a would-be organism, requiring merely that it exhibit a dynamical motif comprising a local pattern of length at least two. Inspired by Langton's loops and their progressively complex generalizations, culminating in the Turing-universal Perrier-Sipper-Zahnd loops, we can envision strengthened conditions that require the local pattern to exhibit more sophisticated forms of computation. In other words, while our existing definition merely requires that an organism `do something,' we could further demand that it `do something interesting,' such as implementing specific forms of computation. Following our discussions of local universality in CAs, one natural extension would be to require that organisms implement computation calibrated by an encoder-decoder pair, analogous to computational dynamical systems~\cite{computational_dynamical_systems}. While we do not pursue these more stringent versions of the local self-replicating condition here, we view this as a promising direction for future work, which we discuss further in Appendix~\ref{sec:openproblems}.

\subsection{Universal self-replicators}

As we discussed at the beginning of this section, giving characterizing conditions for non-trivial self-replication remains an open problem. Thus, we can only informally describe $\mathsf{UniversalSelfReplicating}$ as the set of all cellular automata (in any dimension) that are locally universal and capable of non-trivial self-replication. We assert that any future feasible definition of non-trivial self-replication should satisfy the following properties.

\begin{claim}
Any cellular automaton belonging to the $\mathsf{UniversalSelfReplicating}$ class has to satisfy the local self-replicating condition in Definition \ref{def:local_selfrep_cond}. Moreover, von Neumann's universal self-replicator has to belong to $\mathsf{UniversalSelfReplicating}$.
\end{claim}

It is further natural to require that any CA that can globally simulate an automaton capable of non-trivial self-replication should itself be capable of non-trivial self-replication.  Indeed, such CAs can simulate a universal self-replicating CA via a local encoding, as ensured by the existence of the CA specified by our Theorem~\ref{thm:1d_universal_self_replicator}.

\begin{claim}
    $\mathsf{GloballyUniversal} \subseteq \mathsf{UniversalSelfReplicating}.$
\end{claim}

Lastly, we argue that the above inclusion is strict. Since no reversible cellular automaton can be globally universal, it suffices to prove the following.

\begin{theorem}
\label{thm:universalnotglobal}
There exists a reversible cellular automaton belonging to $\mathsf{UniversalSelfReplicating}$, and thus
\begin{align}
\mathsf{GloballyUniversal} \subsetneq \mathsf{UniversalSelfReplicating}\,.
\end{align}
\end{theorem}
\begin{proof}
    We simply use Toffoli's result from \cite{toffoli1977computation}, showing that for every $d$-dimensional CA, there exists a $d+1$-dimensional reversible CA that can globally simulate it. Thus, there exists a 3-dimensional reversible CA globally simulating von Neumann's universally self-replicating automaton. Combining previous claims, we get that this CA belongs to $\mathsf{UniversalSelfReplicating}$.
    Furthermore, in the previous section, we proved the existence of a 1D CA belonging to $\mathsf{UniversalSelfReplicating}$. This immediately implies the existence of a 2D reversible CA belonging to $\mathsf{UniversalSelfReplicating}$.
\end{proof}

\section{Non-talking heads cellular automata}

In this Section, we construct an example of a 1D CA that is locally Turing-universal and yet cannot sustain universal self-replication. This counters the intuition that local Turing universality of a CA is sufficient to imply the existence of universal self-replicating configurations. Our strategy is to first construct a nice family of locally Turing-universal CAs, and then to modify the construction so we can prove that no configuration can satisfy the periodic self-replicating condition, thereby ruling out universal self-replicating configurations. Before proceeding, we draw a conceptual connection to biological systems to clarify the rationale for our construction.

\subsection{A biological analogy}

In life on Earth, DNA participates in at least two distinct processes. First, DNA comprises a description of an organism and participates in replication, where that description is copied. Second, DNA enables protein synthesis via the central dogma of molecular biology: DNA makes RNA (via transcription), and RNA makes protein (via translation). These proteins, combined with ambient chemical processes, build complex biomolecular structures that serve as machines and can be viewed as facilitating or instantiating computation and information processing. The key point is that replication and protein synthesis are two distinct processes; we can imagine one without the other. Specifically, we can envision a system capable of building protein structures but with instabilities that preclude self-replication. Indeed, many organisms, whether by mutation or other causes, lack the ability to reproduce. It is therefore not difficult to imagine a system where a judicious choice of physical laws makes self-replication not merely impractical, but impossible altogether. We pursue this line of reasoning below in a concrete example.

\subsection{Absence of self-replication in non-talking heads CAs}

It is first useful to recall some of our notation for Turing machines, as discussed above in Definition~\ref{def:Turing0}.  Let $Q$ be a finite set of states and $\Sigma$ be a finite set of symbols.  Then we define the transition function by
\begin{align}
\delta : Q \times \Sigma \to Q \times \Sigma \times \{\text{L, R, S}\}
\end{align}
where $\text{L}$ stands for `left', $\text{R}$ stands for `right', and $\text{S}$ stands for `stay'.  We can write out the function as
\begin{align}
\delta(q,s) = (\delta_Q(q,s),\,\delta_\Sigma(q,s),\,\delta_{\text{move}}(q,s))\,.
\end{align}
Suppose that $\delta$ corresponds to a universal Turing machine.  Now let us define a 1D $4 (|Q|+1) |\Sigma|$-state CA as follows.

We label each state on each site of the CA by a state $(h,s,a)$ where $h \in Q \cup \{\sqcup\}$, $s \in \Sigma$, and $a \in \{\text{L}, \text{R}, \text{S}, \sqcup'\}$.  (Here $h$ stands for `head', $a$ stands for `arrow', and $\sqcup$ and $\sqcup'$ are blank or empty.)  We take one of the elements of $Q$ to be the halting state, denoted by $\textsf{halt}$.  The transition rule of the CA only depends on nearest-neighbors, and is defined as follows:
\begin{align}
\label{E:fCAdef1}
&f((h_1,s_1,a_1),\,(h_2,s_2,a_2),\,(h_3, s_3,a_3)) \\
& = \begin{cases}
 (\delta_Q(h_1, s_1),s_2, a_2) &\text{if }h_1 \not = \sqcup,\,h_2 = h_3 = \sqcup,\, \text{and }\delta_\text{move}(h_1,s_1) = \text{R}
 \\
 (\delta_Q(h_3, s_3),s_2, a_2) &\text{if }h_3 \not = \sqcup,\,h_1 = h_2 = \sqcup,\, \text{and }\delta_\text{move}(h_3,s_3) = \text{L} \\
 (\delta_Q(h_2, s_2), \delta_{\Sigma}(h_2,s_2), \delta_{\text{move}}(h_2, s_2)) &\text{if }h_2 \not = \sqcup,\,h_1 = h_3 = \sqcup,\, \text{and }\delta_\text{move}(h_2,s_2) = \text{S} \\ 
  (\sqcup, \delta_{\Sigma}(h_2,s_2), \delta_{\text{move}}(h_2, s_2)) &\text{if }h_2 \not = \sqcup,\,h_1 = h_3 = \sqcup,\,\delta_\text{move}(h_2,s_2) = \text{L},\,\text{and }a_1 = \text{R or }\sqcup'  \\ 
    (\sqcup, \delta_{\Sigma}(h_2,s_2), \delta_{\text{move}}(h_2, s_2)) &\text{if }h_2 \not = \sqcup,\,h_1 = h_3 = \sqcup,\, \delta_\text{move}(h_2,s_2) = \text{R},\,\text{and }a_3 = \text{L or }\sqcup' \\ 
    (h_2, s_2, a_2) &\text{if }h_1 = h_2 = h_3 = \sqcup \\
   (\textsf{halt}, s_2, a_2) &\text{otherwise} \nonumber
\end{cases}
\end{align}
The above describes a stationary tape with moving heads, where each head records on each tape cell the direction that head is moving next.  There are two additional key features: (i) if there are two or more adjacent heads, then those heads will enter the halting state; (ii) if a head is slated to move and the cell to which it is moving has an $a$ which has an inconsistent `parity', then that head will enter the halting state.  This mechanism (ii) prevents any two heads from `communicating' with one another.  As such, we will refer to the CA as a ``non-talking heads'' CA.  We depict the basic mechanism described here in Fig.~\ref{fig:non-talking}.

\begin{figure}[htbp!]
    \centering
\includegraphics[width=0.75\linewidth]{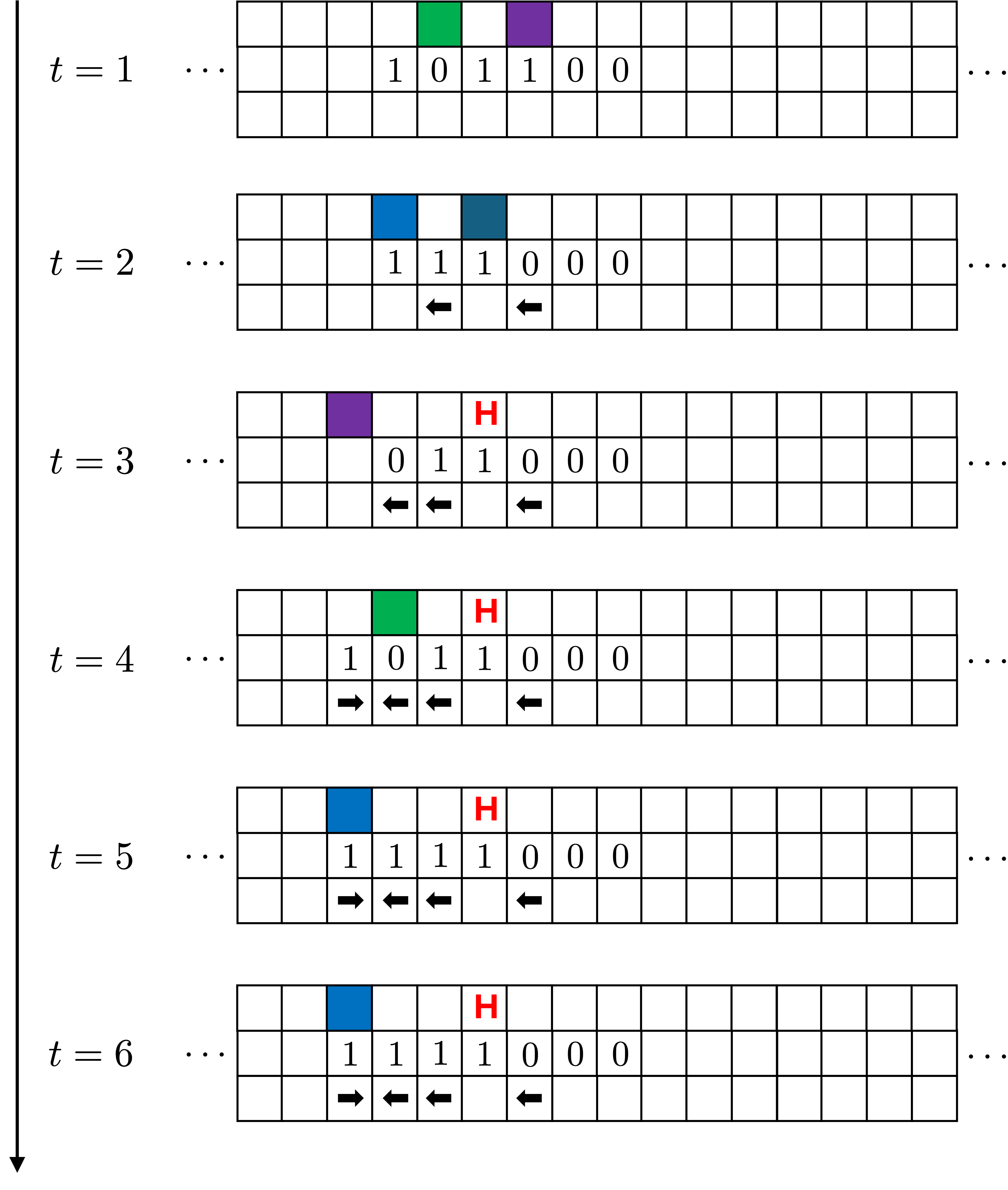}
    \caption{A depiction of the dynamics of the non-talking heads CA, with time running from top to bottom.  At each moment in time, we `unfold' a single cell into three vertical cells, corresponding to $h, s, a$ from top to bottom.  Blank cells correspond to the empty state or symbol; colored cells in the top rows correspond to heads in some particular state (with $\textcolor{red}{\textbf{H}}$ denoting the halting state), $0$'s and $1$'s in the middle rows correspond to symbols on the tape, and arrows $\leftarrow$ and $\rightarrow$ in the bottom rows corresponding to $a = \text{L}$ and $a = \text{R}$, respectively.
    At time $t = 1$, we see two (activated) heads and six non-empty symbols.  At time $t = 2$, the heads have moved to the left, affecting symbols and leaving arrow markers to reflect their direction of movement. At time $t = 3$ the right-most head has halted since it ran into an arrow to its left with a parity indicating that the arrow was laid down by another head.  The left-most head continues to transition, affect symbols, and lay down arrows, as we see at times $t = 4$, $t = 5$, and $t = 6$.}
    \label{fig:non-talking}
\end{figure}

The non-talking heads CA defined above is locally Turing-universal by construction, as recorded in the following lemma.
\begin{lemma}[Local universality of the non-talking heads CA]
The 1D non-talking heads CA defined above is locally universal.
\end{lemma}
\begin{proof}
The CA directly simulates a universal Turing machine, for example if there is only one head (i.e.~only one cell is not in a state $(\sqcup,s,a)$) and the rest of the tape elements are initialized with $a = \sqcup'$.
\end{proof}

Now we will demonstrate that the non-talking heads CA does not have any configurations that satisfy the local self-replicating condition.  To do so, we will utilize the following lemma.
\begin{lemma}[Trapping lemma for non-talking heads]
\label{lemma:trapping}
Label the sites of the CA by $\mathbb{Z}$. Consider a configuration $c$ of the CA and list the heads present at time $t=0$ in strictly increasing order of position
\begin{align}
\cdots < h_{-2}(0) < h_{-1}(0) < h_0(0) < h_1(0) < h_2(0) < \cdots ,
\end{align}
where $h_0(0)$ labels the first head at position greater than or equal to zero in $\mathbb{Z}$. Let $h_i(t)$ denote the position at time $t$ of the (unique) head with index $i$. Then for all $t\ge 0$,
\begin{align}
h_{i-1}(0) < h_i(t) < h_{i+1}(0)\,.
\end{align}
\end{lemma}
\begin{proof}
Consider the trajectory of a particular head $h_i$, which we take to be the half-infinite sequence $(h_i(0), h_i(1), h_i(2),...)$ where $h_i : \mathbb{Z}_{\geq 0} \to \mathbb{Z}$ is a map from times $t \geq 0$ to locations of cells of the CA.  If $h_i(t)$ encounters either (i) another head, or (ii) an $a$ symbol with a parity in the same direction as the transition of the head, then $h_i$ will halt and remain at its position for all futures times $t$.  When a head halts on account of an $a$ which it encounters, we will say that this is an $a$ with the `wrong' parity relative to that head.

There are several cases to consider.  The first case is that $h_i(t)$ encounters an $a$ with the `wrong' parity and halts before it reaches either sites $h_{i-1}(0)$ or $h_{i+1}(0)$.  This can happen in two ways; either because such an $a$ with the `wrong' parity was encoded in the initial conditions (and not laid down by a head), or because it was laid down by a head $h_{i-1}$ or $h_{i+1}$.  Either situation is clearly fine, since $h_i$ will also before reaching $h_{i-1}(0)$ or $h_{i+1}(0)$.

The second case is that $h_i(t)$ encounters either head $h_{i-1}$ or $h_{i+1}$ and thus halts before it reaches either sites $h_{i-1}(0)$ or $h_{i+1}(0)$.  This is fine as well.

The third and final case is that $h_i(t)$ eventually encounters either site $h_{i-1}(0)$ or site $h_{i+1}(0)$.  Let us consider the former, and note that the argument for the latter goes through mutatis mutandis.  In the former setting, $h_i$ must encounter $h_{i-1}(0)$ from the left.  If head $h_{i-1}$ is still present at site $h_{i-1}(0)$ then head $h_i$ will halt at site $h_{i-1}(0) + 1$ and thus $h_{i-1}(0) < h_i(t)$ for all $t$.  If head $h_{i-1}$ is not present at site $h_{i-1}(0)$ then it must be to the left of that site, for if it were to the right then it would have already encountered head $h_i$ causing both of them to halt.  But then $h_{i-1}$ would have moved left from site $h_{i-1}(0)$ and thus recorded an $a$ there which has the wrong parity relative to $h_i$; as such, if $h_i$ attempts to transition leftward into site $h_{i-1}(0)$, then $h_i$ will halt at site $h_{i-1}(0) + 1$, again giving $h_{i-1}(0) < h_i(t)$ for all $t$.  Nearly identical arguments likewise give us $h_i(t) < h_{i+1}(0)$ for all $t$, leading to the desired claim.
\end{proof}

With the above lemma at hand, we can now prove the main theorem of this section.
\begin{theorem}[Non-talking heads CA does not have self-replication]\label{thm:non-talking-no-selfrep}
The 1D non-talking heads CA does not have any configurations satisfying the local self-replicating condition. Thus the CA is in $\mathsf{LocallyUniversal}$ but not in $\mathsf{UniversalSelfReplicating}$, implying
\begin{align}
\mathsf{UniversalSelfReplicating} \subsetneq \mathsf{LocallyUniversal}\,.
\end{align}
\end{theorem}

\begin{proof}
By contradiction, suppose that $c$ is a configuration of the CA satisfying the local self‑replicating condition.  Since by hypothesis $c$ is a self‑replicating pattern, it agrees with an admissible background $B$ outside of a finite region.  Without loss of generality, we decompose $\mathbb{Z} = R_-\sqcup R_0 \sqcup R_+$ where $R_- = (-\infty,a-1] \cap \mathbb{Z}$, $R_0=[a,b] \cap \mathbb{Z}$, and $R_+ = [b+1,\infty) \cap \mathbb{Z}$ for some integers $a<b$.  We suppose that $c$ equals $B_1\big|_{R_-}$ on $R_-$, a finite pattern $Q$ on $R_0$, and $B_2\big|_{R_+}$ on $R_+$, where $B_1,B_2$ are basic backgrounds.

Let $\widetilde{T}_1$ be the time that it takes for $B_1$ to evolve to its steady state (which is guaranteed to exist by Lemma~\ref{lemma:periodic}), which is necessarily spatially periodic and time-periodic.  We similarly let $\widetilde{T}_2$ be the time that it takes for $B_2$ to evolve to its steady state (again guaranteed to exist by Lemma~\ref{lemma:periodic}), and take $\widetilde{T} := \max\{\widetilde{T}_1, \widetilde{T}_2\}$.  Further let $\min Q$ denote the location of the left boundary of $Q$, and let $\max Q$ denote the location of the right boundary of $Q$.  We denote
\begin{align}
S_t^- &:= (-\infty, \min Q - t - 1] \cap \mathbb{Z} \\
S_t^+ &:= [\max Q + t + 1, \infty) \cap \mathbb{Z}
\end{align}
which captures regions untouched by the forward lightcone of $Q$ at time $t$.  Then for all $t \geq \widetilde{T}$,
\begin{align}
F^t(c)\big|_{S_t^-} = F^t(B_1)\big|_{S_t^-}\,,\quad F^t(c)\big|_{S_t^+} = F^t(B_2)\big|_{S_t^+}\,,
\end{align}
reflect the steady state dynamics of $B_1$ and $B_2$, respectively.  Let us consider the right-most head in the half-infinite configuration $F^{\widetilde{T}}(c)\big|_{S_t^-}$, and let $a'-1$ be the right-most coordinate that it reaches in the steady-state dynamics of $B_1$.  Similarly, let us consider the left-most head in the half-infinite configuration $F^{\widetilde{T}}(c)\big|_{S_t^+}$, and let $b'+1$ be the left-most coordinate that it reaches in the steady-state dynamics of $B_2$.  It will be useful to consider the decomposition $\mathbb{Z} = R_-'\sqcup R_0' \sqcup R_+'$ where $R_-' = (-\infty,a'-1] \cap \mathbb{Z}$, $R_0'=[a',b'] \cap \mathbb{Z}$, and $R_+' = [b'+1,\infty) \cap \mathbb{Z}$.

Next we consider the dynamics of $F^t(B_1)\big|_{R_+'}$ for $t \geq \widetilde{T}$.  Since the dynamics in the region $R_+'$ and at times $t \geq \widetilde{T}$ is periodic in both space and time, we can decompose $R_+' = R_+'^{\,(1)} \sqcup R_+'^{\,(2)} \sqcup R_+'^{\,(3)} \sqcup \cdots$ where each $R_+'^{\,(i)}$ has the same length and evolves identically.  Moreover, we can choose the $R_+'^{\,(i)}$ so that each has a fixed number of heads which never leave that region.  As such, the $R_+'^{\,(i)}$ are \emph{non-interacting}: affecting the dynamics of any one $R_+'^{\,(i)}$ does not affect the dynamics of any others.  We can perform a similar decomposition $R_-' = R_-'^{\,(1)} \sqcup R_-'^{\,(2)} \sqcup R_-'^{\,(3)} \sqcup \cdots$ by an analogous argument.

For times $t \geq \widetilde{T}$, the dynamics in the region $R_0'$ can at most affect $R_-'^{\,(1)}$ and $R_+'^{\,(1)}$, causing some of their heads to halt.  But the regions $R_-'^{\,(i)}$ and $R_-'^{\,(i)}$ for all $i \geq 2$ will be unaffected, since they do not interact with $R_-'^{\,(1)}$ and $R_+'^{\,(1)}$, or with $R_0'$.  Therefore
\begin{align}
F^t(c)\big|_{R_-' \setminus R_{-}'^{\,(1)}} = F^t(B_1)\big|_{R_-' \setminus R_{-}'^{\,(1)}}\,,\quad F^t(c)\big|_{R_+' \setminus R_{+}'^{\,(1)}} = F^t(B_2)\big|_{R_+' \setminus R_{+}'^{\,(1)}}\,,\quad\text{for all }t \geq \widetilde{T}\,.
\end{align}
If $\bm{\mathcal P}=(\mP_0,\dots,\mP_{T-1})$ is the self-replicating pattern, letting $L$ be the length of the region on which the pattern is supported.  Then either a copy of $\bm{\mathcal P}$ is fully within $R_-' \setminus R_{-}'^{\,(1)}$ or $R_+' \setminus R_{+}'^{\,(1)}$, or otherwise it is fully within $[a'-L, b'+L]$.  In the former setting, $\bm{\mathcal P}$ is background-borne and excluded by Lemma~\ref{lemma:nobackgroundPs}, since infinitely many copies of it will arise at times $t \geq \widetilde{T}$.

Finally, we note that at any one time, there can only be finitely many disjoint copies of $\bm{\mathcal P}$ within the finite region $[a'-L, b'+L]$.  But this means that the total number of disjoint $\bm{\mathcal P}$'s is bounded for all time, and therefore cannot replicate beyond a certain number.  Thus we reach a contradiction.
\end{proof}

\subsection{Self-replication hierarchy for CAs}

We now have all the results in place to establish a self-replication hierarchy for CAs. By combining Theorem~\ref{thm:universalnotglobal} with Theorem~\ref{thm:non-talking-no-selfrep}, we have:
\begin{theorem}[Self-replication hierarchy for CAs]
We have the strict inclusions
\begin{align}
\mathsf{GloballyUniversal} \subsetneq \mathsf{UniversalSelfReplicating} \subsetneq \mathsf{LocallyUniversal}\,.
\end{align}
\end{theorem}
\noindent This is one of the main results of our paper. We emphasize that this stratification clarifies both the relationship between global and local universality and their connection to (universal) self-replication. The natural intuition from von Neumann's work on universal self-replicators~\cite{neumann_vnsr}, that any system capable of furnishing local Turing-universal computation is also capable of (perhaps universal) self-replication, is evidently false.  However, in establishing our counterexample, we have been forced to develop sharper mathematical criteria for when the dynamics of physical systems constitute non-trivial self-replication, and to identify dynamical mechanisms that are essential for such self-replication to arise.

\section{Open problems and conjectures}
\label{sec:openproblems}

We conclude by articulating a number of open problems and conjectures motivated by our work.  In so doing, we sketch the contours of a research program for establishing a general theoretical foundation for self-replication beyond particular examples.  We will first give some more concrete questions before delving into grander goals.

\paragraph*{Effect of dimensionality of CAs.} In our discussion of \textsf{UniversalSelfReplicating} we took the set of CAs to be a union over all dimensions. That is, letting $\textsf{UniversalSelfReplicating}_d$ be the set of $d$-dimensional CAs exhibiting universal self-replication, we have $\textsf{UniversalSelfReplicating} := \bigcup_{d = 1}^\infty \textsf{UniversalSelfReplicating}_d$. It is natural to ask questions about each $\textsf{UniversalSelfReplicating}_d$ individually.

Our example of a $d = 1$ version of von Neumann's universal self-replicator demonstrates that $\textsf{UniversalSelfReplicating}_1$ is non-empty. The CA in question is not reversible (i.e.~it is not in \textsf{Reversible}). However, we can establish the following result for higher dimensions:
\begin{corollary}[to Theorem~\ref{thm:1dvN}]
There are reversible CAs in any dimension $d \geq 2$ which are universally self-replicating, or equivalently $\textnormal{\textsf{UniversalSelfReplicating}}_d \cap \textnormal{\textsf{Reversible}}$ is non-empty for $d \geq 2$.
\end{corollary}
\noindent This follows from the well-known fact that a CA in $d = 1$ can be embedded in a reversible CA in any higher dimension. Thus, we can use our $d = 1$ version of von Neumann's construction to exhibit a reversible CA in $d \geq 2$ dimensions which is universally self-replicating.  For the 1D case, we conjecture:
\begin{conjecture}
There is a reversible CA in one dimension which is universally self-replicating, or equivalently $\textnormal{\textsf{UniversalSelfReplicating}}_1 \cap \textnormal{\textsf{Reversible}}$ is non-empty.
\end{conjecture}
The most natural way to establish this conjecture would be to build a reversible 1D universal self-replicator `from scratch'. However, one might be tempted to adapt higher-dimensional constructions. One route would be to construct a reversible 2D universal self-replicator, and then to show that it can both perform universal computation and self-replicate in a strip of bounded height (although this might be harder to achieve in the reversible setting than the irreversible setting). Then one could adapt the strategy of Theorem~\ref{thm:1d_universal_self_replicator} to `compress' the reversible 2D CA into a reversible 1D CA.

Another tempting approach would be to take our construction of a 1D universal self-replicator in Theorem~\ref{thm:1d_universal_self_replicator}, and try to `simulate' it by a reversible 1D CA.  However, this route appears to be difficult, or perhaps impossible, for interesting reasons.  To elaborate, we note that there are results in the literature that any irreversible 1D CA can be simulated by a reversible 1D CA~\cite{morita1995, durand2000reversible}.  Take for example the construction of~\cite{morita1995}, which shows how a reversible 1D CA can simulate any other irreversible 1D CA starting from a configuration in which all but finitely many sites are in the quiescent state.  However, the `simulator' reversible 1D CA takes $O(n)$ time steps to simulate any configuration in the irreversible CA with $n$ non-quiescent sites, amounting to a worst-case quadratic computational slowdown.  While such an encoding is viable for lifting a locally universal but irreversible 1D CA into a reversible one, it does not interface well with the local self-replicating condition. The problem is that when simulating self-replicating configurations, the simulation time for dynamical motifs increases with the number of replicator copies being simulated. From the perspective of the simulated system, the replicated organisms become 'slower' as their numbers grow. This means there is no truly preserved dynamical motif across time that is shared by all copies. (Moreover the particular simulation method in~\cite{morita1995} can cause dynamical motifs to occur non-simultaneously when many replicator copies are present, providing another violation of the local self-replicating condition.)  Other methods for simulating irreversible 1D CAs with reversible ones (see e.g.~\cite{durand2000reversible}) exhibit related computational slowdowns. In fact, no-go theorems such as~\cite{hertling1998embedding} suggest that non-constant slowdowns may be unavoidable when simulating irreversible 1D CAs with reversible ones.

Turning to a new question, we can also ask whether specific CAs are universally self-replicating. To this end, we recapitulate our conjecture from earlier in the paper:
\begin{conjecture}
Rule 110 is in $\textnormal{\textsf{UniversalSelfReplicating}}_1$.
\end{conjecture}
\noindent An affirmative answer to this conjecture would be compelling since it would constitute by far the simplest set of CA rules giving rise to (non-trivial) universal self-replication, although the encoding would be necessarily complicated. It would suffice for Rule 110 to be globally universal (as this would imply universal self-replication on account of our $d = 1$ version of von Neumann's construction), but this is not known. The existing universality construction~\cite{cook2009concrete} is tailored to local universality as opposed to global universality, so it would seem that a rather different encoding would be required if global universality were to be possible. Instead, we suspect it may be possible to use (a modest generalization of) the existing local universality construction for Rule 110 to build a 1D universal self-replicator.

\paragraph*{Genericity of universality and self-replication.} Our work serves to crystallize the relationship between computational universality and non-trivial self-replication. Ultimately, we are interested in being able to diagnose if a particular CA is able to exhibit non-trivial self-replication, or perhaps more strongly universal self-replication. A sufficient condition is for the CA to be globally universal; a necessary condition for the CA to exhibit universal self-replication is for it to be locally universal. At the moment, we do not know of any other more discriminating necessary conditions, and we will discuss this later.

Do `generic' CAs exhibit non-trivial self-replicators?  Before examining this question, we take a detour to discuss a more well-studied variant.  There is a belief within the CA community that `generic' CAs exhibit computational universality of some kind. A clear articulation of this view originates in the work of Wolfram, specifically in his work on 1D CAs~\cite{wolfram1984universality,wolfram2002}. He offers a qualitative classification of CA behavior roughly as follows: in Class 1 CAs, almost any initial condition quickly converges to a uniform fixed point; in Class 2 CAs, dynamics settle into simple, localized periodic or stable patterns; in Class 3 CAs, activity remains aperiodic and chaotic with no persistent structure; and in Class 4 CAs, long-lived localized patterns move and interact, mixing order with apparent randomness.  Wolfram's ``Principle of Computational Equivalence'' conjecture stipulates that Class 4 CAs are typically computationally universal~\cite{wolfram1984universality,wolfram2002}.  While conceptually stimulating, this conjecture is necessarily qualitative, as neither the Class 4 designation, computational universality, nor the notion of `typically' are sharply defined. In practice, the conjecture is applied by identifying CAs that qualitatively belong to Class 4 and then conjecturing their computational universality, as was done with Rule 110.

We suggest that Wolfram's conjecture most naturally pertains to local universality rather than global universality. Indeed, Rule 110 is a Class 4 CA whose local universality is taken as evidence for the conjecture. More broadly, one might expect any locally universal CA to exhibit the dynamical features characteristic of Class 4.  However, if Wolfram’s conjecture is read as requiring that random initial conditions typically evolve into ordered, localized structures, then reversible CAs pose a principled obstacle. Because every reversible CA is surjective and (on full shifts) surjective CAs preserve the uniform Bernoulli product measure, an i.i.d.~initial configuration remains i.i.d.~at every fixed time and so no spatial correlations or entropy reduction are possible~\cite{kari2015statistical, capobianco2013surjective}. A complementary result shows that in one dimension, surjectivity is equivalent to preserving Martin–L\"{o}f randomness, so algorithmically random configurations remain random under reversible evolution~\cite{calude2001randomness}. In this precise sense, reversible CAs do not generate order from randomness and therefore, under this reading, should not be counted as Class 4. (They can still display Class‑4‑like behavior from finite or low‑entropy seeds, e.g.~\cite{margolus2018crystalline}.)  These considerations lead us to conjecture:
\begin{conjecture}
There exists a Class 4 CA which is irreversible and locally universal, but not globally universal.
\end{conjecture}
\noindent If true, this would support our suggestion that local universality is the natural interpretation of universality in Wolfram's conjecture.  A difficulty in establishing the conjecture is that while demonstrating local universality requires finding a particular sensible encoding and decoding, ruling out global universality requires showing that no suitable encoding and decoding exist.

We now examine another aspect of Wolfram's conjecture: whether `typical' Class 4 CAs are computationally universal. We reformulate this question to make it precise. Consider CAs on a square lattice in a fixed number of dimensions, say $d = 1$, $2$, or $3$. Consider CAs with interaction radius at most $R$, where each cell has $|S|$ possible states. We denote this set $\textsf{CA}_{d,R,|S|}$, which has cardinality approximately $|S|^{\text{const.} \times R^d \times |S|}$. This leads to the following question:
\begin{question}
\label{Q:CAasymptotics}
Consider the set $\textnormal{\textsf{CA}}_{d,R,|S|}$. For $d$ fixed, $R$ fixed, and $|S| \to \infty$, what fraction of the set is locally universal? What about for $d$ fixed, $|S|$ fixed, and $R \to \infty$? What about for $d$ fixed, and $|S|, R \to \infty$ with some ratio $f(R)/|S|$ fixed?
\end{question}
\noindent Plausibly, only a small or vanishing fraction of CAs are locally universal in the first two cases. Since determining local universality is undecidable~\cite{intrinsic_universality_problem_of_1d_cas}, a proof would need to show that generic CAs in these regimes possess obstructions to universality. This strategy has succeeded in the study of random groups, where all but a vanishing fraction have solvable word problems due to their hyperbolic geometric structure~\cite{gromov1993geometric}. Semon Rezchikov has suggested generalizing this approach to semigroups or Turing machines; extending it to CAs would be natural.

If only a vanishing fraction of CAs (in the first two formulations of Question~\ref{Q:CAasymptotics}) were locally universal, this would contradict the intuition behind Wolfram's conjecture if Class 4 CAs were asymptotically generic. Regardless, resolving Question~\ref{Q:CAasymptotics} would provide valuable insight into the relationship between local universality and the space of CAs.

For non-trivial self-replication, the situation may be more favorable. Exhibiting non-trivial self-replication is likely weaker than local universality, and thus potentially more pervasive among CAs. As we discuss below, there may be a hierarchy of self-replication conditions that we do not yet understand; contingent on such conditions being well-defined, a variant of Question~\ref{Q:CAasymptotics} for non-trivial self-replication might optimistically yield that a large fraction of CAs support it. We pose this as a preliminary question:
\begin{question}
\label{Q:CAasymptotics2}
Do a large fraction of CAs (e.g.~in some asymptotic sense of Question~\ref{Q:CAasymptotics}) exhibit non-trivial self-replication?
\end{question}
\noindent An affirmative answer would be particularly interesting, as it would suggest that non-trivial self-replication is in some sense generic.

\paragraph*{Hierarchy of non-trivial self-replicating conditions.}

One of the contributions of our work is to articulate a necessary condition for a CA configuration to exhibit non-trivial self-replication.  An essential intuition for our condition is that a self-replicating object and its replicas should be able to exhibit non-trivial dynamics (such as a periodic orbit) after replication has occurred.  That is, the replicated objects are in some manner functional.  It is an interesting problem to formally characterize a hierarchy of ever more stringent non-trivial self-replication conditions.  We formulate this as a question:
\begin{question}
Can one define a nested sequence of dynamical conditions on CAs leading to a hierarchy
\begin{align}
\textnormal{\textsf{LocallyPeriodicSelfReplicating}} = \textnormal{\textsf{C}}_1 \subset \textnormal{\textsf{C}}_2 \subset \cdots \subset \textnormal{\textsf{C}}_n = \textnormal{\textsf{UniversalSelfReplicating}},
\end{align}
analogous to increasing levels of biological complexity (e.g.~from virus-like copying to autonomous cell division)?
\end{question}
\noindent If such a hierarchy can be constructed, it would be fruitful to classify known CAs within the hierarchy.

While it is natural to organize the space of CAs according to whether they can furnish configurations whose evolutions exhibit certain dynamical properties, we could equally well classify the CA configurations themselves.  For instance, given a CA rule $G$ and a configuration $c$, consider the pair $(G,c)$.  So far we have been saying that $G$ is e.g.~in \textnormal{\textsf{LocallyPeriodicSelfReplicating}} if there is at least one $c$ whose dynamics under $G$ satisfies the local self-replicating condition.  (The viable initial $c$'s may be restricted, such as by demanding that they vary in only a finite number of cells from a fixed background, etc.)  Instead, we could consider $\textsf{LocallyPeriodicSelfReplicating}_{\textsf{config}}$ to be the set of pairs $(G,c)$ such that the dynamics of $c$ under $G$ satisfies the local self-replicating condition.  This formulation opens up new questions.  For example:
\begin{question}
Let $G_{\text{\rm Byl}}$ be the CA rule giving rise to Byl's loops~\cite{byl1989self}, and let $b$ be its configuration of all quiescent states.  As previously discussed, $G_{\text{\rm Byl}} \in \textnormal{\textsf{LocallyPeriodicSelfReplicating}}$.  What is the best lower bound on
\begin{align}
N_{\text{\rm Byl}} = \min\!\left\{\textnormal{\textsf{Hamming}}(b,c)\, : \, (G_{\text{\rm Byl}},c) \in \textnormal{\textsf{LocallyPeriodicSelfReplicating}}_{\textnormal{\textsf{config}}}\right\},
\end{align}
namely the number of non-quiescent states for the simplest initial configuration of $G_{\text{\rm Byl}}$ whose dynamics satisfies the local self-replicating condition?  Byl's construction gives $N_{\text{\rm Byl}} \leq 12$.
\end{question}
\noindent This question pertains to a notion of \textit{complexity} for self-replicating configurations, particularly the minimum complexity necessary to achieve a specific type of self-replicating behavior.

As remarked earlier, Langton's loops~\cite{langton1984self} and their successive generalizations to Byl's loops~\cite{byl1989self}, etc., culminating in the Perrier-Sipper-Zahnd loops~\cite{perrier1996toward}, qualitatively exhibit a hierarchy of non-trivial self-replicating behaviors. Specifically, Langton's loops belong to $\textsf{LocallyPeriodicSelfReplicating}_{\textsf{config}}$, but the replicas do not exhibit more sophisticated properties beyond periodic orbits, and thus ostensibly do not belong to $\textsf{C}_{2,\,\textsf{config}}$. At the other end, the Perrier-Sipper-Zahnd loops belong to $\textsf{UniversalSelfReplicating}_{\textsf{config}}$. It would therefore be natural to distill a hierarchy of dynamical conditions from these different kinds of self-replicating loops to furnish descriptions of the sets $\textsf{C}_i$.

We can further set our sights beyond self-replication to ask for necessary conditions on CAs that furnish configurations exhibiting sexual reproduction. This is more challenging since offspring can differ substantially from their parents, making it difficult to quantify family resemblance without presupposing specific organismal structure. We thus pose the following question:
\begin{question}
Are there natural necessary conditions on CA configurations such that they exhibit sexual reproduction?
\end{question}
\noindent Addressing this question, and more broadly identifying dynamical mechanisms for sexual reproduction, would be both mathematically interesting and theoretically important, with consequences for mathematical biology and the philosophy of biology.

\paragraph*{Robustness to noise, mutations, and evolution.}

Our explorations of a theoretical foundation for self-replicating systems have focused on the cleanest setting: dynamics in the absence of environmental noise. A comprehensive theory of self-replication must explain how the mechanisms involved are robust to noise. Moreover, introducing noise raises the possibility of mutations during self-replication, which complicates the identification of `replicas' since they may no longer be exact copies. This is a more modest version of the problem of abstractly defining sexual reproduction that we mentioned above.

Several noise models for CAs merit consideration. In probabilistic CAs, where the rules themselves are subject to error, robust self-replication would require some form of error correction. However, a crucial distinction must be made: the error correction should be performed by the organism itself, not by the global CA configuration. While it is possible to globally simulate a noiseless CA within a noisy one using error correction schemes, this approach would be inappropriate for studying robust self-replication. For instance, simply embedding von Neumann's universal self-replicator in a globally error-corrected noisy CA would protect the entire configuration rather than demonstrating that the organism itself can handle errors. The organism, which occupies only part of the configuration, should be responsible for its own error correction. Formalizing this distinction between organism-level and global-level error correction would be valuable for understanding robust self-replication.

A different error model fits within our existing paradigm of deterministic CAs. Consider a CA configuration whose dynamics exhibit non-trivial self-replication. We can ask whether this self-replicating behavior remains robust under perturbations of the initial state, both within the organism itself and in its surrounding environment. After all, organisms in nature do not exist in a void; they must contend with external objects and environmental interactions that may interfere with their replication processes.

These considerations of robustness and mutation naturally lead to von Neumann's grander vision for his universal self-replicator construction~\cite{neumann_vnsr, vonneumann1963general}. He viewed self-replication not as an end in itself, but as an intermediate step toward achieving an artificially evolving system where organisms would compete to write their own descriptions on the universal constructor's tape, thereby securing their replication. This competition inherently requires the robustness we have been discussing: organisms must not only replicate successfully but do so despite environmental interference and in competition with other replicating entities. The talking-heads model offers a potentially interesting framework for exploring such evolutionary dynamics. If the heads that enable self-replication are both robust to perturbations and limited in number, we might observe emergent competition for these scarce replication resources. Observing the dynamics of such a system could reveal how competition for limited replication machinery drives evolutionary processes in simulated models of life.  This motivates the question:
\begin{question}
What are minimal conditions on a CA for evolutionary dynamics to emerge? Specifically, can we characterize CAs where limited resources (such as talking heads) combined with mutation-prone replication lead to competitive exclusion, adaptation, or other hallmarks of evolution?
\end{question}
\noindent This question may initially be most readily addressed empirically, via simulations that may lead to open-ended evolution.

\paragraph*{Kinetic models.}

In his original approach to the theoretical study of self-replication, von Neumann considered `kinetic models'~\cite{neumann_vnsr, vonneumann1963general, burks1970essays}, namely particles in continuous space with interactions, as the dynamics to furnish a self-replicating system. This proved particularly challenging, and his friend Stanislaw Ulam suggested cellular automata as an abstraction of kinetic models, which von Neumann ultimately pursued. We remark that the precise sense in which a CA is an `abstraction' of a kinetic model is subtle; one might say that a CA faithfully abstracting a kinetic model captures the latter's salient degrees of freedom, but this might make the CA somewhat non-local, with dynamics on a graph rather than $\mathbb{Z}^d$. At the very least, CAs provide natural, readily simulable, visually suggestive models for investigating self-replicating systems.

Today, simulating kinetic models (described by systems of ODEs or PDEs) is much easier than in von Neumann's time. Consequently, exploring self-replicating systems within such models has become more feasible, as one can more readily test ideas. Moreover, techniques like neural ODEs enable backpropagation through ODE dynamics, allowing optimization problems that find configurations with specified properties.

A concrete setting where these ideas have been explored is Lenia~\cite{chan2018lenia, chan2020lenia, chan2023towards} and its variants (see e.g.~\cite{plantec2023flow, mordvintsev2022particlelenia, papadopoulos2025mace}). Lenia is a continuous-space version of a CA that evolves according to a continuum kernel. For various kernel parameters, Lenia exhibits visually compelling dynamics with objects that move, appear to copy themselves, and interact in complex ways. While it may be tempting to regard these interactive objects as artificial models of single-cellular life, we argue that this analogy may not be appropriate for several reasons. First, these objects are approximately the size of the continuum kernel's width, placing them at the smallest length scale of available interactions. Thus, they have the character of fundamental particles in Lenia's dynamics rather than composite structures. Second, in parameter regimes where these objects can `copy' by splitting, this does not constitute non-trivial self-replication; the dynamical mechanism only produces copies of a single object type and cannot produce copies of different objects. This is analogous to our earlier distinction between a copy machine that can copy any document versus a printer that only reproduces a single loaded file.

This reasoning demonstrates that the underlying logic of our definitions for non-trivial self-replication in CAs generalizes to kinetic models. Indeed, one could provide precise generalizations of our definitions (and many of our results) for various classes of kinetic models. Thus our work on non-trivial self-replication in CAs yields conceptual insights that organize our understanding of non-trivial self-replication across various model types beyond CAs. To this end, it is natural to ask the following question:
\begin{question}
Can one construct a kernel in Lenia (or one of its natural variants) such that there are configurations with dynamics satisfying a natural analog of the local self-replicating condition?
\end{question}
\noindent We also venture the following conjecture:
\begin{conjecture}
It is possible to build a universal self-replicator within Lenia (or one of its natural variants).
\end{conjecture}
\noindent Addressing the question and establishing (or refuting) the conjecture would be valuable both for understanding Lenia specifically and kinetic models more generally.

\paragraph*{Efficient identification of non-trivial self-replication.}

In studying the dynamics of CAs or kinetic models, we face the practical question of how to observe a simulation and determine whether it exhibits non-trivial self-replication.  Ultimately, we would like to identify which CA rules and configurations give rise to non-trivial self-replication empirically, making a practical identification procedure highly valuable.

One immediate difficulty with identifying configurations satisfying the local self-replicating condition is that the condition concerns infinite time-horizon dynamics. It requires an unbounded number of organism copies to be produced and exist simultaneously, necessitating both arbitrarily large times and volumes. The condition examines the outputs of self-replication and verifies that they are dynamically non-trivial and growing in number, but it does not directly address the mechanism of self-replication itself. We do not yet know how to abstract such mechanisms in a way that remains agnostic to the full scope of possible self-reproduction strategies. If such mechanisms could be classified, they might be recognizable in local spatial patches over finite time intervals.\footnote{Related work in the philosophy of biology by Rosen~\cite{rosen1991life} argues that organisms cannot be fully captured by formal computational models. His argument centers on organisms being closed to efficient causation (i.e.~every component that produces effects within the organism is itself produced by other components within the organism) and containing internal predictive models of their environment and their own interactions with it. This leads Rosen to argue for the presence of impredicative causal loops that, in a manner analogous to G\"{o}del, cannot be formalized without losing essential properties of life. A consequence would be that identifying life from finite observations of a physical system is impossible. While the soundness of these definitions and arguments has been debated, we note that any simulation of an organism should be intelligible at finite resolution, in finite volume, and over finite time intervals (only idealized organisms exist forever). In such finitary settings, direct appeals to G\"{o}del's incompleteness theorems do not apply. Nevertheless, Rosen's concepts of causal closure and anticipatory systems~\cite{rosen2011anticipatory} offer an interesting structural perspective on life that could interface with the ideas in this paper.} This leads to the following question:
\begin{question}
Is it possible to classify or partially characterize dynamical mechanisms for the act of non-trivial self-replication?
\end{question}
\noindent This question relates to our earlier discussion about whether self-replicating organisms require abstracted versions of `head' and `tape' degrees of freedom analogous to those in a Turing machine. While it remains unclear whether such degrees of freedom are strictly necessary (and articulating substantively distinct alternatives is challenging), it would be valuable to have an algorithm that identifies head and tape degrees of freedom from observed dynamics.

Recent work~\cite{kumar2024automating} has explored using large language models (LLMs)~\cite{vaswani2017attention, kaplan2020scaling} to identify interesting `life-like' organisms in open-ended simulations. These algorithms attempt to automatically identify which simulations, among many with different initial conditions or evolution rules, give rise to life-like forms. This approach leverages our ability to semantically articulate properties of living organisms that we would recognize when we see them, automating this semantic recognition through LLMs that can both process semantically rich instructions and interpret images and videos. It would therefore be useful to develop ordinary language versions of our criteria for non-trivial self-replication that could serve as instructions for LLMs to identify examples in simulations.

\paragraph*{Probabilistic bounds for emergence of self-replication.}

A theory of self-replication should provide probabilistic bounds on whether randomized initial conditions within a fixed CA, or across a probabilistic ensemble of CAs, give rise to living organisms (within a time interval $[0,T]$ for large $T$). To make this precise, we venture the following conjecture:
\begin{conjecture}
Suppose we have a non-trivial necessary condition for non-trivial self-replication in finite area $A$. Consider the Game of Life CA with periodic boundary conditions on a square of area $A$. Then there exists a sufficiently large area $A_0$ such that for all $A \geq A_0$, there is a $\delta > 0$ such that at least a $\delta$-fraction of all initial conditions exhibit non-trivial self-replication at some point in their dynamics.
\end{conjecture}
\noindent It would be interesting to understand which CA rules maximize the probability of non-trivial self-replication emerging.

A more ambitious goal would be to consider chemical systems present during early Earth's development and establish lower bounds on the probability that they give rise to non-trivial self-replication and ultimately complex life. Addressing this would help us understand how rare or common life is, both on Earth and potentially throughout the universe.
\\ \\
\indent The questions and conjectures above collectively work toward a general theory of self-replication that captures its essential features across different dynamical systems. From understanding the minimal computational requirements to predicting spontaneous emergence, from formalizing evolutionary dynamics to developing practical identification methods, these problems span the theoretical and empirical domains necessary for a comprehensive understanding of self-replication. Progress on this program would fulfill von Neumann's original vision of understanding self-replication as a fundamental phenomenon, independent of its specific physical or biological implementation.

\bibliographystyle{unsrt}
\bibliography{ref}

\end{document}